\documentclass[10pt]{article}

\usepackage[margin=1in]{geometry}
\usepackage{amsthm,amsmath,amsfonts,amssymb}
\usepackage{graphicx}

\newcommand\numberthis{\addtocounter{equation}{1}\tag{\theequation}}
\newcommand{\eqtext}[1]{\ \mathrm{#1} \ }
\newcommand{\myspace}{\\[0.2in]}

\newcommand{\di}[1]{\, \mathrm{d}#1}

\newcommand{\dd}{\mathrm{d}}
\newcommand{\ee}{\mathrm{e}}
\newcommand{\ii}{\mathrm{i}}
\newcommand{\ve}{\varepsilon}
\newcommand{\fhat}{\widehat{f}}
\newcommand{\Vhat}{\widehat{V}}

\DeclareMathOperator{\real}{Re}
\DeclareMathOperator{\imag}{Im}
\DeclareMathOperator*{\argmin}{argmin}
\DeclareMathOperator{\sinc}{sinc}

\newtheorem{theorem}{Theorem}
\newtheorem{lemma}{Lemma}

\newtheorem{remark}{Remark}
\newtheorem{conjecture}{Conjecture}

\numberwithin{equation}{section}

\title{Anomalous localized resonance phenomena in the nonmagnetic,
finite-frequency regime}
\author{Daniel Onofrei\footnotemark[1]%
\and Andrew E.\ Thaler\footnotemark[2]}

\usepackage{hyperref}

\begin{document}
\maketitle

\renewcommand{\thefootnote}{\fnsymbol{footnote}}

\footnotetext[1]{Department of Mathematics, University of Houston,
Houston, TX 77204}
\footnotetext[2]{Institute for Mathematics and its Applications,
    University of Minnesota, College of Science and Engineering,
    Minneapolis, MN 55455.  The research of AET was supported in part by
    the Institute for Mathematics and its Applications with funds
    provided by the National Science Foundation through NSF Award
0931945 and in part by AFOSR under the 2013 YIP Award FA9550-13-1-0078.}


\begin{abstract}

The phenomenon of anomalous localized resonance (ALR) is observed at the
interface between materials with positive and negative material
parameters and is characterized by the fact that, when a given source is
placed near the interface, the electric and magnetic fields start to
have very fast and large oscillations around the interface as the
absorption in the materials becomes very small while they remain smooth
and regular away from the interface. 

In this paper, we discuss the phenomenon of anomalous localized
resonance (ALR) in the context of an infinite slab of homogeneous,
nonmagnetic material ($\mu=1$) with permittivity $\epsilon_s=-1-\ii\delta$ for
some small loss $\delta \ll 1$ surrounded by positive, nonmagnetic,
homogeneous media. We explicitly characterize the limit value of the
product between frequency and the width of slab beyond which the ALR
phenomenon does not occur and analyze the situation when the phenomenon
is observed. In addition, we also construct sources for which the ALR
phenomenon never appears. 
\end{abstract}

\section{Introduction}

In the following, we discuss the anomalous localized resonance phenomenon
(ALR) appearing at the interface between materials with positive and
negative material parameters in the finite frequency regime.  We
consider the particular slab geometry described by (see
Figure~\ref{fig:slab})
\begin{equation}\label{eqn:CSM}
    \mathcal{C} \equiv \{(x,y) \in \mathbb{R}^2 : x < 0\}; \quad
    \mathcal{S} \equiv \{(x,y) \in \mathbb{R}^2 : 0 < x < a\}; \quad
    \mathcal{M} \equiv \{(x,y) \in \mathbb{R}^2 : x > a\};
\end{equation}
where $a>0$ denotes the width of the slab and the sets $\mathcal{C}$,
$\mathcal{S}$, $\mathcal{M}$ represent the regions to the left of the
slab, within the slab, and to the right of the slab, respectively.
We also define
\[
    d_0 \equiv \min \{x : (x,y) \in \mathrm{supp} f\} \eqtext{and}
    d_1 \equiv \max \{x : (x,y) \in \mathrm{supp} f\}.
\]
\begin{figure}[!hbpp]
    \centering
    \includegraphics[width=0.75\textwidth]{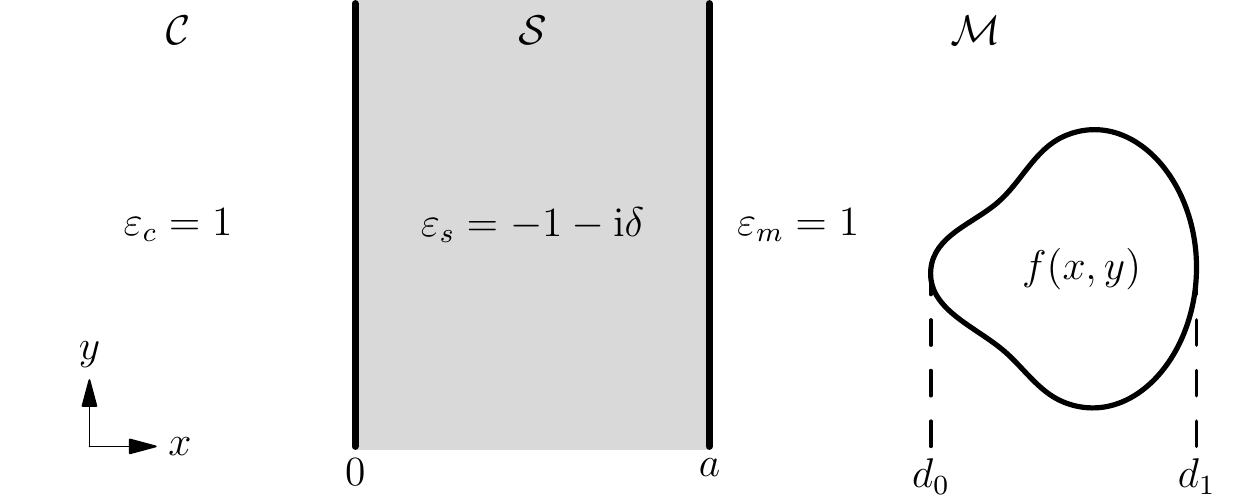}
    \caption{\emph{In this figure, we illustrate the geometry of the
    problem we consider in this paper.}}
    \label{fig:slab}
\end{figure}
In this geometry, we assume that all materials are homogeneous and
nonmagnetic (i.e., with magnetic permeability $\mu=1$); the electrical
permittivity is given by
\begin{equation}\label{1}
     \ve \equiv
     \begin{cases}
         1 &\text{for } x < 0, \\
         -1 - \ii\delta &\text{for } 0 < x < a, \\
         1 &\text{for } x > a 
     \end{cases}
 \end{equation}
for some $\delta \in (0,1)$. We consider the following partial
differential equation (PDE) in 2D:
\begin{equation}\label{eqn:finite_freq}
    \nabla\cdot\left(\frac{1}{\ve}\nabla V\right) + (k_0^2+i\zeta) V = 
    -f \quad \text{in } \mathbb{R}^2,
\end{equation}
where $\zeta\geq0$, $k_0>0$, $f \in L^2(\mathcal{M})$ with compact
support in $\mathcal{M}$, and $\ve$ is given in \eqref{1} (see
\S~\ref{subsec:Maxwell_to_Helmholtz} for a derivation of
\eqref{eqn:finite_freq} from the Maxwell equations).

For convenience, we define 
\begin{equation}\label{eqn:Vcsm_def}
    V_c \equiv V|_{\mathcal{C}}; \quad 
    V_s \equiv V|_{\mathcal{S}}; \quad
    V_m \equiv V|_{\mathcal{M}}.
\end{equation}
We assume the solution $V$ also satisfies the following continuity
conditions across the boundaries at $x = 0$ and $x = a$ for almost every
$y \in \mathbb{R}$:
\begin{equation}\label{eqn:continuity_conditions}
    \begin{cases}
        V_c(0,y) = V_s(0,y); \quad &
        \dfrac{\partial V_c}{\partial x}(0,y) = 
        \dfrac{1}{-1-\ii\delta}\dfrac{\partial V_s}{\partial x}(0,y), 
        \myspace{}
        V_s(a,y) = V_m(a,y); \quad &
        \dfrac{1}{-1-\ii\delta}\dfrac{\partial V_s}{\partial x}(a,y) = 
        \dfrac{\partial V_m}{\partial x}(a,y).  
     \end{cases}
\end{equation}
In what follows we assume that the parameters and data are such
that problem \eqref{eqn:finite_freq}, \eqref{eqn:continuity_conditions}
admits a unique solution $V\in L^2_{loc}(\mathbb{R}^2)$ with
$V(x,\cdot)\in H^{1}(\mathbb{R})$ and $\frac{\partial V}{\partial
x}(x,\cdot)\in  L^{2}(\mathbb{R})$ for almost every $x\in \mathbb{R}$. 
\begin{remark}
    Note that in the case when $\zeta>0$, the unique solution of the
    problem will have the property that $V(x,y)\rightarrow 0$ as
    $|x|\rightarrow\infty$ for almost every $y\in\mathbb{R}$; for
    $\zeta\ll 1$, this solution will be well approximated by the
    solution in the case $\zeta=0$. 
\end{remark}
We say anomalous localized resonance (ALR) occurs if the following two
properties hold as $\delta \rightarrow 0^+$ \cite{Milton:2005:PSQ}:
\begin{enumerate}
    \item $|V| \rightarrow \infty$ in certain localized regions with
          boundaries that are not defined by discontinuities in the
          relative permittivity and 
    \item $V$ approaches a smooth limit outside these localized regions. 
\end{enumerate}    
In \cite{Milton:2005:PSQ}, Milton, Nicorovici, McPhedran, and Podolskiy
showed that if $f$ is a dipole and $\varepsilon_c = \varepsilon_m = 1$,
then ALR occurs if $a < d_0 < 2a$, where $d_0$ is the location of the
dipole. In this case there are two locally resonant strips --- one
centered on each face of the slab.  As the loss parameter (represented
by $\delta$) tends to zero, the potential diverges and oscillates wildly
in these resonant regions.  Outside these regions the potential
converges to a smooth function.  Also, if the source is far enough away
from the slab, i.e., if $d_0 > 2a$, then there is no resonance and again
the potential converges to a smooth function.
    
Applications of ALR to superlensing were first discussed by Nicorovici,
McPhedran, and Milton in \cite{Nicorovici:1994:ODP} and were analyzed in
more depth in \cite{Milton:2005:PSQ} (see also the works by Yan, Yan,
and Qiu \cite{Yan:2008:CSC}, Bergman \cite{Bergman:2014:PIP}, Nguyen
\cite{Nguyen:2013:SUC}, Pendry \cite{Pendry:2000:NRM}, and Pendry and
Ramakrishna \cite{Pendry:2002:NFL} for a description of superlensing
phenomena).
    
Applications of ALR to cloaking in the quasistatic regime were first
analyzed Milton and Nicorovici \cite{Milton:2006:CEA}; they showed that
if $\varepsilon_c = \varepsilon_m = 1$ and a fixed field is applied to
the system (e.g., a uniform field at infinity), then a polarizable
dipole located in the region $a < d_0 < 3a/2$ causes anomalous localized
resonance and is cloaked in the limit $\delta \rightarrow 0^+$. Cloaking
due to anomalous localized resonance (CALR) in the quasistatic regime
was further discussed in \cite{Ammari:2013:ALR, Kohn:2012:VPC,
Ammari:2013:STN, Ammari:2013:STNII, Bouchitte:2010:CSO, Bruno:2007:SCS,
Nicorovici:2009:CPR, Xiao:2012:TEC, Nguyen:2015:CVA2}. CALR in the
long-time limit regime was discussed in \cite{Milton:2007:OPL,
Xiao:2012:TEC} (see also \cite{Yaghjian:2006:PWS}).
    
In \cite{Nicorovici:2008:FWC}, Nicorovici, McPhedran, Enoch, and Tayeb
studied CALR for the circular cylindrical superlens in the
finite-frequency case; they showed that for small values of $\delta$
the cloaking device (the superlens) can effectively cloak a tiny
cylindrical inclusion located within the cloaking region but that the
superlens does not necessarily cloak itself --- they deemed this
phenomenon the ``ostrich effect.'' The finite-frequency case was further
discussed by Kettunen, Lassas, and Ola \cite{Kettunen:2015:OAE} and
Nguyen \cite{Nguyen:2015:CVA}.

In the present report we prove, analytically and numerically, the
existence of a limit value $\gamma_*$, such that for $k_0$ with
$k_0a>\gamma_*$, ALR does not occur regardless of the position of the
source with respect to the slab interface. Under suitable conditions on
the source, we present numerical evidence for the occurrence of ALR in
the regime $k_0a<\gamma_*$ when the source is close enough to the
material interface, and we discuss some characteristics of the
phenomenon in this frequency regime as well. In the end we present two
examples of sources $f$ which do not generate ALR regardless of the
frequency regime and their relative position with respect to the
material interface.

The paper is organized as follows: in \S~\ref{Fourier} we present
highlights of the derivation of the unique solution in the Fourier
domain while in \S~\ref{Energy} we describe the energy around the right
interface of the slab. In \S~\ref{sec:large_gamma}, we show the absence
of ALR phenomena for large enough values of $k_0a$ while in
\S~\ref{subsec:shielding} we present an interesting side effect of the
nonmagnetic case, namely the shielding effect of the slab which behaves
as an almost perfect reflector. Next, for suitable conditions on the
source, in \S~\ref{subsec:pd_small_gamma} we present numerical evidence
for the ALR phenomenon in the case of small enough values of $k_0a$. In
\S~\ref{subsec:ALR_busting}, we construct two examples of possible
sources for which there is no ALR phenomenon regardless of the range of
$k_0a$ or the relative position of the source with respect to the slab
interface. The Appendix contains the technical proofs and derivations
which where not included in the main text. 
        
%
%
%



\subsection{Solution in Fourier domain}

\label{Fourier}

Due to our wellposedness assumption it follows that our problem will
admit a unique solution after applying the Fourier transform with
respect to the $ y $ variable. Recall that, for a given function
$h(x,\cdot) \in L^2(\mathbb{R})$ for some $x \in \mathbb{R}$, the
Fourier transform of $h$ with respect to $y$ is
\begin{equation}\label{eqn:Fourier_def}
    \widehat{h}(x,q) = \int_{-\infty}^{\infty} h(x,y)\ee^{-\ii q y}
    \di{y}.
\end{equation}

We will study the Fourier domain solution in each of the relevant sub-domains defined in \eqref{eqn:CSM}.


\subsubsection{The solution in $\mathcal{C}$}

In the region $\mathcal{C}$, the relevant equation is
\begin{equation}\label{eqn:core_equation}
    \frac{\partial^2 V_c}{\partial x^2} + 
    \frac{\partial^2 V_c}{\partial y^2} + k_0^2 V_c = 0.
\end{equation}
Taking the Fourier transform of \eqref{eqn:core_equation} with respect
to $y$, we find that $\Vhat_c(x,q)$ satisfies
\begin{equation}\label{eqn:core_equation_hat}
    \frac{\partial^2 \Vhat_c}{\partial x^2} - k_0^2\nu_c^2\Vhat_c = 0,
    \eqtext{where}
    \nu_c^2 \equiv \frac{q^2}{k_0^2}-1.
\end{equation}
\begin{remark}\label{rem:principal_root}
    Here and throughout the paper, we take the principal square root of
    complex numbers; that is, for a complex number $z = z' + \ii z'' =
    |z|\ee^{\ii\theta}$ where $\theta \in (-\pi,\pi]$, we take
    \[
        \sqrt{z} = |z|^{1/2}\ee^{\ii\theta/2},
    \]
    where $\theta/2 \in (-\pi/2,\pi/2]$.  In particular, this implies
    $\real\sqrt{z} \ge 0$.
\end{remark}
Remark~\ref{rem:principal_root} implies
\begin{equation}\label{eqn:nuc}
    \nu_c = 
    \begin{cases}
        \ii\sqrt{|q^2/k_0^2 - 1|} &\text{if } q^2/k_0^2 < 1, \\
        \sqrt{q^2/k_0^2 - 1} &\text{if } q^2/k_0^2 \ge 1.
    \end{cases}
\end{equation}
Then the general solution to \eqref{eqn:core_equation_hat} is
\begin{equation}\label{eqn:Vc_hat_gen}
    \Vhat_c(x,q) = A_q\ee^{k_0\nu_c x} + B_q\ee^{-k_0\nu_c x}
\end{equation}
for coefficients $A_q$ and $B_q$ that are independent of $x$.  

If $q^2/k_0^2 < 1$, then $\nu_c$ is purely imaginary.  Because $V_c$
should be outgoing (i.e., leftgoing) as $x\rightarrow -\infty$ and we
are considering $\ee^{\ii\omega t}$ time dependence (see
\S~\ref{subsec:Maxwell_to_Helmholtz}), we should have
\[
    \Vhat_c \sim \ee^{\ii k_0\sqrt{|q^2/k_0^2-1|}x} \eqtext{as} x
    \rightarrow -\infty.
\]
From \eqref{eqn:nuc} and \eqref{eqn:Vc_hat_gen}, we see that we can
ensure this by taking $B_q = 0$.

On the other hand, if $q^2/k_0^2 > 1$, then $\nu_c > 0$.  Thus we take
$B_q = 0$ in this case to ensure that $\Vhat_c(x,q) \rightarrow 0$ as $x
\rightarrow \infty$.  Finally, without loss of generality we may also
take $B_q = 0$ for $q^2/k_0^2 = 1$.  Therefore, 
\begin{equation}\label{eqn:Vc_hat}
    \Vhat_c(x,q) = A_q\ee^{k_0\nu_c x}.
\end{equation}


\subsubsection{The solution in $\mathcal{S}$}

In the region $\mathcal{S}$, the Fourier transform of $V_s$ satisfies
\begin{equation}\label{eqn:shell_equation_hat}
    \frac{\partial^2 \Vhat_s}{\partial x^2} - k_0^2\nu_s^2\Vhat_s = 0,
    \eqtext{where}
    \nu_s^2 \equiv \left(\frac{q^2}{k_0^2} + 1\right) + \ii\delta.
\end{equation}
The general solution is
\[
    \Vhat_s(x,q) = C_q\ee^{k_0\nu_s x} + D_q\ee^{-k_0\nu_s x};
\]
the coefficients $C_q$ and $D_q$ may be found by using the
continuity conditions across $x = 0$ from
\eqref{eqn:continuity_conditions}.  In particular, we find
\begin{equation}\label{eqn:Vs_hat}
    \Vhat_s(x,q) = A_q\left(\frac{\alpha+1}{2\alpha}\right)
        \ee^{k_0\nu_s x}\left(1 + R\ee^{-2k_0\nu_s x}\right),
\end{equation}
where
\begin{equation}\label{eqn:alpha_R}
    \alpha \equiv \frac{\nu_s}{(-1-\ii\delta)\nu_c}
    \eqtext{and}
    R \equiv \frac{\alpha-1}{\alpha+1} = 
        \frac{\nu_s + (1+\ii\delta)\nu_c}{\nu_s - (1+\ii\delta)\nu_c}.
\end{equation}

Although one can observe that $\alpha$ degenerates for $q^2=k_0^2$ we will see in \eqref{3}, \eqref{4} that 
$\displaystyle\frac{A_q}{\alpha}$ is well defined in the limit when $q^2=k_0^2$.

\subsubsection{The solution in $\mathcal{M}$}

In the region $\mathcal{M}$, the Fourier transform of $V_m$ satisfies
\begin{equation}\label{eqn:matrix_equation_hat}
    \frac{\partial^2 \Vhat_m}{\partial x^2} - k_0^2\nu_m^2\Vhat_m =
    -\fhat(x,q),
    \eqtext{where}
    \nu_m^2 \equiv \frac{q^2}{k_0^2} - 1.
\end{equation}
If $q^2/k_0^2 \ne 1$, then the general solution to
\eqref{eqn:matrix_equation_hat} can be found using the Laplace transform
and the continuity conditions across $x = a$ from
\eqref{eqn:continuity_conditions} \cite{Thaler:2014:BVI,
Meklachi:2016:SAL}; we have
\begin{equation}\label{eqn:matrix_hat_general}
    \begin{aligned}
        \Vhat_m(x,q) = 
        &\frac{\ee^{k_0\nu_m x}}{2}
            \left[A_q\ee^{-k_0\nu_m a}\left(\psi_q^+ +
            \frac{\psi_q^-}{\nu_m}\right) -
            \frac{1}{k_0\nu_m}\int_{d_0}^x \ee^{-k_0\nu_m s}\fhat(s,q) 
            \di{s}\right] \myspace{}
        + &\frac{\ee^{-k_0\nu_m x}}{2}
            \left[A_q\ee^{k_0\nu_m a}\left(\psi_q^+ -
            \frac{\psi_q^-}{\nu_m}\right) +
            \frac{1}{k_0\nu_m}\int_{d_0}^x \ee^{k_0\nu_m s}\fhat(s,q) 
            \di{s}\right],
    \end{aligned}
\end{equation}
where
\begin{equation}\label{eqn:psi}
    \begin{aligned}
        \psi_q^+ 
            &\equiv \frac{1}{A_q}\Vhat_s(a,q) 
            = \left(\frac{\alpha+1}{2\alpha}\right)
                \ee^{k_0\nu_sa}\left(1+R\ee^{-2k_0\nu_sa}\right); \\
        \psi_q^- 
            &\equiv \frac{1}{k_0A_q}\left(\frac{1}{(-1-\ii\delta)}
                \frac{\partial\Vhat_s}{\partial x}(a,q) \right)
            = \left(\frac{\nu_s}{-1-\ii\delta}\right) 
              \left(\frac{\alpha+1}{2\alpha}\right)
              \ee^{k_0\nu_sa}\left(1-R\ee^{-2k_0\nu_sa}\right).
    \end{aligned}
\end{equation}

If $q^2/k_0^2 < 1$, then $\nu_m$ is purely imaginary.  Because $V_m$
should be outgoing (i.e., rightgoing) as $x\rightarrow \infty$ and we
are considering $\ee^{\ii\omega t}$ time dependence, we should have 
\[
    \Vhat_m \sim \ee^{-\ii k_0\sqrt{|q^2/k_0^2-1|} x} \eqtext{as} x
    \rightarrow \infty.
\]
To ensure this, we take the first expression in brackets in
\eqref{eqn:matrix_hat_general} to be zero and find that
\begin{equation}\label{eqn:Aq}
    A_q \equiv \frac{I_q\ee^{k_0\nu_m a}}{k_0\left(\nu_m\psi_q^+ +
    \psi_q^-\right)},
\end{equation}
where
\begin{equation}\label{eqn:Iq}
    I_q \equiv \int_{d_0}^{d_1} \fhat(s,q)\ee^{-k_0\nu_m s} \di{s}.
\end{equation}

If $q^2/k_0^2 > 1$, then $\nu_m > 0$; to ensure that $\Vhat_m(x,q)
\rightarrow 0$ as $x\rightarrow \infty$, we again take $A_q$ as in
\eqref{eqn:Aq}.  

Finally, if $q^2/k_0^2 = 1$, then we can use the Laplace transform and
the continuity conditions across $x = a$ to find that
\begin{equation}
\label{3}
    \Vhat_m(x,\pm k_0)  = A_{\pm k_0}(\psi_{\pm k_0}^+ - a\psi_{\pm k_0}^-) + \int_{d_0}^x s\fhat(s,{\pm k_0})
    \di{s} + x\left[k_0A_{\pm k_0}\psi_{\pm k_0}^- - \int_{d_0}^x \fhat(s,{\pm k_0})
    \di{s}\right].
\end{equation}
where 

\begin{equation}\label{eqn:psi-k0}
    \begin{aligned}
        \psi_{\pm k_0}^+ 
            &
            = \left(1+\ee^{-2k_0\nu_sa}\right); \\
        \psi_{\pm k_0}^- 
            &
            = \left(\frac{\nu_s}{-1-\ii\delta}\right) 
              \left(1-\ee^{-2k_0\nu_sa}\right).
    \end{aligned}
\end{equation}
with $\nu_s$ defined at \eqref{eqn:shell_equation_hat} is computed for $q=\pm k_0$, and where again we take $A_{\pm k_0}$ so that we ensure $\Vhat_m$ is outgoing as $x\rightarrow \infty$; in this case
\begin{equation}
\label{4}
    A_{\pm k_0} = \frac{1}{k_0\psi_{\pm k_0}^-}\int_{d_0}^{d_1} \fhat(s,{\pm k_0})\di{s}.
    \end{equation}


\subsection{Energy discussion}

\label{Energy}

For $0 < \xi \le a$, we define the strip
\begin{equation}\label{eqn:S_xi}
    S_{\xi} \equiv \{(x,y) \in \mathbb{R}^2 : a-\xi < x < a\}.
\end{equation}
Then, due to the Plancherel
Theorem and properties of Fourier transforms, we have
\begin{align*}
    \|\nabla V\|^2_{L^2(S_{\xi})} 
    &= \int_{a-\xi}^a \int_{-\infty}^{\infty} |\nabla V_s(x,y)|^2
        \di{y}\di{x} \\
    &= \int_{a-\xi}^a \int_{-\infty}^{\infty} 
        \left|\frac{\partial^2 V_s}{\partial x^2}\right|^2 + 
        \left|\frac{\partial^2 V_s}{\partial y^2}\right|^2 \di{y} \di{x}
        \\
    &= \frac{1}{2\pi}\int_{a-\xi}^a \int_{-\infty}^{\infty}
        \left|\frac{\partial^2 \Vhat_s}{\partial x^2}\right|^2 +
        |q|^2\left|\Vhat_s\right|^2 \di{q} \di{x}.
\end{align*}
Using \eqref{eqn:Vs_hat}--\eqref{eqn:alpha_R} and
\eqref{eqn:psi}--\eqref{eqn:Iq} in this expression, switching the order
of integration, computing the integral with respect to $x$, using the
fact that $|\nabla V_s|^2$ is an even function of $q$ if $f$ is
real-valued, making the change of variables $p = q/k_0$, and simplifying
the resulting expression, we obtain
\begin{equation}\label{eqn:pd}
    \begin{aligned}
        E_{\delta}(\xi) 
        &\equiv \|\nabla V\|^2_{L^2(S_{\xi})} \\
        &= \frac{1+\delta^2}{\pi}\int_{0}^{\infty}
            \frac{|I_p|^2\ee^{2k_0\nu_m'a}|\nu_s-(1+\ii\delta)\nu_m|^2}
            {|g_{\delta}(p;\gamma)|^2}\cdot\\
        &\qquad\left\{
            \left(|\nu_s|^2 + |p|^2\right)
            \left[\left(\frac{1-\ee^{-2k_0\nu_s'\xi}}{2\nu_s'}\right)
            + |R|^2\ee^{-4k_0\nu_s'\left(a-\frac{\xi}{2}\right)}
            \left(\frac{1-\ee^{-2k_0\nu_s'\xi}}{2\nu_s'}\right)\right]
            \right.\\
        &\qquad\left.
            + 2\left(-|\nu_s|^2 + |p|^2\right)
            \ee^{-2k_0\nu_s'a}\imag
            \left[\overline{R}\ee^{2\ii k_0\nu_s''a}
            \left(\frac{1-\ee^{-2\ii k_0\nu_s''\xi}}{2\nu_s''}\right)
            \right]\right\}\di{p},
    \end{aligned}
\end{equation}
where
\begin{equation}\label{eqn:g}
    g_{\delta}(p;\gamma) \equiv
    \left[\nu_s-\left(1+\ii\delta\right)\nu_m\right]^2 - \left[\nu_s +
    \left(1+\ii\delta\right)\nu_m\right]^2 \ee^{-2\gamma\nu_s},
\end{equation}
we have used that fact that $\nu_c = \nu_m$ (see \eqref{eqn:nuc} and
\eqref{eqn:matrix_equation_hat}), and we have replaced $q$ by $k_0 p$
throughout the integrand (e.g., we have $\nu_m = \sqrt{p^2-1}$).

Similarly, we have
\begin{equation}\label{eqn:v_int}
    \begin{aligned}
        \|V\|^2_{L^2(S_{\xi})}
        &= \frac{1+\delta^2}{\pi}\int_{0}^{\infty}
            \frac{|I_p|^2\ee^{2k_0\nu_m'a}|\nu_s-(1+\ii\delta)\nu_m|^2}
            {|g_{\delta}(p;\gamma)|^2}\cdot\\
        &\qquad\left\{
            \left[\left(\frac{1-\ee^{-2k_0\nu_s'\xi}}{2\nu_s'}\right)
            + |R|^2\ee^{-4k_0\nu_s'\left(a-\frac{\xi}{2}\right)}
            \left(\frac{1-\ee^{-2k_0\nu_s'\xi}}{2\nu_s'}\right)\right]
            \right.\\
        &\qquad\left.
            + 2
            \ee^{-2k_0\nu_s'a}\imag
            \left[\overline{R}\ee^{2\ii k_0\nu_s''a}
            \left(\frac{1-\ee^{-2\ii k_0\nu_s''\xi}}{2\nu_s''}\right)
            \right]\right\}\di{p},
    \end{aligned}
\end{equation}
\begin{remark}\label{rem:V_and_grad_V}
    One of the quantities we are most interested in studying in this
    paper is
    \[
        \|V\|_{H^1(S_{\xi})}^2 = \|V\|_{L^2(S_{\xi})}^2 + 
            \|\nabla V\|_{L^2(S_{\xi})}^2.
    \]
    Due to the similarity between the expressions in \eqref{eqn:pd} and
    \eqref{eqn:v_int}, without loss of generality we focus on $\|\nabla
    V\|^2_{L^2(S_{\xi})}$. In particular, our arguments 
    depend heavily on the exponential terms in the integrands in
    \eqref{eqn:pd} and \eqref{eqn:v_int}, so the additional terms
    $|\nu_s|^2$ and $|q|^2$ in \eqref{eqn:pd} will have no bearing on
    our results.
\end{remark}


\section{Properties of $g_{\delta}(p;\gamma)$}

In this section, we collect some essential properties about the
denominator $|g_{\delta}|^2$ in \eqref{eqn:pd}.  As we will see, the
parameter 
\begin{equation}\label{eqn:gamma}
    \gamma \equiv k_0 a
\end{equation}
plays a crucial role in the behavior of the solution $V$ and
$E_{\delta}(a)$ in the limit $\delta \rightarrow 0^+$.
\begin{lemma}\label{lem:g_0}
    Suppose $g_{\delta}$ is defined as in \eqref{eqn:g}.  Then for $p
    \ge 0$ and $\gamma > 0$ we have
    \begin{equation}\label{eqn:g_0}
        \lim_{\delta\rightarrow 0^+} g_{\delta}(p;\gamma) =
        g_0(p;\gamma) 
        \equiv 
        \left(\sqrt{p^2+1}-\sqrt{p^2-1}\right)^2 -
        \left(\sqrt{p^2+1}+\sqrt{p^2-1}\right)^2
        \ee^{-2\gamma\sqrt{p^2+1}}.
    \end{equation}
\end{lemma}
\begin{proof}
    The result follows from direct calculations since $g_{\delta}$ is a
    continuous function of $\delta$.
\end{proof}
The next lemma  plays an essential role in the following discussion.
\begin{lemma}\label{lem:g_0_roots}
    Suppose $g_0(p;\gamma)$ is defined as in \eqref{eqn:g_0} for $p\ge
    0$ and $\gamma > 0$.  Then
    there is a $\gamma_* \approx 0.9373$ such that
    \begin{enumerate}
        \item if $0 < \gamma < \gamma_*$, then $g_0(p;\gamma)$ has two
            distinct real roots of order $1$, namely $1 < p^1_{\gamma} 
            < p^2_{\gamma}$;
        \item if $\gamma > \gamma_*$, then $g_0(p;\gamma)$ has no
            real roots.
    \end{enumerate}
\end{lemma}
We note that $\gamma_*$ can be computed as the solution of an
optimization problem; more importantly, we emphasize that
Lemmas~\ref{lem:g_0}--\ref{lem:g_0_roots} are \emph{independent} of the
source term $f$ in \eqref{eqn:finite_freq}.  We will see later that the
roots of $g_0(p;\gamma)$ are indicative of anomalous localized
resonance.  For brevity, we defer the proof of Lemma~\ref{lem:g_0_roots}
to the appendix.


\section{Short wavelength/high frequency regime ($\gamma > \gamma_*$)}

\label{sec:large_gamma}

In this section, we prove that, for $\gamma > \gamma_*$ (where $\gamma$ was introduced at \eqref{eqn:gamma}), $E_{a}(\delta)$ remains bounded as
$\delta \rightarrow 0^+$ for all sources $f\in L^2(\mathcal{M})$ with
bounded support in $\mathcal{M}$, regardless of how close the source is
to the slab. In addition, we also prove that the slab lens behaves as a
``shield'' in the sense that the solution to the left of the lens, i.e.,
$V_c$, is vanishingly small in the limit $\delta \rightarrow 0^+$.


\subsection{$E_{\delta}(a)$ for $\gamma > \gamma_*$}
\label{subsec:pd_large_gamma}

From \eqref{eqn:pd}, we have
\begin{equation}\label{eqn:pd_a}
    E_{\delta}(a) = \int_0^{\infty} L_{\delta}(p;\gamma) \di{p},
\end{equation}
where, for $\delta > 0$, $p \ge 0$, and $\gamma > 0$, 
\begin{equation}\label{eqn:L}
    L_{\delta}(p;\gamma) \equiv 
    \frac{|I_p|^2\ee^{2\gamma\nu_m'}}{|g_{\delta}(p;\gamma)|^2}
    M_{\delta}(p;\gamma)
\end{equation}
and
\begin{multline}\label{eqn:M}
    M_{\delta}(p;\gamma) \equiv 
    \frac{1+\delta^2}{\pi}|\nu_s-(1+\ii\delta)\nu_m|^2\left\{
        \left(|\nu_s|^2 + p^2\right)
        \left[\left(\frac{1-\ee^{-2\gamma\nu_s'}}{2\nu_s'}\right)
        + |R|^2\ee^{-2\gamma\nu_s'}
        \left(\frac{1-\ee^{-2\gamma\nu_s'}}{2\nu_s'}\right)\right]
        \right.\\
    \left.
        + 2\left(-|\nu_s|^2 + p^2\right)
        \ee^{-2\gamma\nu_s'}\imag
        \left[\overline{R}\ee^{2\ii \gamma\nu_s''}
        \left(\frac{1-\ee^{-2\ii \gamma\nu_s''}}{2\nu_s''}\right)
        \right]\right\}.
\end{multline}
We now state the main theorem from this section.
\begin{theorem}\label{thm:bounded}
    Suppose $\gamma > \gamma_*$ (where $\gamma_*$ is introduced in
    Lemma~\ref{lem:g_0_roots}).  If there is a constant $C > 0$ such
    that
    \begin{equation}\label{eqn:I_p_bounds}
        |I_p| \le
        \begin{cases}
            C &\text{for } 0 \le p \le 1,\\
            C\ee^{-\gamma\frac{d_0}{a}\sqrt{p^2-1}} &\text{for } 1 \le p
            < \infty,
        \end{cases}
    \end{equation}
    then there is a constant $C_{\gamma} > 0$ and a $\delta_{\gamma} >
    0$ such that $\|V\|_{H^1(S_a)} \le C_{\gamma}$ as for all $\delta
    \le \delta_{\gamma}$.
\end{theorem}
The proof of this theorem is somewhat tedious and may be found in the
appendix --- although we only prove the theorem for $\|\nabla
V\|^2_{L^2(S_a)}$, Remark~\ref{rem:V_and_grad_V} implies that it holds
for $\|V\|_{L^2(S_a)}$ as well. In the next lemma, we show that the bound \eqref{eqn:I_p_bounds} holds
for very general sources $f$.
\begin{lemma}\label{lem:I_p_upper_bound}
    Suppose $f\in L^2(\mathcal{M})$ with compact support; then
    \eqref{eqn:I_p_bounds} holds.
\end{lemma}
\begin{proof}
    For $0 \le p \le 1$, recall from \eqref{eqn:Iq} that 
    \[
        I_p = \int_{d_0}^{d_1} \fhat(s,k_0p)
            \ee^{-\ii k_0\sqrt{1-p^2}s} \di{s}.
    \]
    Then the triangle, Cauchy--Schwarz, and Jensen inequalities
    imply 
    \begin{align*}
        |I_p| &\le \int_{d_0}^{d_1} \left|\fhat(s,k_0p)\right|\di{s}
        \\
        &\le (d_1-d_0)^{1/2} \left[\int_{d_0}^{d_1}
            \left|\fhat(s,k_0p)\right|^2\di{s}\right]^{1/2} \\
        &= (d_1-d_0)^{1/2} \left[\int_{d_0}^{d_1}
            \left|\int_{-\infty}^{\infty} f(s,y) \ee^{-\ii k_0py}
            \di{y}\right|^2\di{s}\right]^{1/2} \\
        &\le(d_1-d_0)^{1/2} \left[\int_{d_0}^{d_1}
            \int_{-\infty}^{\infty} \left|f(s,y)\right|^2 \di{y}
            \di{s}\right]^{1/2} \\
        &= (d_1-d_0)^{1/2}\|f\|_{L^2(\mathcal{M})}.
    \end{align*}
    
    Similarly, for $p \ge 1$, recall from \eqref{eqn:Iq} that
    \[
        I_p = \int_{d_0}^{d_1} \fhat(s,k_0p)
        \ee^{-k_0\sqrt{p^2-1}s}\di{s}.
    \]
    Then 
    \begin{align*}
        |I_p| &\le \int_{d_0}^{d_1} \left|\fhat(s,k_0p)\right|
        \ee^{-k_0\sqrt{p^2-1}s} \di{s} \\
        &\le (d_1-d_0)^{1/2}\left[\int_{d_0}^{d_1}
            \left|\fhat(s,k_0p)\right|^2\ee^{-2k_0\sqrt{p^2-1}s}
        \di{s}\right]^{1/2} \\
        &\le (d_1-d_0)^{1/2}\|f\|_{L^2(\mathcal{M})}
            \ee^{-k_0d_0\sqrt{p^2-1}} \\
        &= (d_1-d_0)^{1/2}\|f\|_{L^2(\mathcal{M})}
            \ee^{-\gamma \frac{d_0}{a}\sqrt{p^2-1}}.
    \end{align*}
    
    To complete the proof, we define $C \equiv
    (d_1-d_0)^{1/2}\|f\|_{L^2(\mathcal{M})}$.
\end{proof}


\subsection{Shielding effect for large $\gamma$}
\label{subsec:shielding}

It turns out that the slab lens behaves as a shield  and acts as an almost perfect reflector. This fact was also observed in \cite{Kettunen:2015:OAE} where it was explained based on the fact that, at least in the lossless non-magnetic case $\epsilon=-1, \mu=1$ will give a purely imaginary wave number inside the slab and thus no propagation beyond the slab in region $C$. We have,
\begin{theorem}\label{thm:shielding}
    Suppose $\gamma \ge 2\gamma_*$, $|I_p|$ satisfies
    \eqref{eqn:I_p_bounds}, and choose $0 < \eta < 1$; then there
    is a constant $C_{\eta} > 0$ such that
    \begin{equation}\label{eqn:V_c_exp_bound}
        |V_c(x,y)| \le C_\eta\ee^{-\eta\gamma}
        \eqtext{for\ all} (x,y) \in \mathcal{C}.
    \end{equation}
    In particular,
    \[
        |V_c(x,y)| \rightarrow 0 \eqtext{as} k_0 \rightarrow \infty
        \eqtext{for\ all} (x,y) \in \mathcal{C}.  
    \]
\end{theorem}
\begin{remark}\label{rem:dipole_bounded}
    Lemma~\ref{lem:I_p_upper_bound} implies that
    Theorems~\ref{thm:bounded} and \ref{thm:shielding} hold for all
    sources $f \in L^2(\mathcal{M})$ with compact support.  However, the
    bound in \eqref{eqn:I_p_bounds} is stronger than we need.  For
    example, suppose there is a positive, real-valued function
    $B(p;\gamma)$ that is continuous for $0 \le p < \infty$ and
    $\gamma_* \le \gamma < \infty$.  In addition, for every $\epsilon
    >0$, suppose that
    \begin{equation}\label{eqn:B_p_lim}
        \lim_{p\rightarrow\infty}
        B(p;\gamma)\ee^{-\epsilon\gamma\sqrt{p^2-1}} = 0 \eqtext{for\ all}
        \gamma \ge \gamma_*
    \end{equation}
    and
    \begin{equation}\label{eqn:B_gamma_lim}
        \lim_{\gamma\rightarrow\infty}
        B(p;\gamma)\ee^{-\epsilon\gamma\sqrt{p^2-1}}  = 0 \eqtext{for\ all}
        p\ge 1.
    \end{equation}
    For example, if $B(p;\gamma)$ is a continuous function of $p$ and
    $\gamma$ that is of polynomial order for $p \rightarrow \infty$ and
    $\gamma \rightarrow \infty$, it will satisfy \eqref{eqn:B_p_lim} and
    \eqref{eqn:B_gamma_lim}.
    Finally, suppose
    \begin{equation}\label{eqn:I_p_stricter_bounds}
        |I_p| \le
        \begin{cases}
            B(p;\gamma) &\text{for } 0 \le p \le 1,\\
            B(p;\gamma)\ee^{-\gamma\frac{d_0}{a}\sqrt{p^2-1}} 
                &\text{for } 1 \le p < \infty.
        \end{cases}
    \end{equation}
    Then, by appropriately modifying
    \eqref{eqn:L_small_p_bound}--\eqref{eqn:L_int_large_p_bound}, one
    can prove that the result of Theorem~\ref{thm:bounded} will hold for
    sources satisfying \eqref{eqn:I_p_stricter_bounds}.  Similarly, 
    by appropriately modifying
    \eqref{eqn:shielding_small_p}--\eqref{eqn:shielding_large_p}, one can
    show that Theorem~\ref{thm:shielding} also holds for sources
    satisfying \eqref{eqn:I_p_stricter_bounds} as long as we replace
    \eqref{eqn:V_c_exp_bound} by 
    \[
        |V_c(x,y)| \le C_{\eta}\ee^{-(\eta-\epsilon)\gamma}
    \]
    where $0 < \epsilon < \eta$.

    In particular, certain distributional sources such as dipoles,
    quadrupoles, etc.\ satisfy \eqref{eqn:I_p_stricter_bounds} --- see
    \S~\ref{subsec:dipole} for more details.
\end{remark}

In Figure~\ref{fig:dipole_large_gamma}, we plot the solution $V$ to
\eqref{eqn:finite_freq} in the case where $f$ is a dipole with dipole
moment $[1,0]^T$, $\delta = 10^{-12}$, and $\gamma = 4\gamma_*$ (we take
$a = 1$ in all figures throughout the paper).  In
Figures~\ref{fig:dipole_large_gamma}(a) and (b), the dipole is located
at the point $(d_0,0) = (4a,0)$; in
Figures~\ref{fig:dipole_large_gamma}(c) and (d), the dipole is located
closer to the slab at the point $(d_0,0) = (1.2a,0)$.  The solution $V$
is smooth throughout the domain; in addition, we observe the ``shielding
effect'' from Theorem~\ref{thm:shielding} in the region to the left of
the lens.
\begin{figure}[!hbpp]
    \begin{center}
        \begin{tabular}{c c}
            \includegraphics[width=0.45\textwidth]{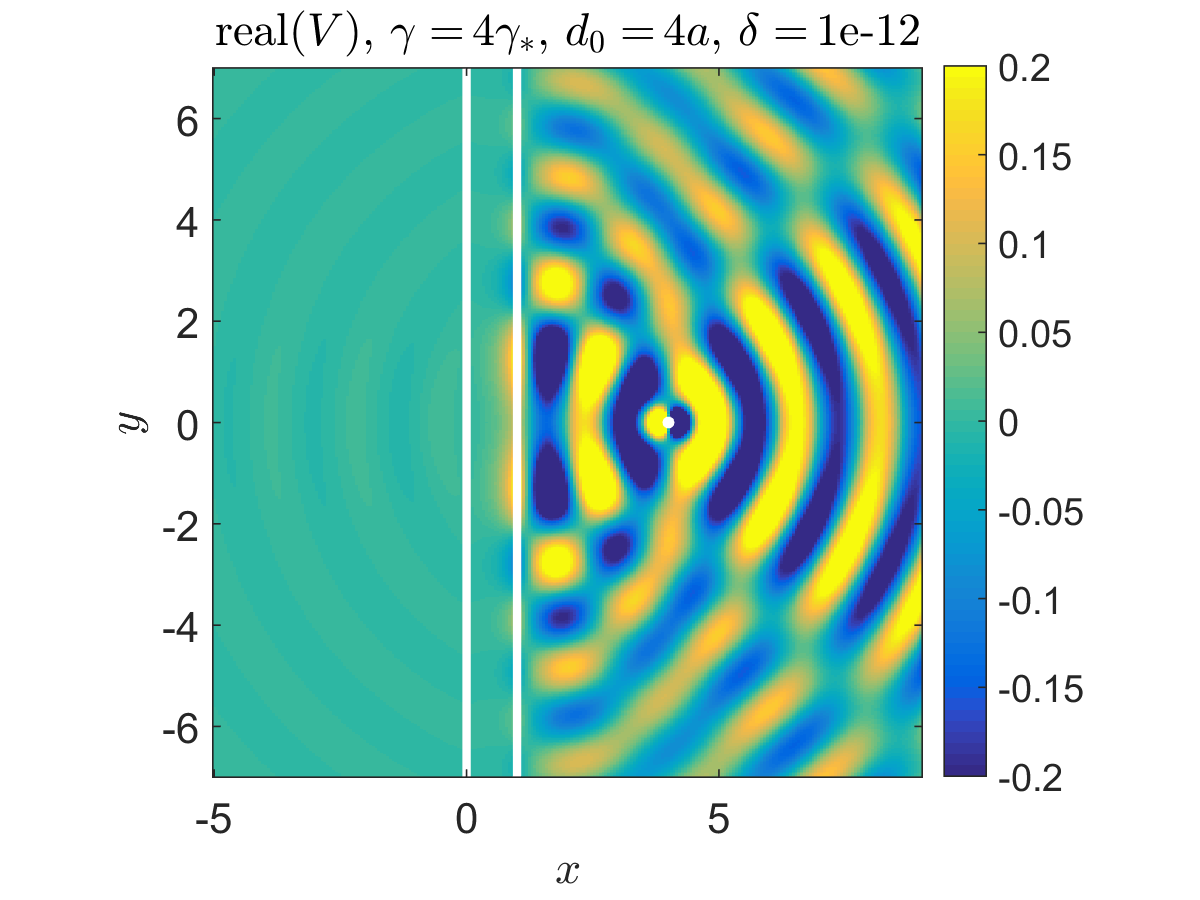}
            & \includegraphics[width=0.45\textwidth]{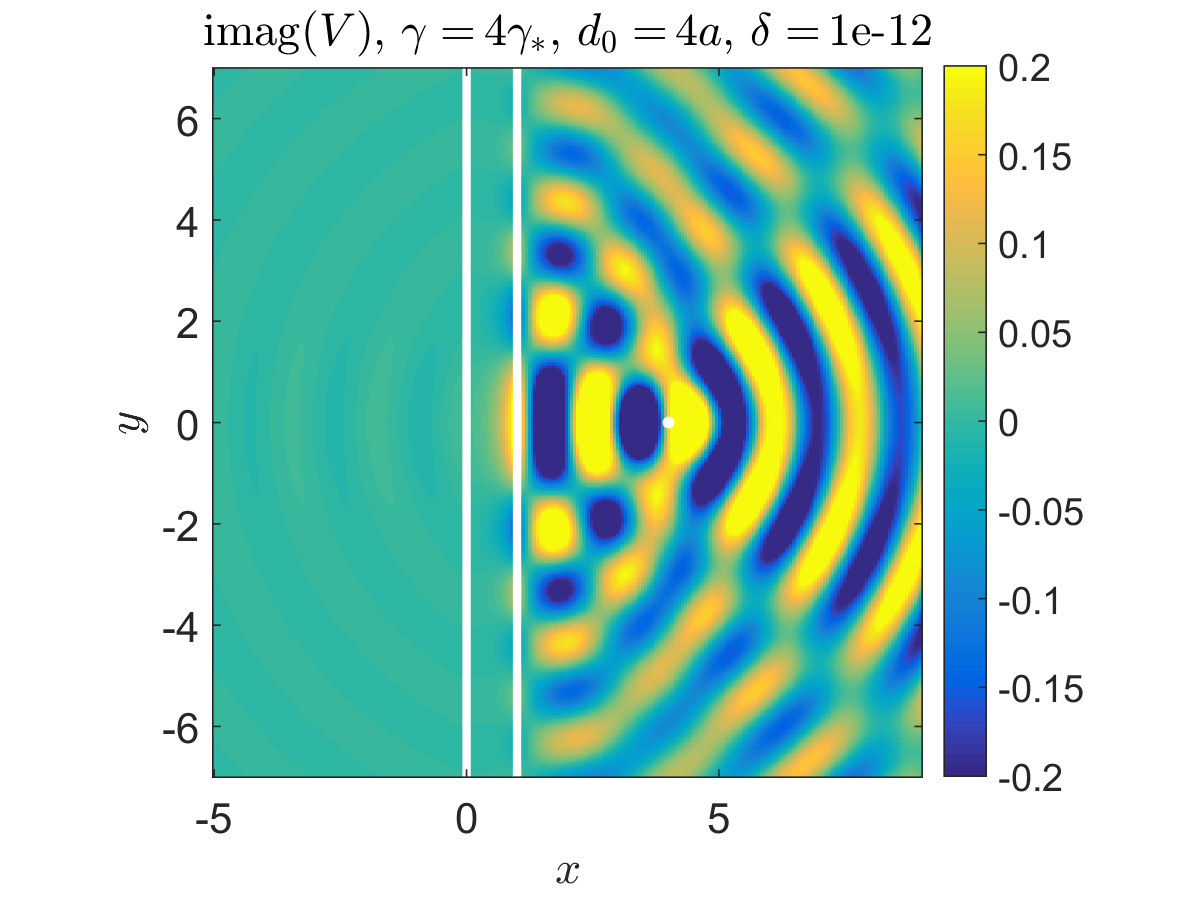}
            \\
            (a) & (b) \myspace{}
            \includegraphics[width=0.45\textwidth]{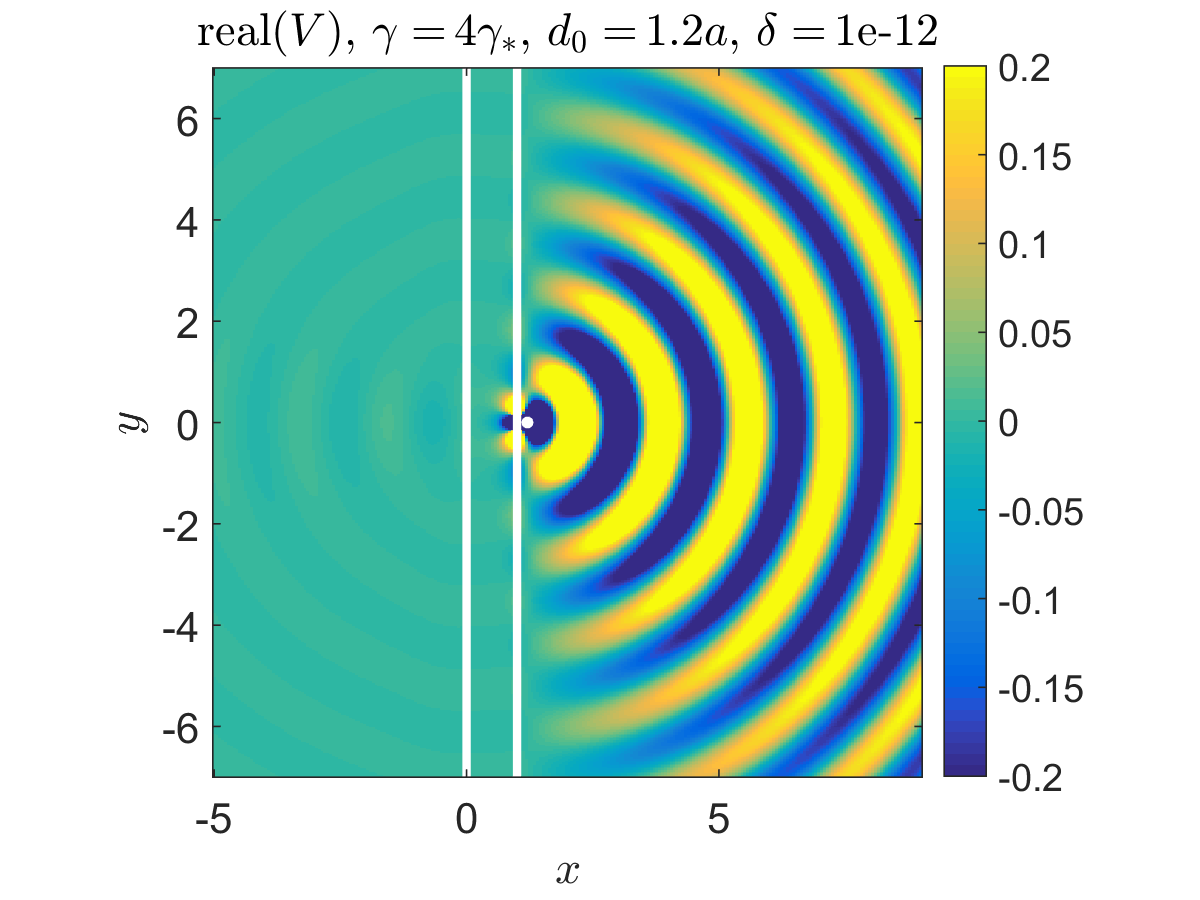}
            & \includegraphics[width=0.45\textwidth]{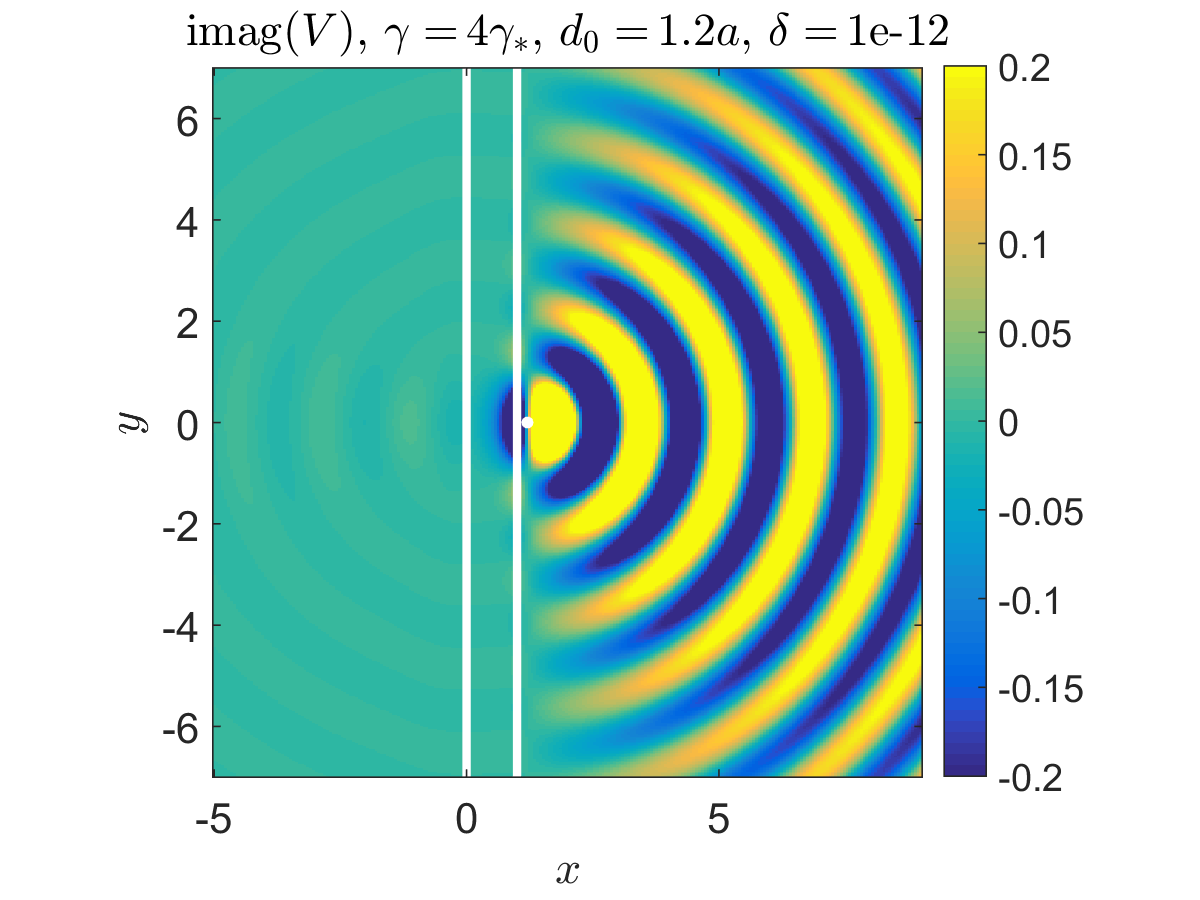}
            \\
            (c) & (d)
        \end{tabular}
        \caption{\emph{This is a plot of $V$, the solution to
                \eqref{eqn:finite_freq}, when $f$ is a dipole and
                $\gamma = 2\gamma_*$: (a) $\real(V)$ and (b) $\imag(V)$
                for $d_0 = 4a$; (c) $\real(V)$ and (d) $\imag(V)$ for
                $d_0 = 1.2a$.  To make the behavior of $V$ more clear,
                we clipped the maximum and minimum values in each plot 
                to $0.2$ (yellow) and $-0.2$ (blue) respectively.}}
        \label{fig:dipole_large_gamma}
    \end{center}
\end{figure}

In Figure~\ref{fig:dipole_bounded_pd}, we plot $E_{\delta}(a)$ as a
function of various parameters for a dipole source $f$.  The parameters
we used are in the ranges $10^{-12} \le \delta \le 10^{-10}$,
$1.01\gamma_* \le \gamma \le 2\gamma_*$, and $1.2a \le d_0 \le 2a$.  We
note that $E_{\delta}(a)$ depends strongly on $\delta$, $\gamma$, and
$d_0$, but, because $\gamma > \gamma_*$, $E_{\delta}(a)$ is quite small.
\begin{figure}[!hbp]
    \begin{center}
        \begin{tabular}{c c}
            \includegraphics[width=0.45\textwidth]{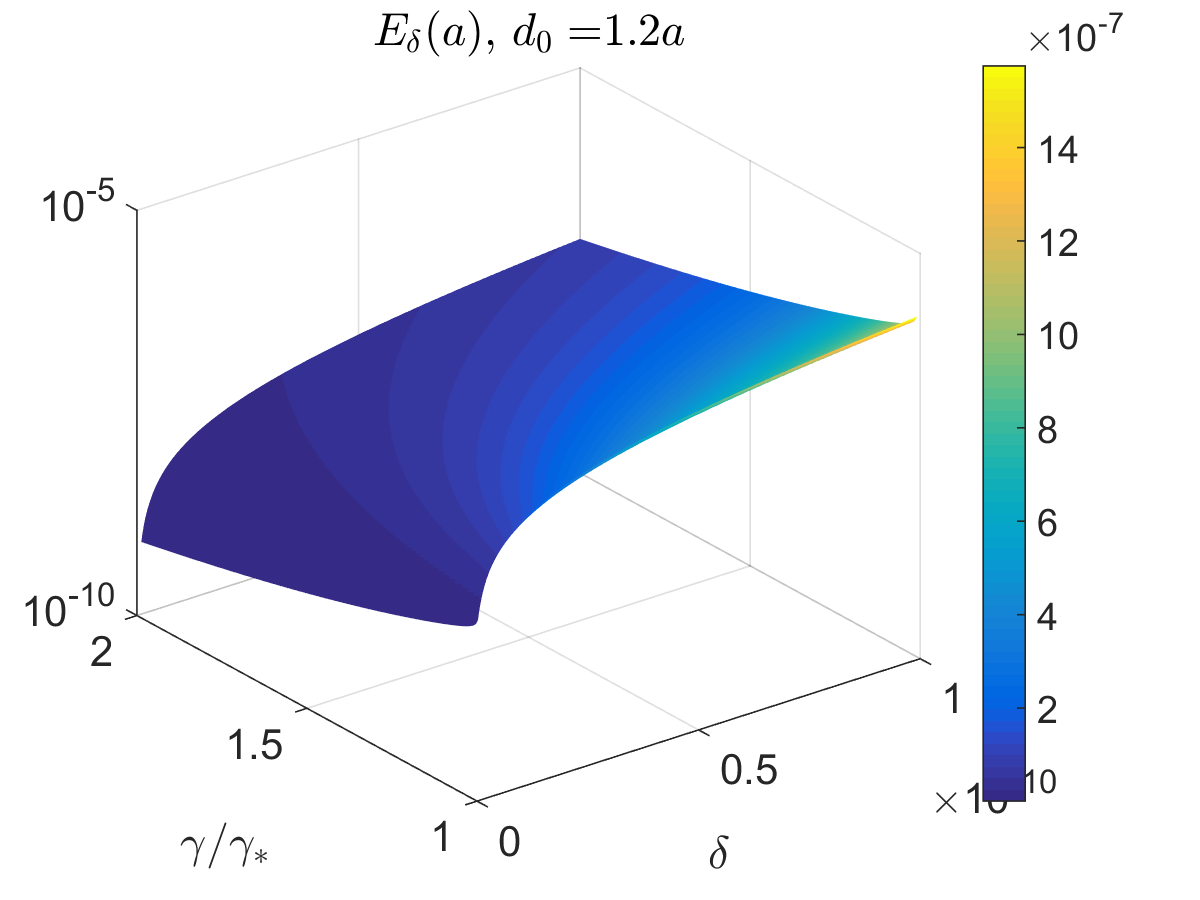}
            &
            \includegraphics[width=0.45\textwidth]{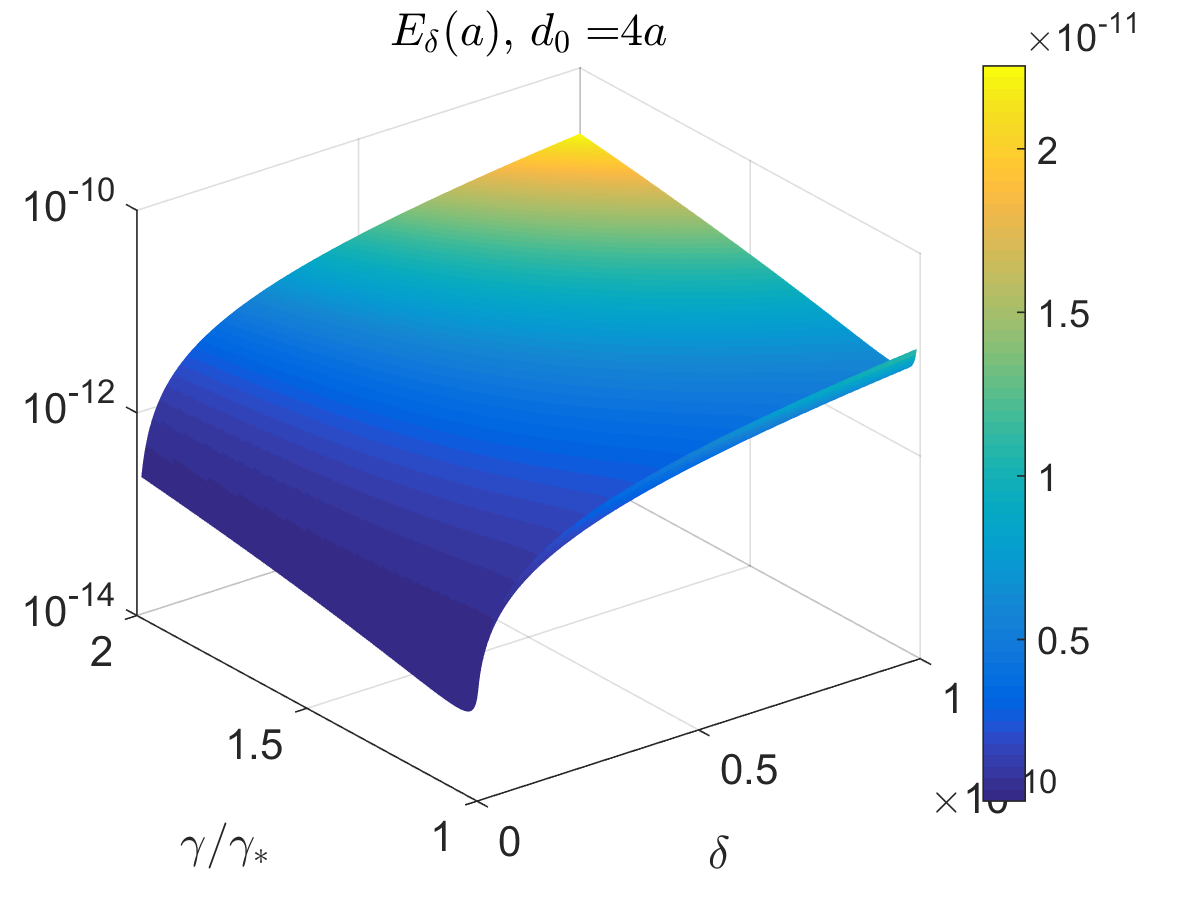}
            \myspace{}
            (a) & (b)\myspace{}
            \includegraphics[width=0.45\textwidth]{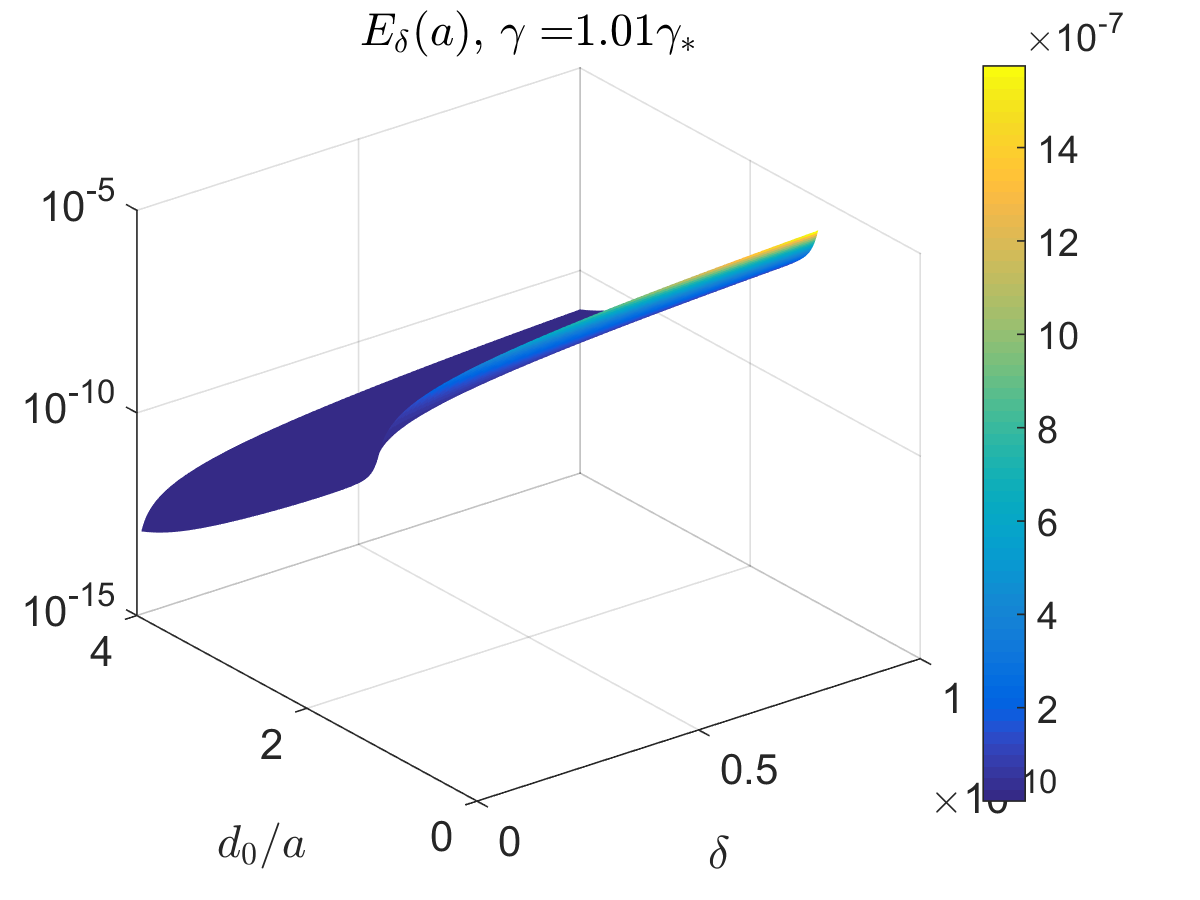}
            &
            \includegraphics[width=0.45\textwidth]{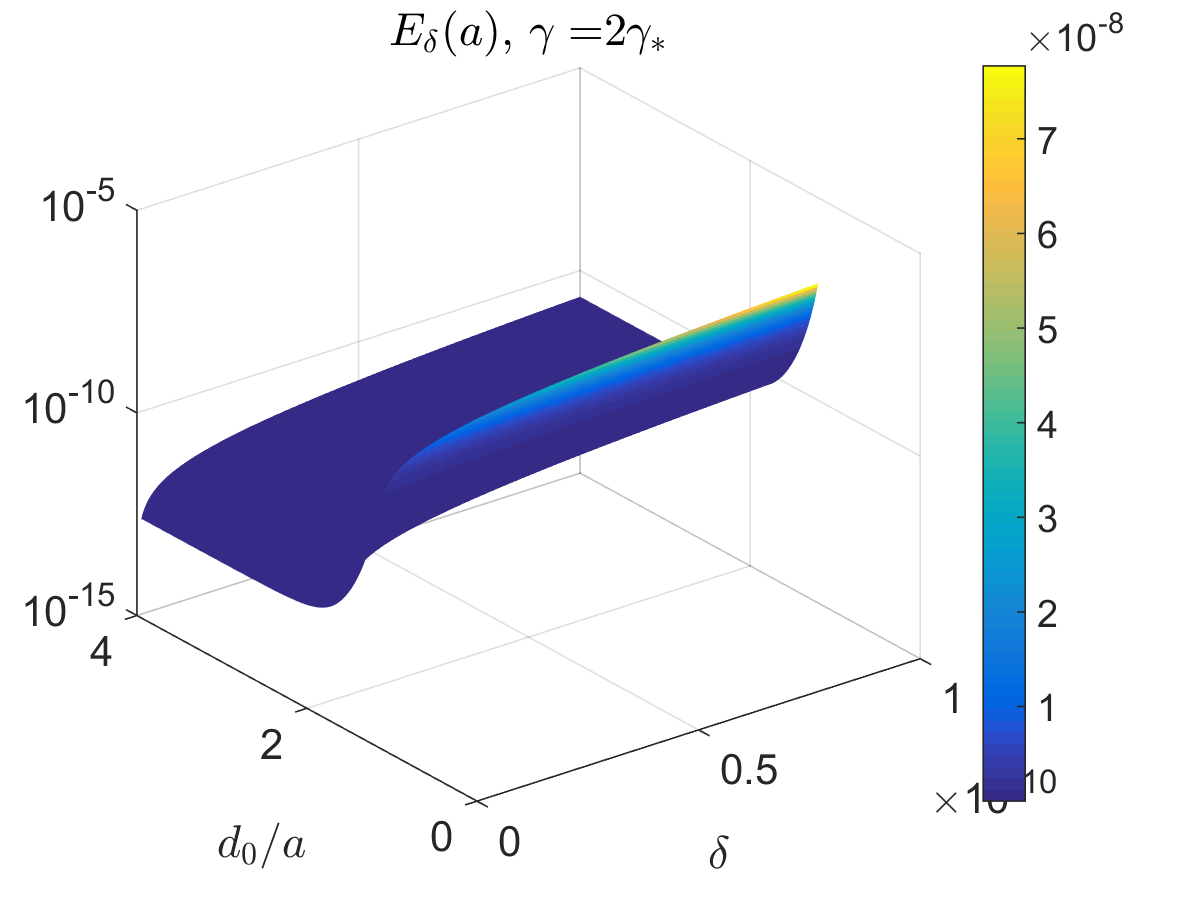}
            \myspace{}
            (c) & (d)\myspace{}
            \includegraphics[width=0.45\textwidth]{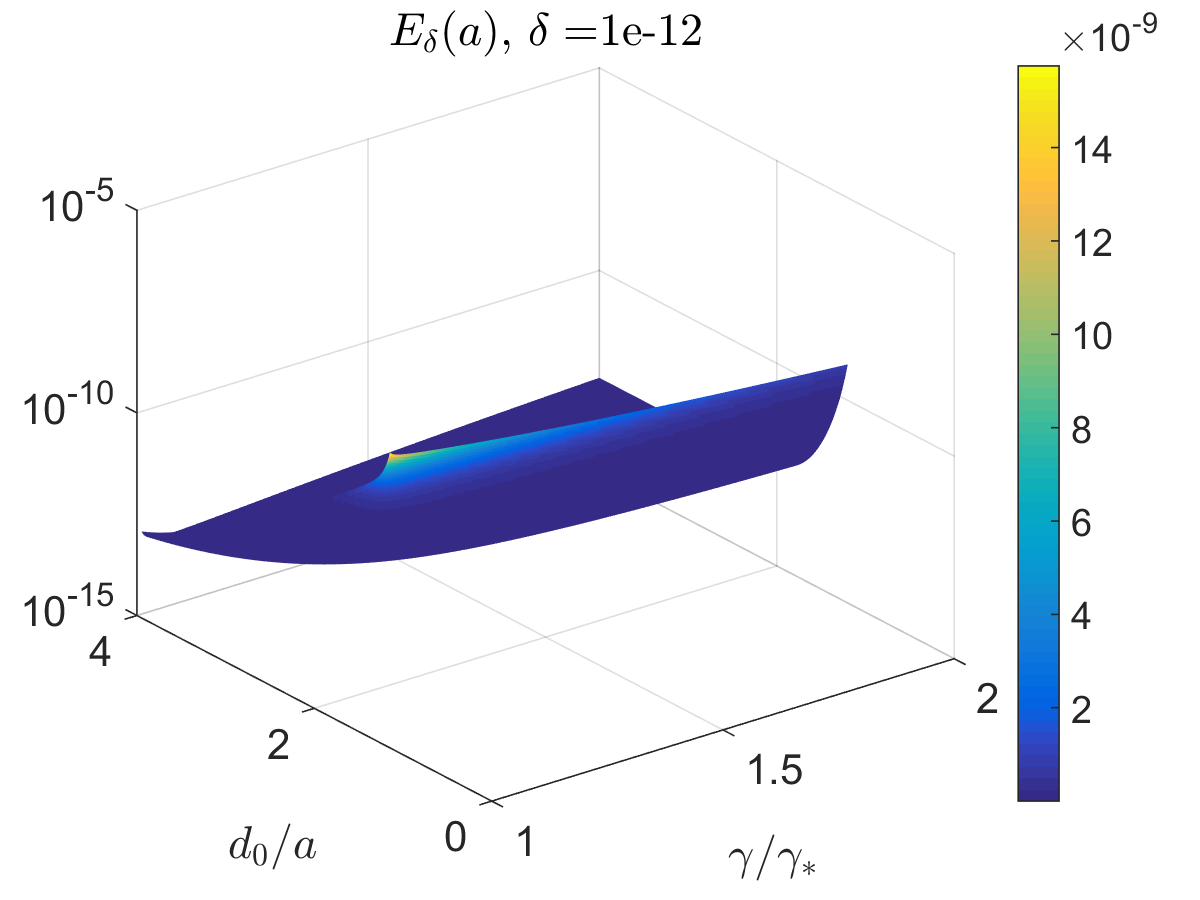}
            &
            \includegraphics[width=0.45\textwidth]{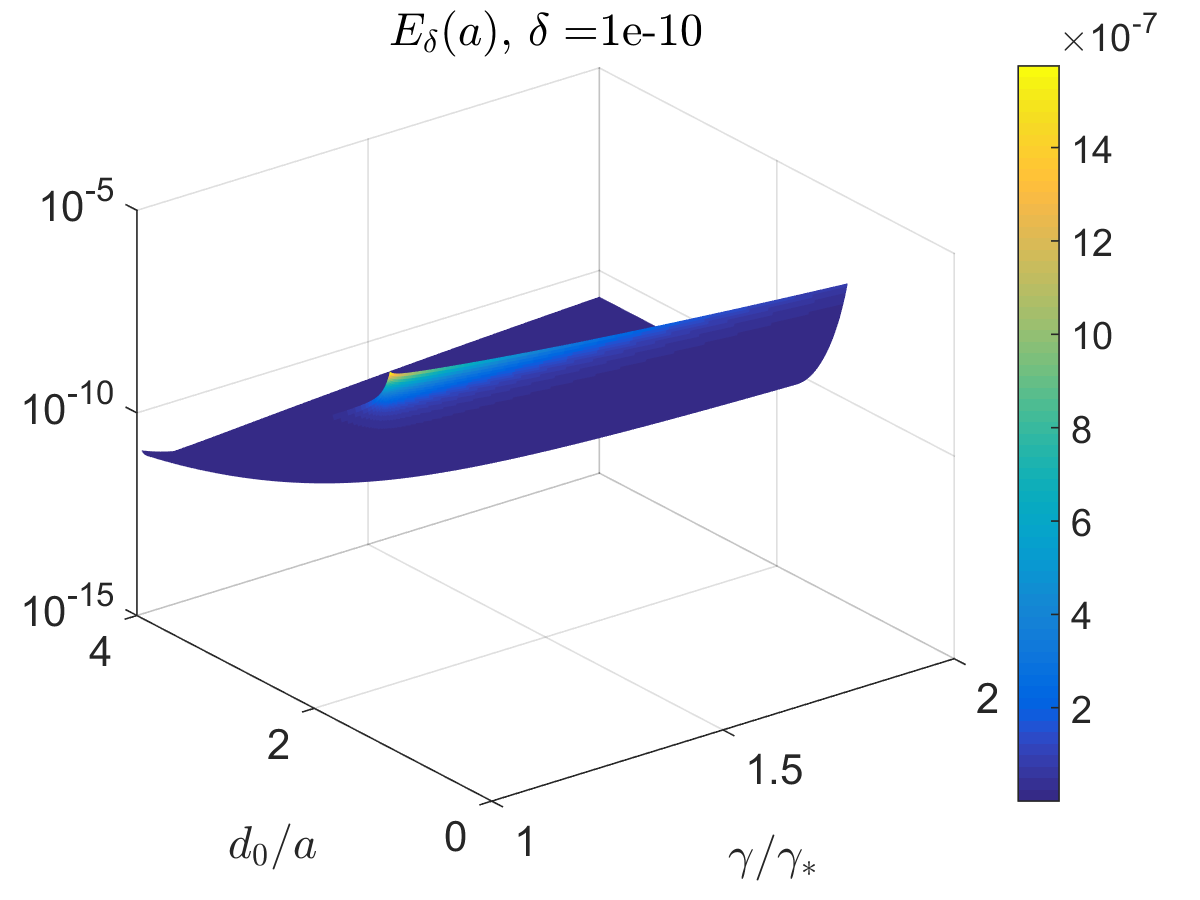}
            \myspace{}
            (e) & (f)
        \end{tabular}
        \caption{\emph{These are plots of $E_{\delta}(a)$ as a function
                of (a) $\delta$ and $\gamma$ ($d_0 = 1.2a$); (b)
                $\delta$ and $\gamma$ ($d_0 = 4a$); (c) $\delta$ and
                $d_0$ ($\gamma = 1.01\gamma_*$); (d) $\delta$ and $d_0$
                ($\gamma = 2\gamma_*$); (e) $\gamma$ and $d_0$ ($\delta 
                = 10^{-10}$); (f) $\gamma$ and $d_0$ 
               ($\delta = 10^{-12}$).}}
        \label{fig:dipole_bounded_pd}
    \end{center}
\end{figure}

Figure~\ref{fig:general_large_gamma} is similar to
Figure~\ref{fig:dipole_large_gamma}, except in
Figure~\ref{fig:general_large_gamma} we take
\begin{equation}\label{eqn:general}
    f(x,y) = 
    \begin{cases}
        \begin{aligned}
            &C\left[\dfrac{2}{d_1-d_0}\left(x-d_0\right)-1\right]^3
        \left[\left|\frac{2}{d_1-d_0}\left(x-d_0\right)-1\right|
        -1\right]^3\\
        &\qquad
        \cdot\left[\dfrac{2}{h_1-h_0}\left(y-h_0\right)-1\right]^3
        \left[\left|\frac{2}{h_1-h_0}\left(y-d_0\right)-1\right|
        -1\right]^3
    \end{aligned}
        &\text{for } d_0 \le x \le d_1, h_0 \le y \le h_1 \\
        0 &\text{otherwise}.
    \end{cases}
\end{equation}
To construct the plots, we have taken $C = 10^4$, $h_1 = -h_0 = 1$, and
$d_1 = d_0 + 2$.  The solution $V$ is smooth throughout the domain and
very small in the region to the left of the slab.
\begin{figure}[!hbp]
    \begin{center}
        \begin{tabular}{c c}
            \includegraphics[width=0.45\textwidth]{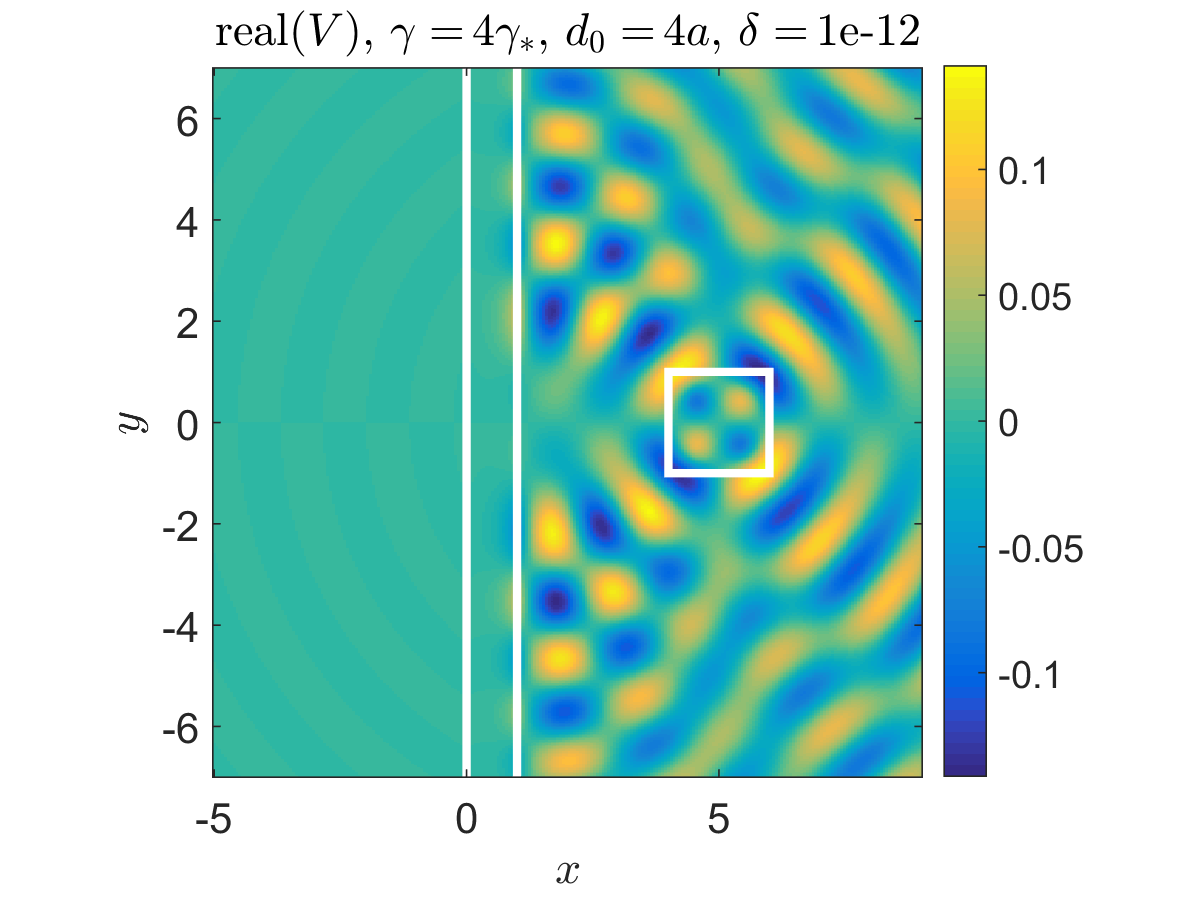}
            & \includegraphics[width=0.45\textwidth]{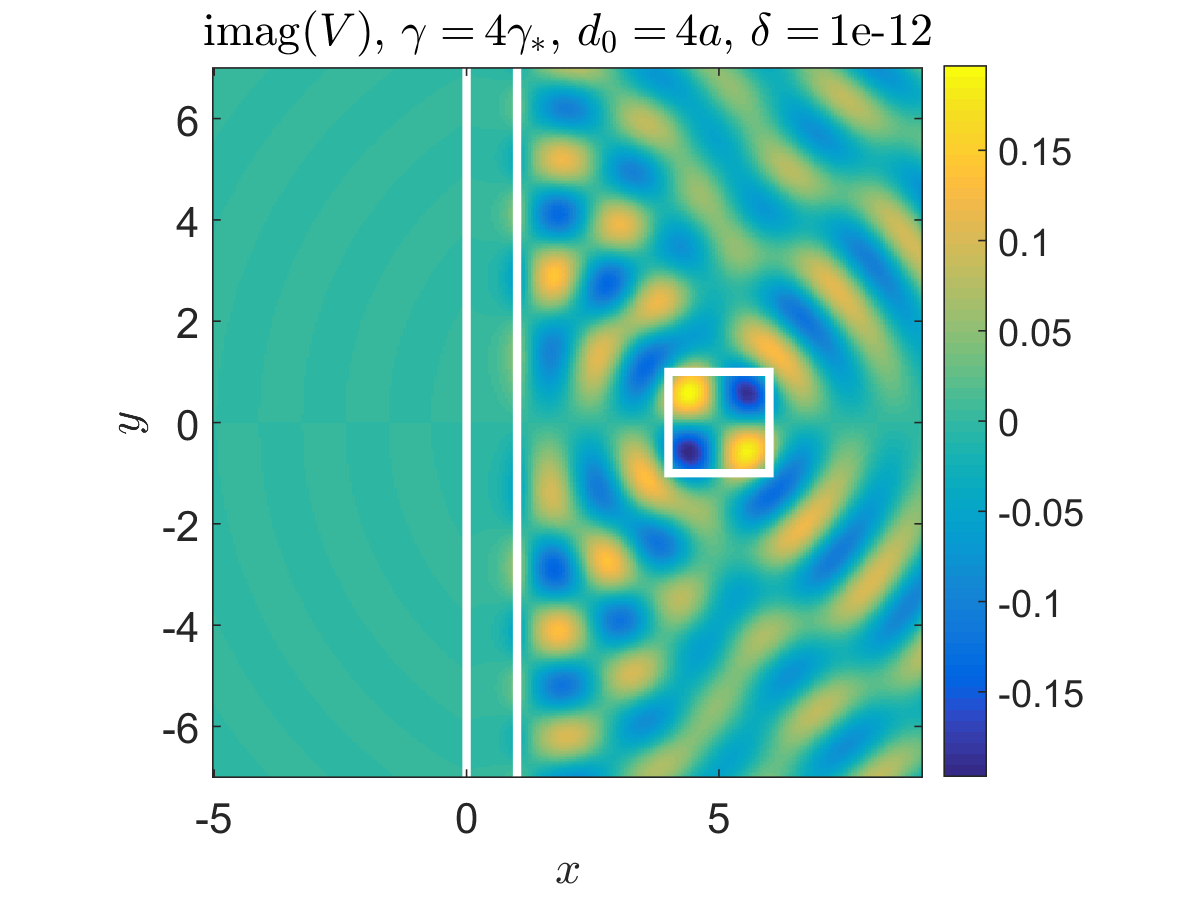}
            \\
            (a) & (b) \myspace{}
            \includegraphics[width=0.45\textwidth]{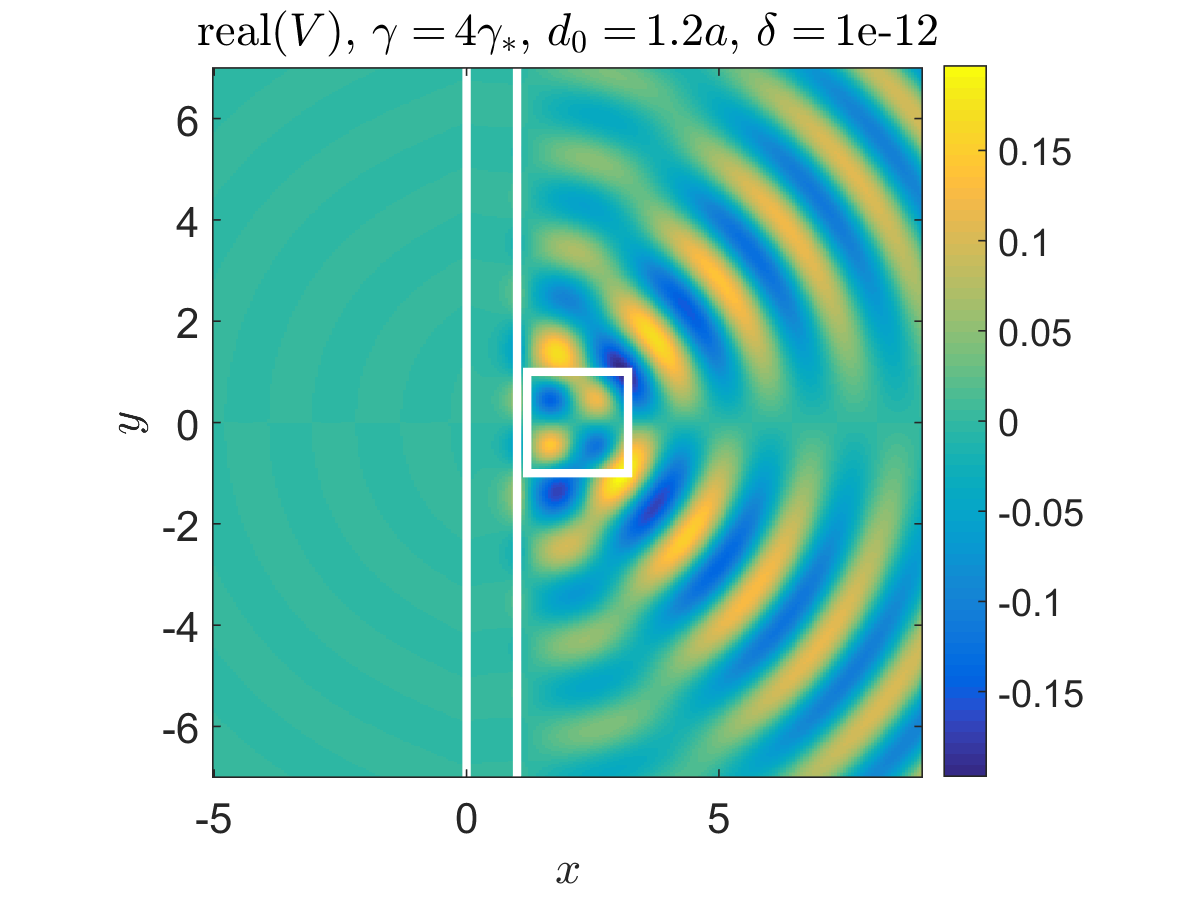}
            & \includegraphics[width=0.45\textwidth]{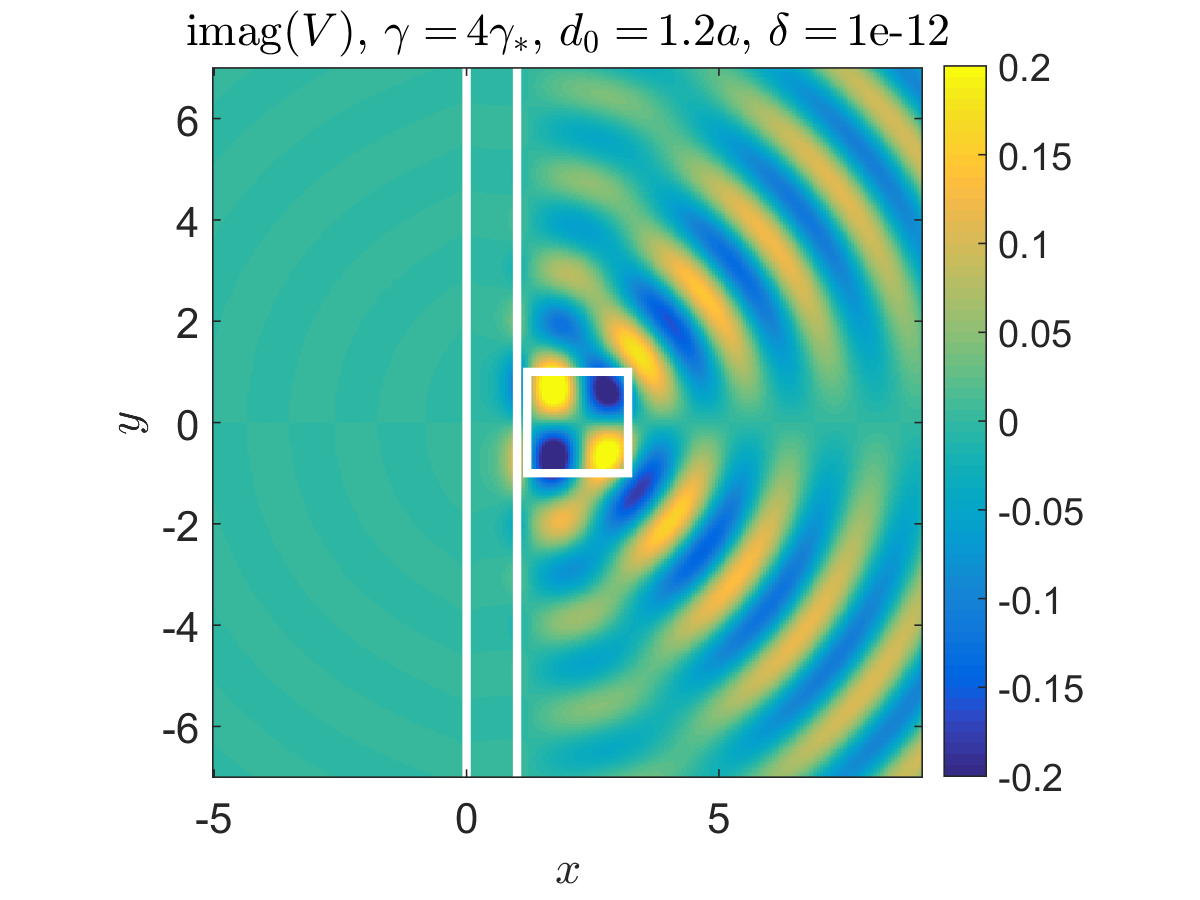}
            \\
            (c) & (d)
        \end{tabular}
        \caption{\emph{This is a plot of $V$, the solution to
                \eqref{eqn:finite_freq}, when $f$ is the function in
                \eqref{eqn:general} and
                $\gamma = 2\gamma_*$: (a) $\real(V)$ and (b) $\imag(V)$
                for $d_0 = 4a$; (c) $\real(V)$ and (d) $\imag(V)$ for
                $d_0 = 1.2a$.  To make the behavior of $V$ more clear,
                we clipped the maximum and minimum values in each plot 
                to $0.2$ (yellow) and $-0.2$ (blue) respectively.}}
        \label{fig:general_large_gamma}
    \end{center}
\end{figure}


\section{Long wavelength/low frequency regime ($\gamma < \gamma_*$)}

Unfortunately, the complicated nature of the expression \eqref{eqn:pd}
has thus far prevented us from deriving lower bounds on
$E_{\delta}(a)$ that would allow us to prove that $E_{\delta}(a)
\rightarrow \infty$ as $\delta\rightarrow 0^+$.  Undaunted, in this
section we present an heuristic argument, coupled with numerical
experiments, to illustrate why we believe the slab lens under
consideration exhibits ALR in the long-wavelength regime.


\subsection{Blow-up of $E_{\delta}(a)$}\label{subsec:pd_small_gamma}

The key result of this section is Lemma~\ref{lem:g_0_roots}:
$|g_0(p;\gamma)|$ has two real roots when $\gamma < \gamma_*$, namely $1
< p^1_{\gamma} < p^2_{\gamma}$.  Because both roots are larger than $1$,
the main contribution to the blow-up of $E_{\delta}(a)$ comes from the
integral over the interval $1 \le p < \infty$.  Indeed, the following
lemma shows that we do not need to worry about the integral over the
interval $0 \le p \le 1$.
\begin{lemma}\label{lem:blow_up_small_p}
    Suppose $0 < \gamma \le \gamma_*$ and $f \in L^2(\mathcal{M})$ with
    compact support.  Then there is a positive
    constant $C_{\gamma}$ and a $\delta_{\gamma} > 0$ such that 
    \[
        \int_0^1 L_{\delta}(p;\gamma) \di{p} = 
        \int_0^1 \frac{|I_p|^2}{|g_{\delta}|^2}M_{\delta}(p;\gamma)
        \di{p} \le C_{\gamma}
    \]
    for all $0 < \delta \le \delta_{\gamma}$.
\end{lemma}
\begin{remark}\label{rem:blow_up}
    We emphasize that Lemma~\ref{lem:blow_up_small_p} also holds for
    those sources for which the bound in \eqref{eqn:I_p_stricter_bounds}
    holds (e.g., dipole sources) --- see
    Remark~\ref{rem:dipole_bounded}.
\end{remark}
\begin{proof}
    First, we note that $M_{\delta}(p;\gamma)$ is continuous for
    $\delta\in[0,1]$, $p \in [0,1]$, and $\gamma \in [0,\gamma_*]$, so
    it is bounded by a constant independent of $\delta$, $p$, and
    $\gamma$.  Additionally, $|I_p|^2$ is also bounded by a constant,
    thanks to Lemma~\ref{lem:I_p_upper_bound}.  All that remains for us
    to show is that $|g_{\delta}(p;\gamma)|$ is bounded away from $0$.  

    We define the function 
    \begin{equation}\label{eqn:Xi_def}
        \Xi_{\delta}(\gamma) \equiv \max_{p\in [0,1]}
        \left||g_{\delta}(p;\gamma)|-|g_0(p;\gamma)|\right|.
    \end{equation}
    Because $|g_{\delta}(p;\gamma)|$ and $|g_0(p;\gamma)|$ are both
    continuous for $0 \le p \le 1$, the above maximum is attained, say
    at $p = p^*_{\delta}(\gamma)$.  This means that
    \[
        \Xi_{\delta}(\gamma) = ||g(p^*_{\delta}(\gamma);\gamma)| -
        |g_0(p^*_{\delta}(\gamma);\gamma)||.
    \]

    Now let $\{\delta_n\}_{n=1}^{\infty}$ be a sequence converging to
    $0$ as $n \rightarrow \infty$.  Because $p^*_{\delta_n}(\gamma)$
    is a bounded sequence, it has a convergent subsequence
    $p^*_{\delta_{n_k}}(\gamma)$.  Along this subsequence, 
    \[
        \Xi_{\delta_{n_k}}(\gamma) =
        ||g_{\delta_{n_k}}(p^*_{\delta_{n_k}}(\gamma);\gamma)| -
        |g_0(p^*_{\delta_{n_k}}(\gamma);\gamma)|| 
        \rightarrow 0 \eqtext{as } k \rightarrow
        \infty
    \]
    by Lemma~\ref{lem:g_0}.  In other words, every sequence
    $\Xi_{\delta_n}(\gamma)$ has a subsequence that converges to $0$,
    which implies that every sequence $\Xi_{\delta_n}$ converges to $0$.
    Because the original sequence $\delta_n$ was arbitrary, this implies
    that
    \[
        \lim_{\delta\rightarrow 0^+}\Xi_{\delta}(\gamma) = 0.
    \]
    In combination with \eqref{eqn:Xi_def}, this implies
    that $|g_{\delta}(p;\gamma)|$ converges to
    $|g_0(p;\gamma)|$ \emph{uniformly} in $p$ for $0 \le p \le 1$.
    Thus for every $\epsilon > 0$ there is a $\delta_{\gamma} > 0$ 
    such that 
     \[
         |g_{\delta}(p;\gamma)| \ge |g_0(p;\gamma)| - \epsilon
     \]
     for all $p \in [0,1]$ and all $0 < \delta \le \delta_{\gamma}$.   
     If we take 
     \[
         \epsilon = \frac{1}{2}\min_{0 \le p \le 1} |g_0(p;\gamma)|,
     \]
     then
     \[
         |g_{\delta}(p;\gamma)| \ge \frac{1}{2}|g_0(p;\gamma)| \ge
         C_{\gamma} > 0
     \]
     for all $p \in [0,1]$ (the last two inequalities hold because the
     roots of $|g_0|$ are larger than $1$ by Lemma~\ref{lem:g_0_roots}).
     Combining this result with the first paragraph of the proof gives
     us the bound
     \[
         \int_0^1 \frac{|I_p|^2}{|g_{\delta}(p;\gamma)|^2}
             M_{\delta}(p;\gamma) \di{p}
         \le 
         C\int_0^1 \frac{1}{|g_0(p;\gamma)|^2} \di{p} \le C_{\gamma}
     \]
     for some constant $C_{\gamma} > 0$.
\end{proof}

The preceding lemma proves that we only need to study the integral in
\eqref{eqn:pd} over the interval $1 \le p < \infty$.  Because
$|g_{\delta}(p;\gamma)| \rightarrow |g_0(p;\gamma)|$ as $\delta
\rightarrow 0^+$, it should be the case that $|g_{\delta}(p;\gamma)|
\approx 0$ near the roots of $|g_0(p;\gamma)|$.  Inspired by our earlier
work in the quasistatic regime, we conjecture that
$|g_{\delta}(p^1_{\gamma};\gamma)|$ and
$|g_{\delta}(p^2_{\gamma};\gamma)|$ are on the order of $\delta$ as
$\delta \rightarrow 0^+$. 
\begin{conjecture}\label{con:g_order_delta}
    Suppose $0 < \gamma < \gamma_*$, and let $1 < p^1_{\gamma} <
    p^2_{\gamma}$ be the
    roots of $g_0(p;\gamma)$.  Then there is a $\delta_{\gamma} > 0$
    such that $|g_{\delta}(p;\gamma)| \ne 0$ for all $1 \le p < \infty$
    and all $0 < \delta \le \delta_{\gamma}$; however,
    $|g_{\delta}(p^1_{\gamma};\gamma)| = \mathcal{O}(\delta)$ and
    $|g_{\delta}(p^2_{\gamma};\gamma)| = \mathcal{O}(\delta)$ as $\delta
    \rightarrow 0^+$.
\end{conjecture}
One way to prove this conjecture would be to expand
$|g_{\delta}(p^j_{\gamma};\gamma)|$ (for $j = 1$, $2$) in Taylor series
around $\delta = 0$ and then prove that
$\partial|g_{\delta}(p^j_{\gamma};\gamma)|/\partial \delta$ is uniformly
bounded for $p \in [1,\infty)$ and $\delta$ small enough.
Unfortunately, these derivatives are quite complicated; moreover,
numerical experiments indicate that they become unbounded as $p
\rightarrow \infty$, so it is unlikely that this technique would
work even if the expressions were suitable for analytic study.
To provide some justification for Conjecture~\ref{con:g_order_delta}, in Figures~\ref{fig:g_delta}(a) and (b) we plot
\begin{equation}\label{eqn:g1g2}
    \frac{|g_{\delta}(p^1_{\gamma}; \gamma)|}{\delta} \eqtext{and}
    \frac{|g_{\delta}(p^1_{\gamma}; \gamma)|}{\delta}
\end{equation}
as functions of $\delta$ and $\gamma$ over the ranges $10^{-12}\le
\delta \le 10^{-10}$ and $0.1\gamma_* \le \gamma \le
0.99\gamma_*$\footnote{We believe the functions in \eqref{eqn:g1g2}
    remain bounded as $\delta \rightarrow 0$ for all $0 < \gamma <
    \gamma_*$; however, $p^2_{\gamma} \rightarrow \infty$ as $\gamma
    \rightarrow 0$, so the numerical computation of the roots becomes
    more difficult as $\gamma$ gets closer to $0$.  Similarly,
$p^1_{\gamma_*} = p^2_{\gamma_*}$, so as $\gamma$ gets close to
$\gamma_*$ it becomes difficult to distinguish the roots}.  For
each $\gamma$, we see that the functions in \eqref{eqn:g1g2} remain
bounded as $\delta$ gets close to $0$, which seems to indicate that
$|g_{\delta}(p^1_{\gamma}; \gamma)| = \mathcal{O}(\delta)$ and
$|g_{\delta}(p^2_{\gamma}; \gamma)| = \mathcal{O}(\delta)$ as $\delta
\rightarrow 0$.  Curiously, both functions in \eqref{eqn:g1g2} seem to
depend very weakly on $\delta$.
\begin{figure}[!hbp]
    \begin{center}
        \begin{tabular}{c c}
            \includegraphics[width=0.45\textwidth]{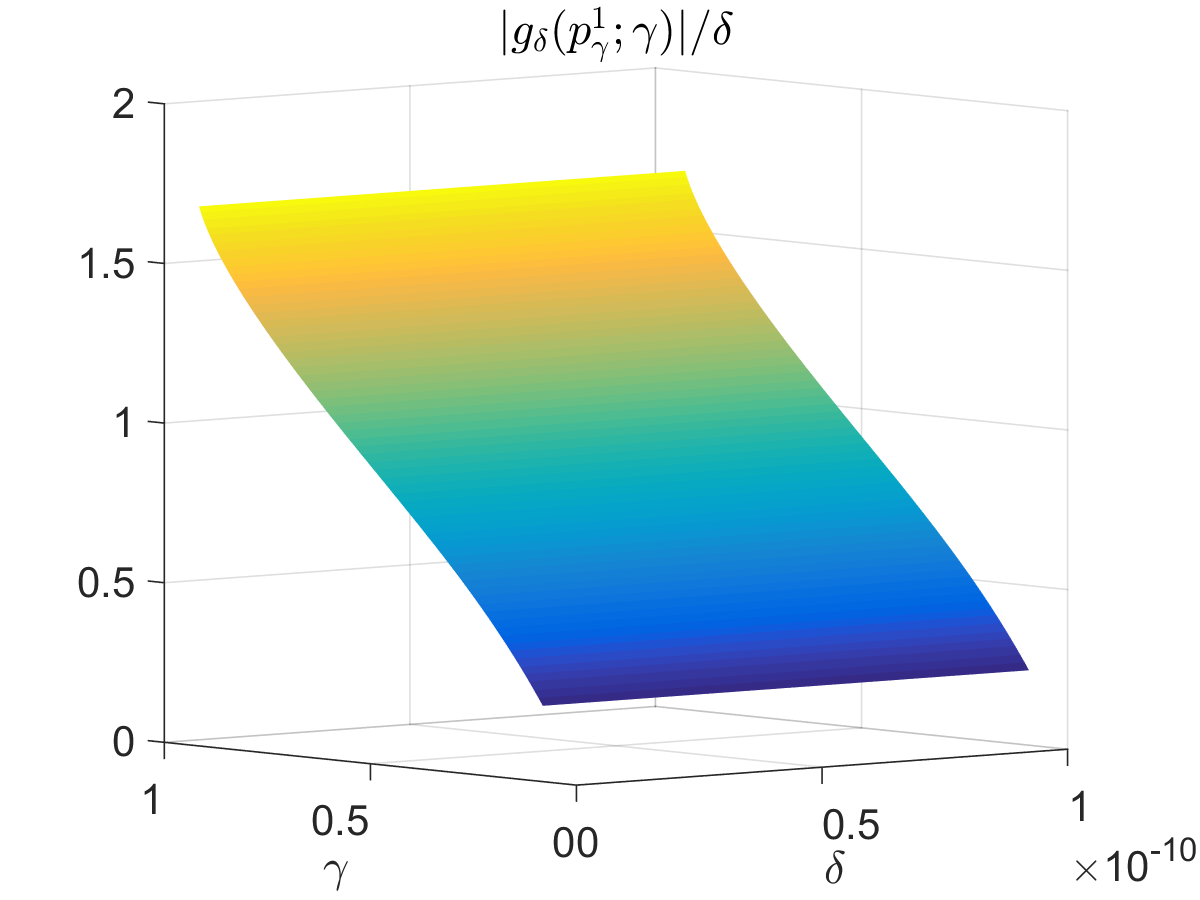}
            &
            \includegraphics[width=0.45\textwidth]{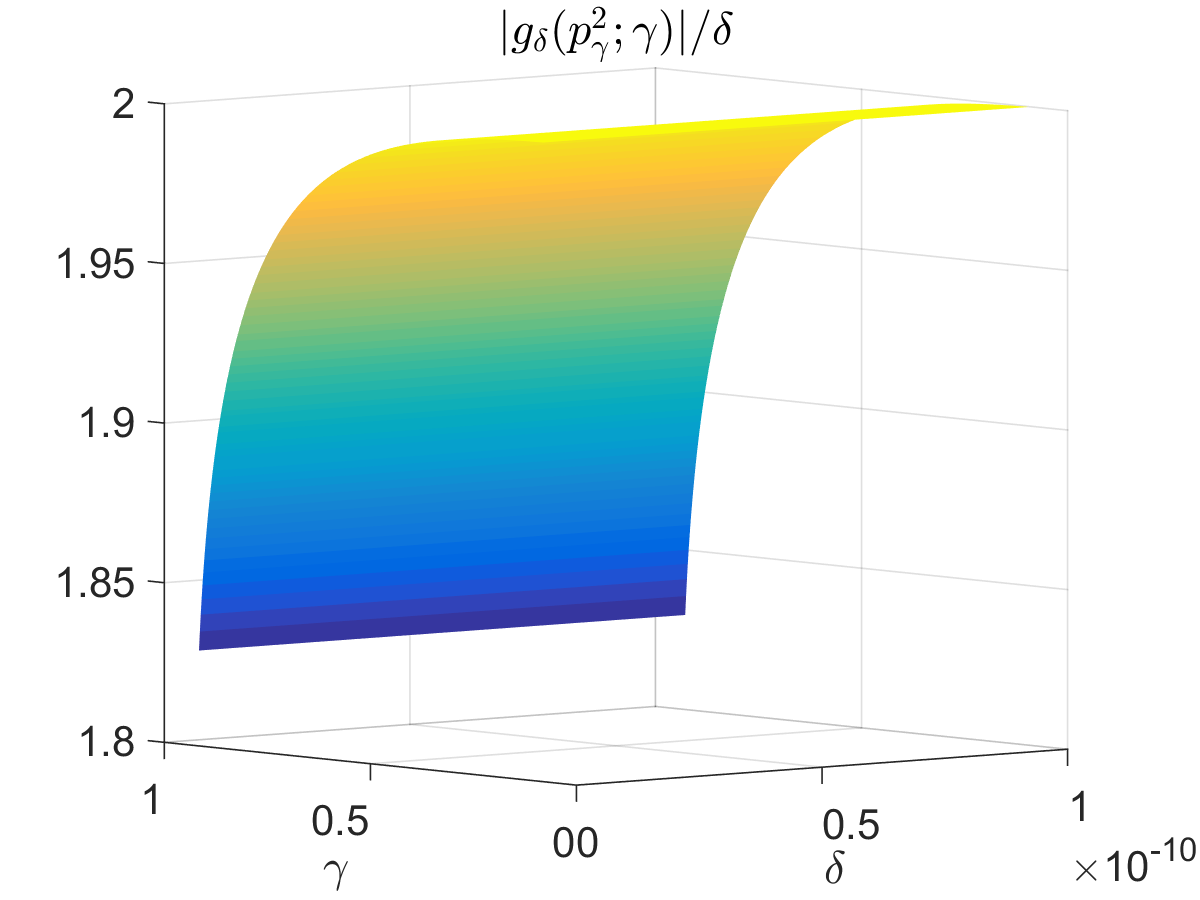}
            \\
            (a) & (b)
        \end{tabular}
        \caption{\emph{In this figure we plot (a)
                $|g_{\delta}(p^1_{\gamma};\gamma)|/\delta$ and (b)
                $|g_{\delta}(p^2_{\gamma};\gamma)|/\delta$ over the
                range $10^{-12} \le \delta \le 10^{-10}$ and
        $0.1\gamma_* \le \gamma \le 0.99\gamma_*$.}}
        \label{fig:g_delta}
    \end{center}
\end{figure}

Next, we conjecture that the $\mathcal{O}(\delta)$ behavior of
$|g_{\delta}(p;\gamma)|$ near $p^1_{\gamma}$ and $p^2_{\gamma}$ is not
canceled by the term $M_{\delta}(p;\gamma)$ in the numerator.
\begin{conjecture}\label{con:M_not_zero}
    Suppose $0 < \gamma < \gamma_*$, and define $M_{\delta}(p;\gamma)$
    as in \eqref{eqn:M}.  Then there exist positive constants
    $\delta_{\gamma}$ and $C_{\gamma}$ such that $M_{\delta}(p;\gamma)
    \ge C_{\gamma}$ near $p^1_{\gamma}$ and $p^2_{\gamma}$ for all $0 <
    \delta \le \delta_{\gamma}$.
\end{conjecture}
If Conjectures~\ref{con:g_order_delta} and \ref{con:M_not_zero} are
true, then \eqref{eqn:pd_a}--\eqref{eqn:L} imply that the part of the integrand $L_{\delta}(p;\gamma)$ that is
independent of the source $f$, namely
\[
    \ee^{2\gamma\nu_m'}\frac{M_{\delta}(p;\gamma)}
    {|g_{\delta}(p;\gamma)|^2},
\]
is on the order of $\delta^{-2}$ near $p^1_{\gamma}$ and $p^2_{\gamma}$
as $\delta \rightarrow 0^+$.  If $|I_p|^2$ is also bounded away from $0$
near $p^1_{\gamma}$ and $p^2_{\gamma}$, the entire integrand
$L_{\delta}(p;\gamma)$ will have values on the order of $\delta^{-2}$
near $p^1_{\gamma}$ and $p^2_{\gamma}$. 

To provide some justification for Conjecture~\ref{con:M_not_zero}, in
Figures~\ref{fig:M_delta}(a) and (b) we plot 
\[
    M_{\delta}(p^1_{\gamma};\gamma) \eqtext{and} 
    M_{\delta}(p^2_{\gamma};\gamma)
\]
as functions of $\delta$ and $\gamma$ over the same intervals as in
Figure~\ref{fig:g_delta}.  In particular, we note that
$M_{\delta}(p^1_{\gamma};\gamma)$ and $M_{\delta}(p^2_{\gamma};\gamma)$
are both bounded away from $0$ and seem to depend quite weakly on
$\delta$.  
\begin{figure}[!hbp]
    \begin{center}
        \begin{tabular}{c c}
            \includegraphics[width=0.45\textwidth]{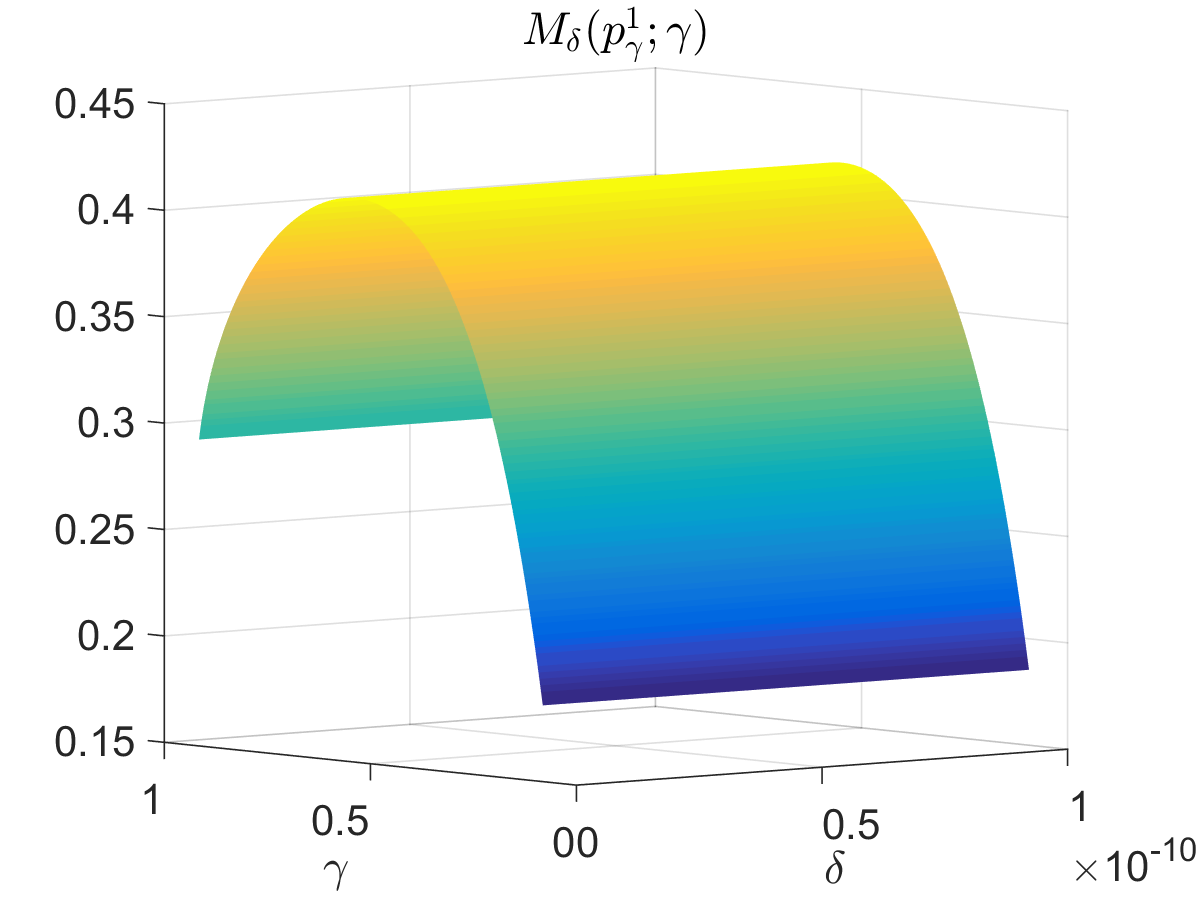}
            &
            \includegraphics[width=0.45\textwidth]{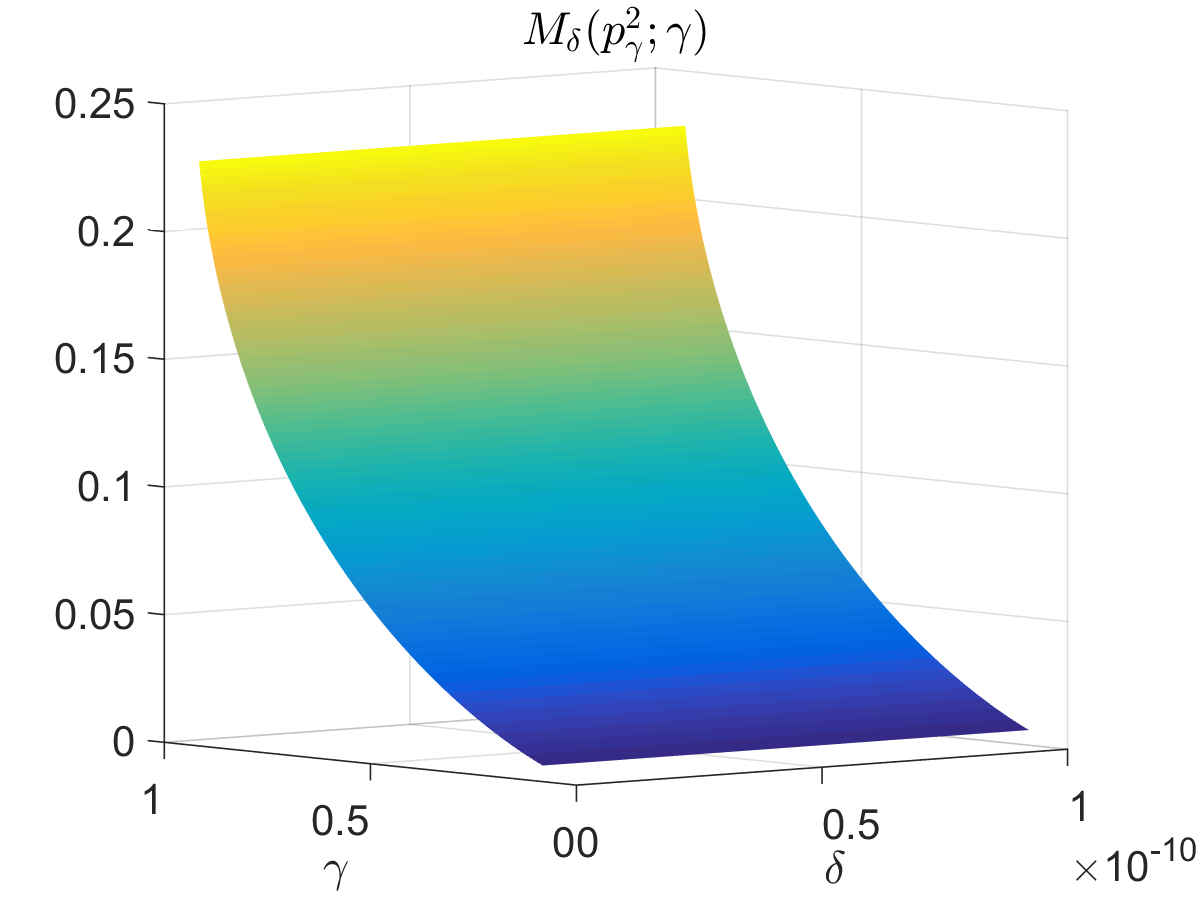}
            \\
            (a) & (b)
        \end{tabular}
        \caption{\emph{In this figure we plot (a)
                $M_{\delta}(p^1_{\gamma};\gamma)$ and (b)
                $M_{\delta}(p^2_{\gamma};\gamma)$ over the
                range $10^{-12} \le \delta \le 10^{-10}$ and
        $0.1\gamma_* \le \gamma \le 0.99\gamma_*$.}}
        \label{fig:M_delta}
    \end{center}
\end{figure}

Finally, to obtain a blow-up in $E_{\delta}(a)$, it should be the
case that $|I_p|$ does not conquer the small values of $|g_{\delta}|$
near $p^1_{\gamma}$ and $p^2_{\gamma}$.   Heuristically, there will be
no blow-up if $|I_p| \approx 0$ near
$p^1_{\gamma}$ and $p^2_{\gamma}$.  In the next section, we present
numerical evidence that suggests that sources with $|I_{p^1_{\gamma}}| =
|I_{p^2_{\gamma}}| = 0$ do not lead to ALR.  

On the other hand, recall from \eqref{eqn:Iq} that
\begin{equation*}
    I_p = \int_{d_0}^{d_1} \fhat(x,k_0p)\ee^{-k_0\sqrt{p^2-1}s} \di{s}.
\end{equation*}
Again we take our inspiration from the quasistatic case
\cite{Thaler:2014:BVI, Meklachi:2016:SAL}.
If $d_0 \gg a$, then the exponential in the above integrand will be
extremely small (especially because $p_{\gamma}^1$ and $p_{\gamma}^2$
are both greater than $1$).  In particular, the exponential may be small
enough so that it cancels out the effect of the denominator near
$p^1_{\gamma}$ and $p^2_{\gamma}$.  We emphasize that this is not
rigorous, but we hope that it may provide a starting point for future
investigations. 
\begin{conjecture}\label{con:pd_blow_up}
    Suppose $0 < \gamma < \gamma_*$.  Then there exist sources $f\!\in\! L^2(\mathcal{M})$ with compact support or distributional such as
    dipoles) such that, for any $0 < \xi \le a$,
    $E_{\delta}(\xi)\rightarrow \infty$ if $d_0$ is ``close enough'' to
    $a$ and $E_{\delta}(\xi) \le C_{\gamma}$ for some positive constant
    $C_{\gamma}$ if $d_0$ is ``far enough away'' from $a$.  This
    critical distance may depend on $\gamma$.

    Moreover, there are positive constants $b_{\gamma}$, $C_{\gamma}$,
    and $\delta_{\gamma}$ such that, for all $0 < \delta \le
    \delta_{\gamma}$,
    \[
        |V(x,y)| \le C_{\gamma}
    \]
    for all $(x,y) \in \mathcal{C}\cup\mathcal{M}$ with $|x| >
    b_{\gamma}$.
\end{conjecture}
\begin{remark}\label{rem:weak}
    If it is only the case that
    \[
        \limsup_{\delta\rightarrow 0^+} E_{\delta}(\xi) = \infty,
    \]
    then we say that \emph{weak} ALR occurs.  Because $E_{\delta}(\xi)$
    is difficult to deal with analytically, we cannot say much more on
    this.  It is difficult to determine whether
    \[
        \limsup_{\delta\rightarrow 0^+} E_{\delta}(\xi) = \infty
        \eqtext{or}
        \lim_{\delta\rightarrow 0^+} E_{\delta}(\xi) = \infty
    \]
    using only numerical techniques.  In particular, if the limit
    supremum of $E_{\delta}(\xi)$ is $\infty$, there is at least one
    sequence $\delta_n\rightarrow 0^+$ along which
    $E_{\delta_n}(\xi)\rightarrow \infty$; however, it may be the case
    that $E_{\delta_n}(\xi) \rightarrow \infty$ for all sequences
    $\delta_n \rightarrow 0^+$ except
    a few very special sequences that would be extremely difficult to
    find via numerical experiments alone.
\end{remark}

Figures~\ref{fig:dipole_small_gamma} and \ref{fig:general_small_gamma}
are exactly the same as Figures~\ref{fig:dipole_large_gamma} and
\ref{fig:general_large_gamma} except $\gamma = 0.5\gamma_*$ in
Figures~\ref{fig:dipole_small_gamma} and \ref{fig:general_small_gamma}.
In Figures~\ref{fig:dipole_small_gamma}(a) and (b) and
Figures~\ref{fig:general_small_gamma}(a) and (b), the sources (a dipole
in Figure~\ref{fig:dipole_small_gamma} and the source $f$ from
\eqref{eqn:general} in Figure~\ref{fig:general_small_gamma}) are located
at $d_0 = 4a$, and the solution $V$ appears to be smooth throughout the
domain.  As the sources move closer to the slab, resonant regions appear
around both boundaries of the slab at $x = 0$ and $x = a$.
Figures~\ref{fig:dipole_small_gamma}(c) and (d) and
Figures~\ref{fig:general_small_gamma}(c) and (d) contain plots of $V$ when
$d_0 = 1.2a$.  From these figures we see that the extreme oscillations
of $V$ are contained near the boundaries of the slab, and that the
boundaries between the resonant and nonresonant regions are sharp and
not defined by the boundaries of the slab; away from the slab, $V$ is
smooth and bounded.  This is highly
characteristic of ALR (see, e.g., \cite{Milton:2005:PSQ,
Thaler:2014:BVI}, and the references therein).
Moreover, Figures~\ref{fig:dipole_small_gamma} and
\ref{fig:general_small_gamma} indicate that an image of (part of) the
solution $V$ is focused in the region to the left of the lens (outside
of the resonant region); this is in stark contrast to the high frequency
regime illustrated in Figures~\ref{fig:dipole_large_gamma} and
\ref{fig:general_large_gamma}, in which the solution $V$ in the region
to the left of the slab is barely noticeable.  Indeed, in the
quasistatic regime, ALR is closely associated with this so-called
superlensing \cite{Milton:2005:PSQ}; since ALR does not occur for $\gamma >
\gamma_*$ (see Theorem~\ref{thm:bounded}), we do not expect to see the
superlensing effect in this regime (see Theorem~\ref{thm:shielding}).

Figures~\ref{fig:dipole_small_gamma}(c) and
\ref{fig:general_small_gamma}(c) provide an additional insight into
Conjecture~\ref{con:pd_blow_up}.  In general, for $q \approx
k_0p^2_{\gamma}$ (where $p^2_{\gamma}$ is the larger root of
$g_0(p;\gamma)$) the coefficient $A_q$ from \eqref{eqn:Aq} becomes very large
since its denominator is proportional to $g_{\delta}(p;\gamma)$
and $g_{\delta}(p^2_{\gamma};\gamma) \approx g_0(p^2_{\gamma}; \gamma)
= 0$ for $\delta$ small enough.  Recalling that the Fourier transform
variable $q = k_0p$ represents a wavenumber in the $y$-direction with
corresponding wavelength $\lambda = 2\pi/q$, this
implies that the solution $V$ should exhibit prominent oscillations with
wavelength on the order of 
\begin{equation}\label{eqn:lambda_gamma}
    \lambda_{\gamma} = \frac{2\pi}{k_0p^2_{\gamma}}.
\end{equation}
In Figures~\ref{fig:dipole_small_gamma}(c) and
\ref{fig:general_small_gamma}(c), we have drawn a vertical,
red line of length $2\lambda_{\gamma}$.  This red line covers
approximately $2$ wavelengths of oscillation in the resonant region,
which seems to indicate that at least one of the zeros of $g_0$, namely
$p^2_{\gamma}$, is responsible for ALR.  Because $p^2_{\gamma}$ is
independent of $f$, the above argument also suggests that the wavelength
of the resonant oscillations of $V$ is also independent of the source
$f$.  We emphasize that this is speculative at best, but it would be
interesting to investigate further.
\begin{figure}[!hbp]
    \begin{center}
        \begin{tabular}{c c}
            \includegraphics[width=0.45\textwidth]{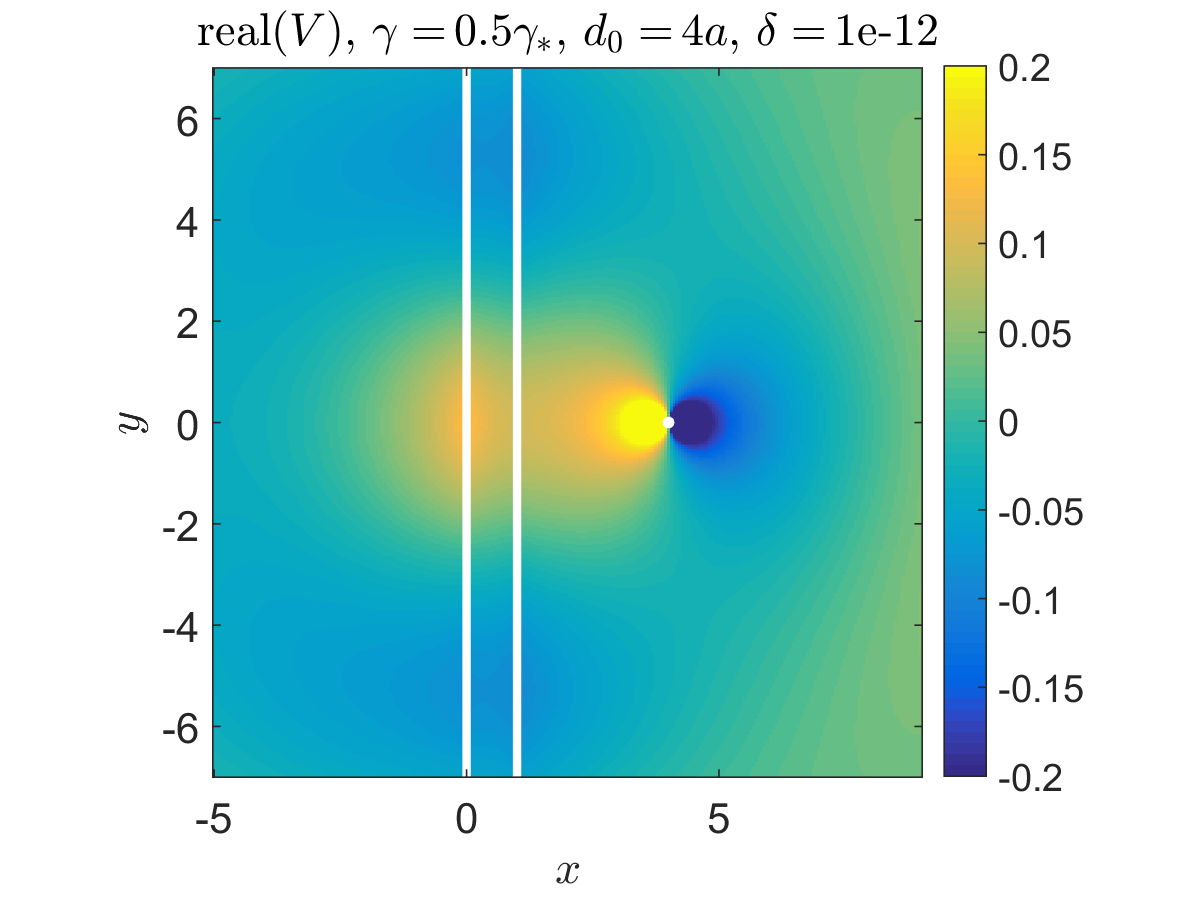}
            & \includegraphics[width=0.45\textwidth]{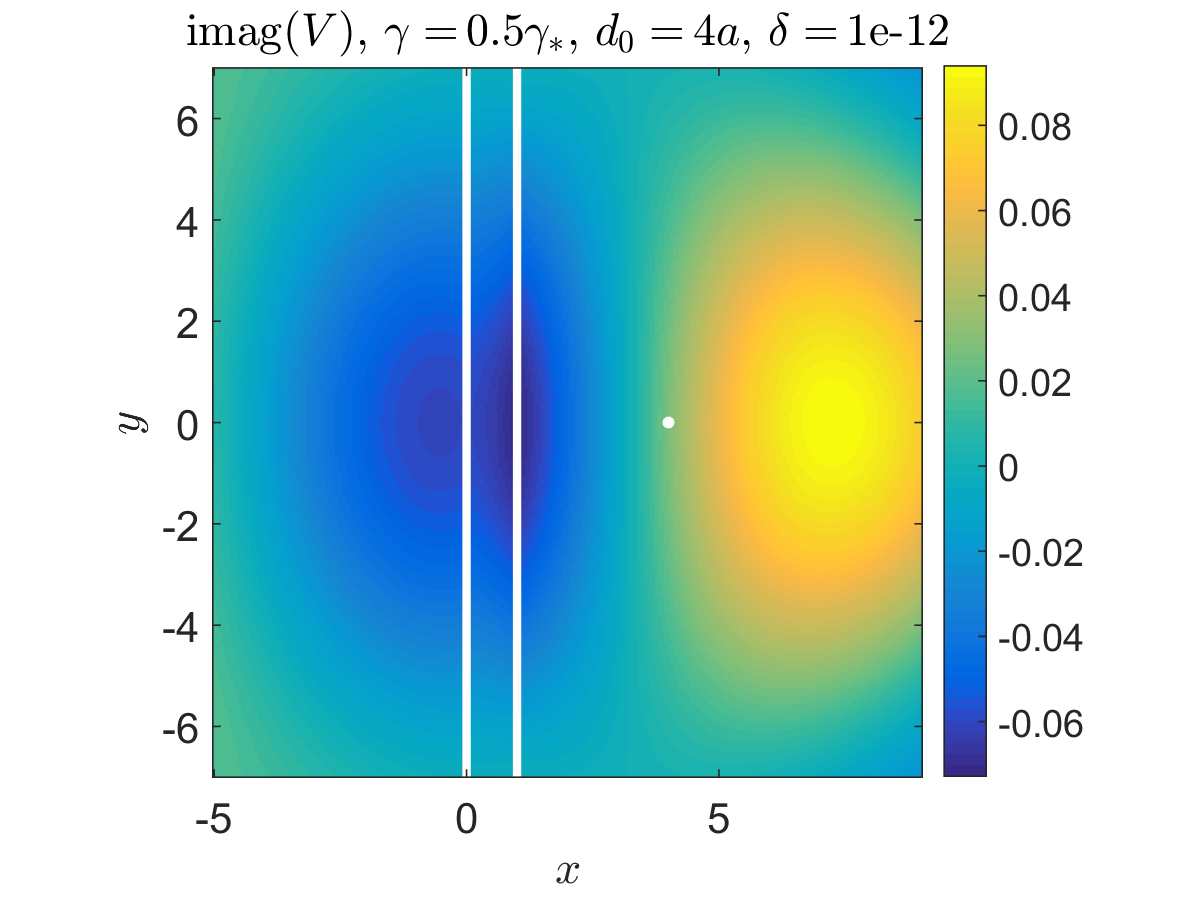}
            \\
            (a) & (b) \myspace{}
            \includegraphics[width=0.45\textwidth]{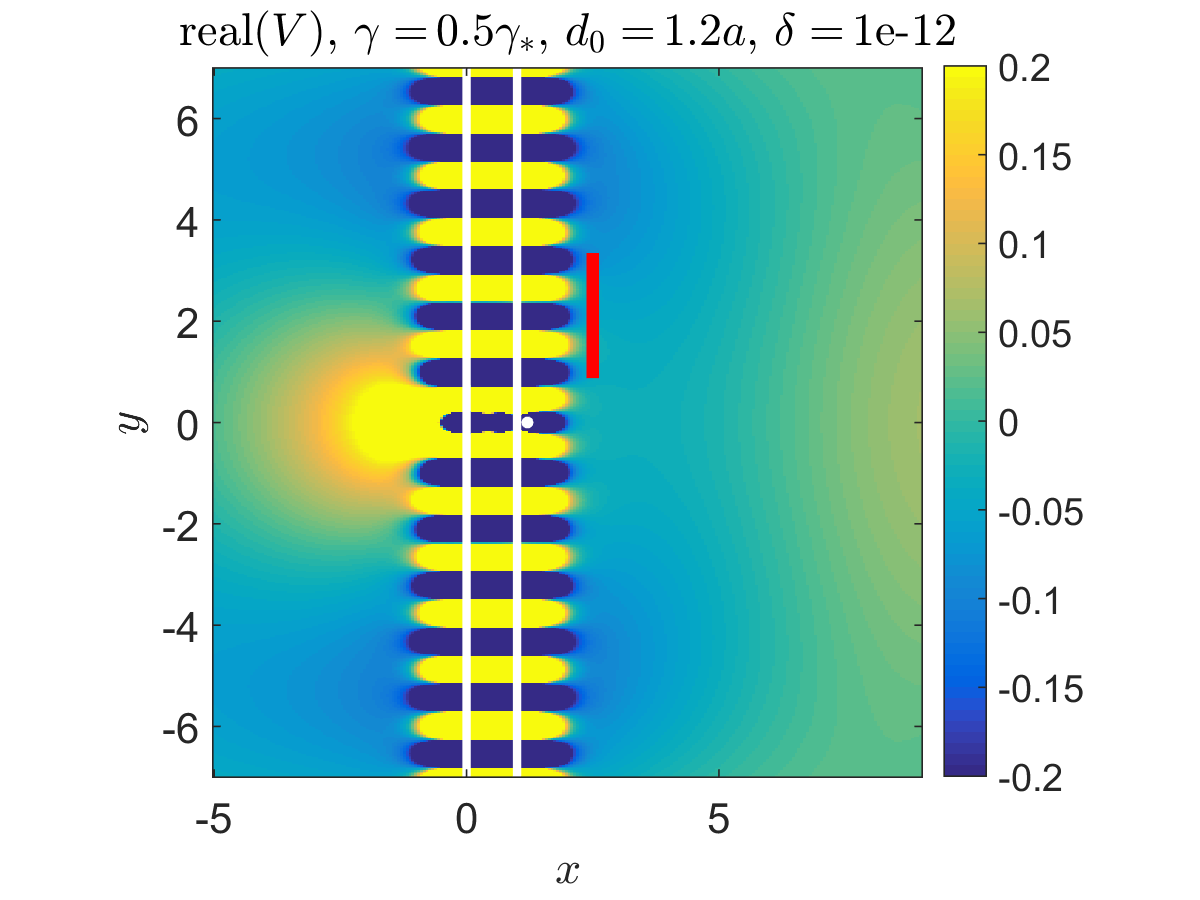}
            & \includegraphics[width=0.45\textwidth]{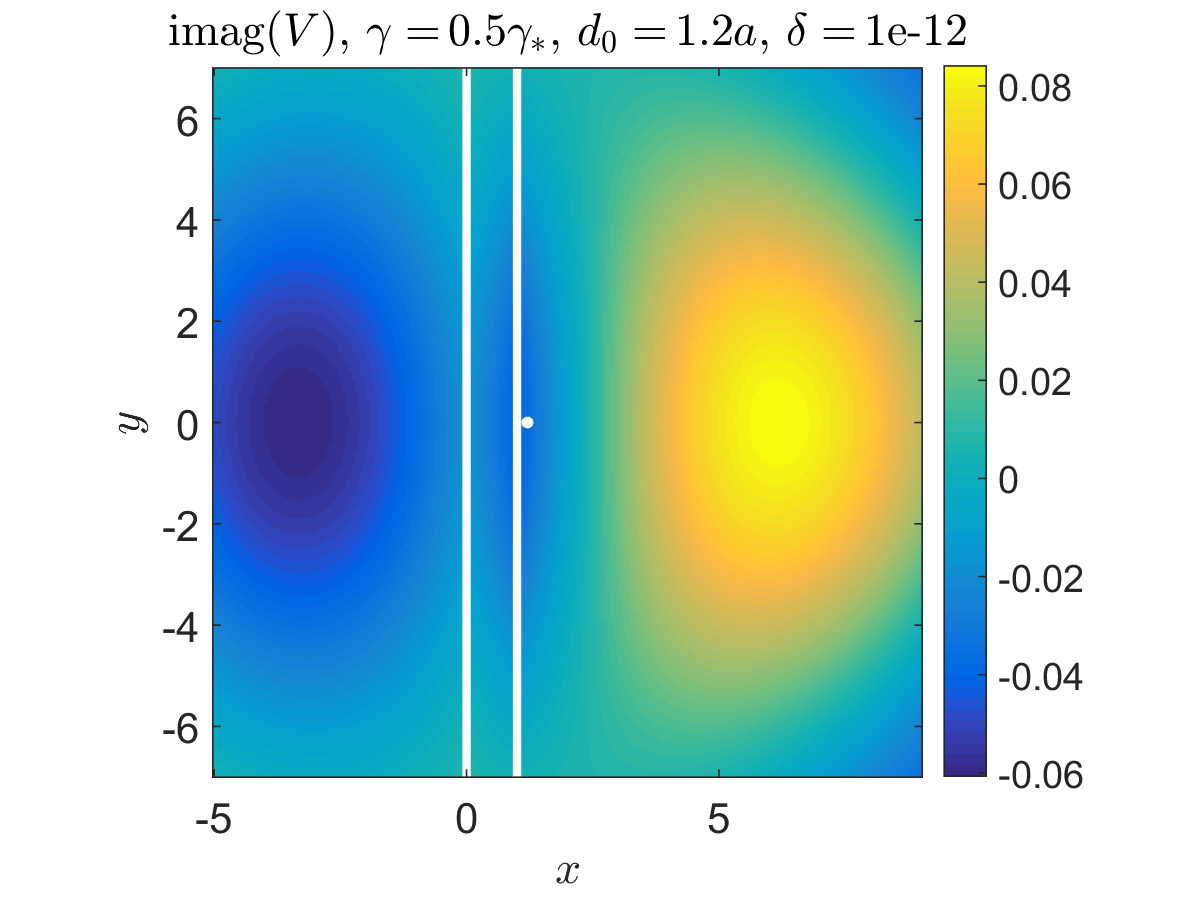}
            \\
            (c) & (d)
        \end{tabular}
        \caption{\emph{This is a plot of $V$, the solution to
                \eqref{eqn:finite_freq}, when $f$ is a dipole and
                $\gamma = 0.5\gamma_*$: (a) $\real(V)$ and (b) $\imag(V)$
                for $d_0 = 4a$; (c) $\real(V)$ and (d) $\imag(V)$ for
                $d_0 = 1.2a$.  To make the behavior of $V$ more clear,
                we clipped the maximum and minimum values in each plot 
                to $0.2$ (yellow) and $-0.2$ (blue) respectively. The
            vertical, red line in (c) extends a distance of
        $2\lambda_{\gamma}$, where $\lambda_{\gamma}$ is defined in
\eqref{eqn:lambda_gamma}.}}
        \label{fig:dipole_small_gamma}
    \end{center}
\end{figure}
\begin{figure}[!hbp]
    \begin{center}
        \begin{tabular}{c c}
            \includegraphics[width=0.45\textwidth]{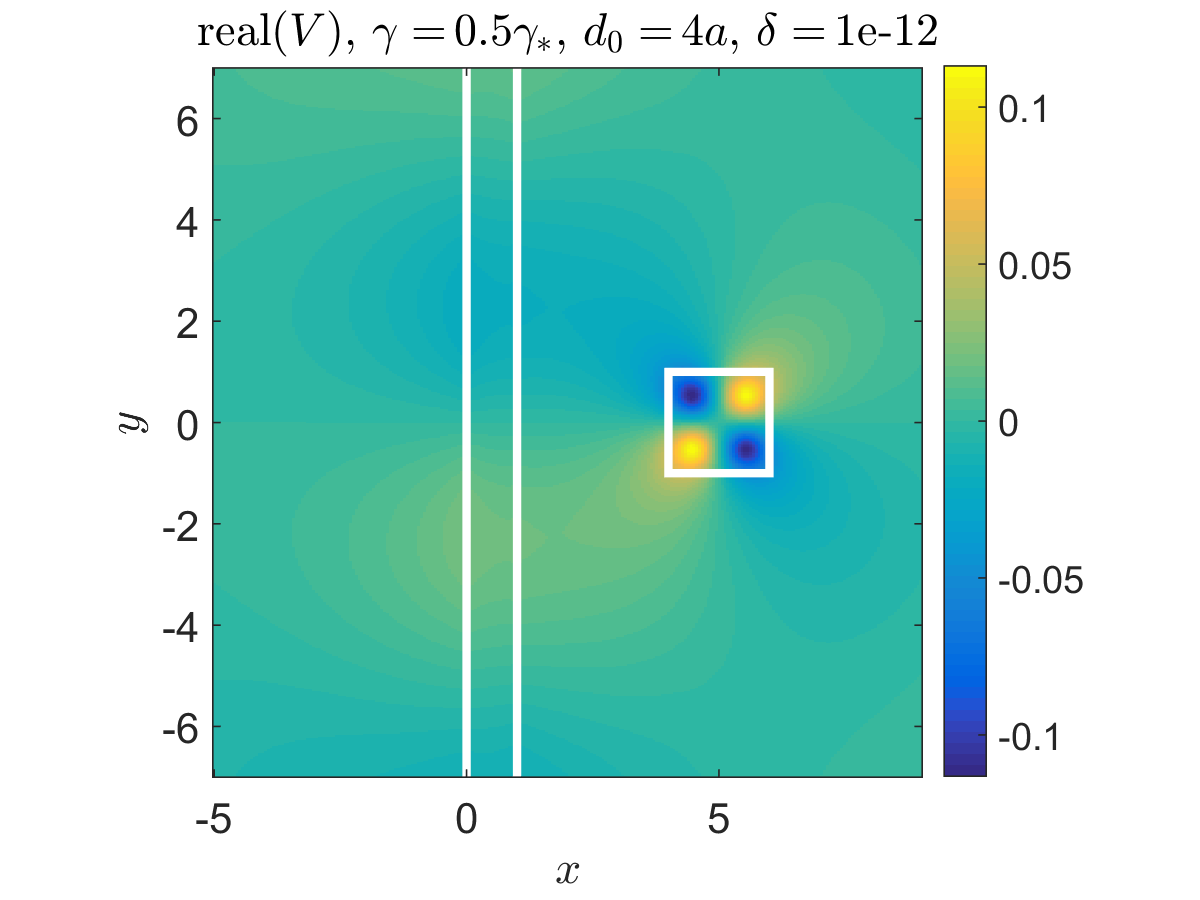}
            & \includegraphics[width=0.45\textwidth]{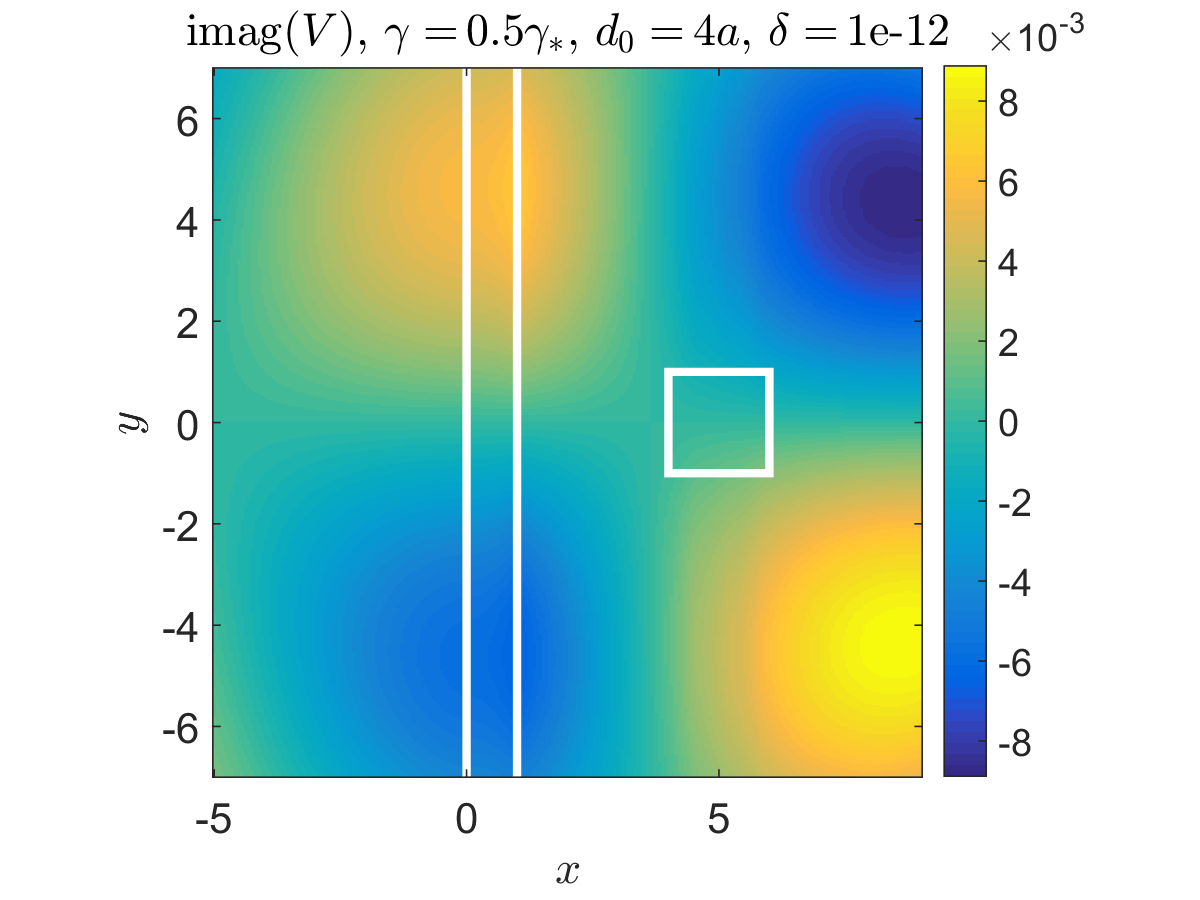}
            \\
            (a) & (b)\myspace{}
            \includegraphics[width=0.45\textwidth]{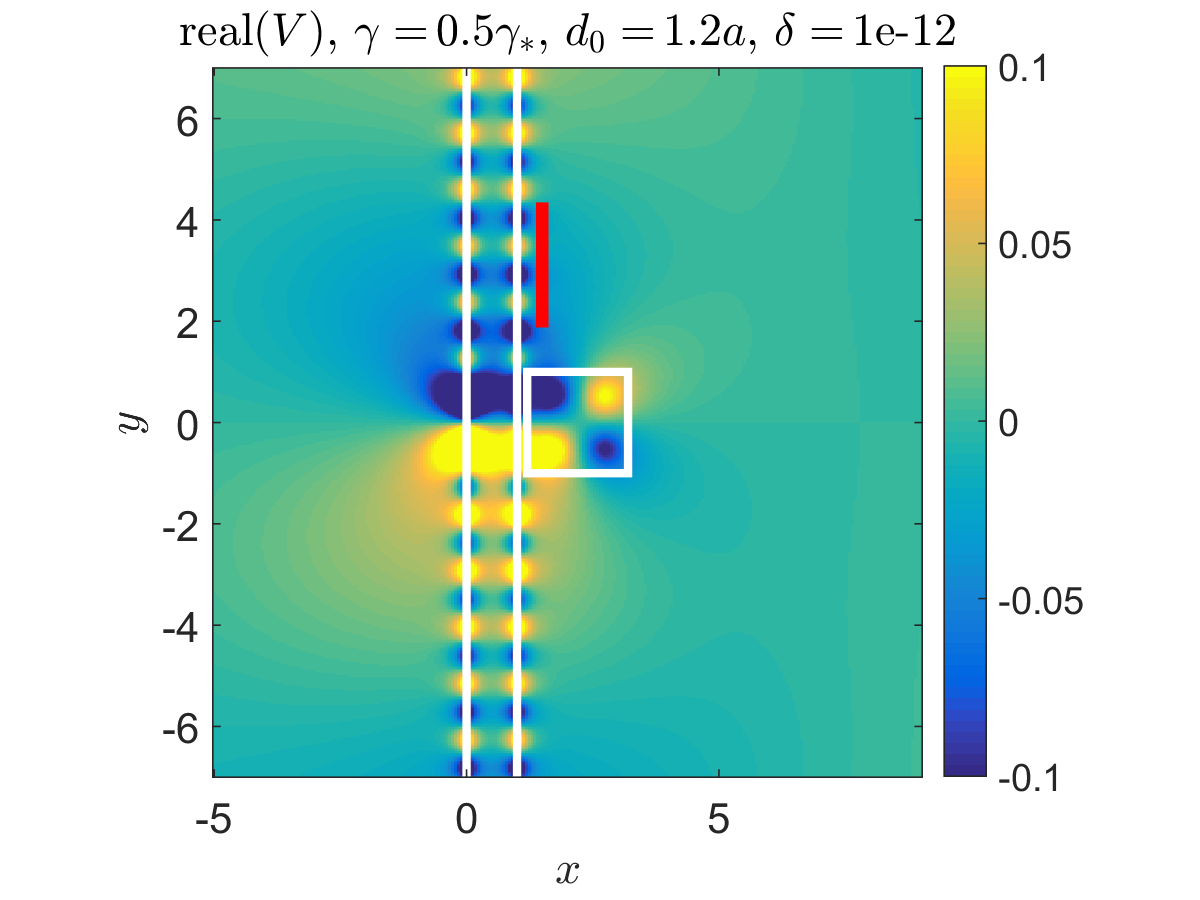}
            & \includegraphics[width=0.45\textwidth]{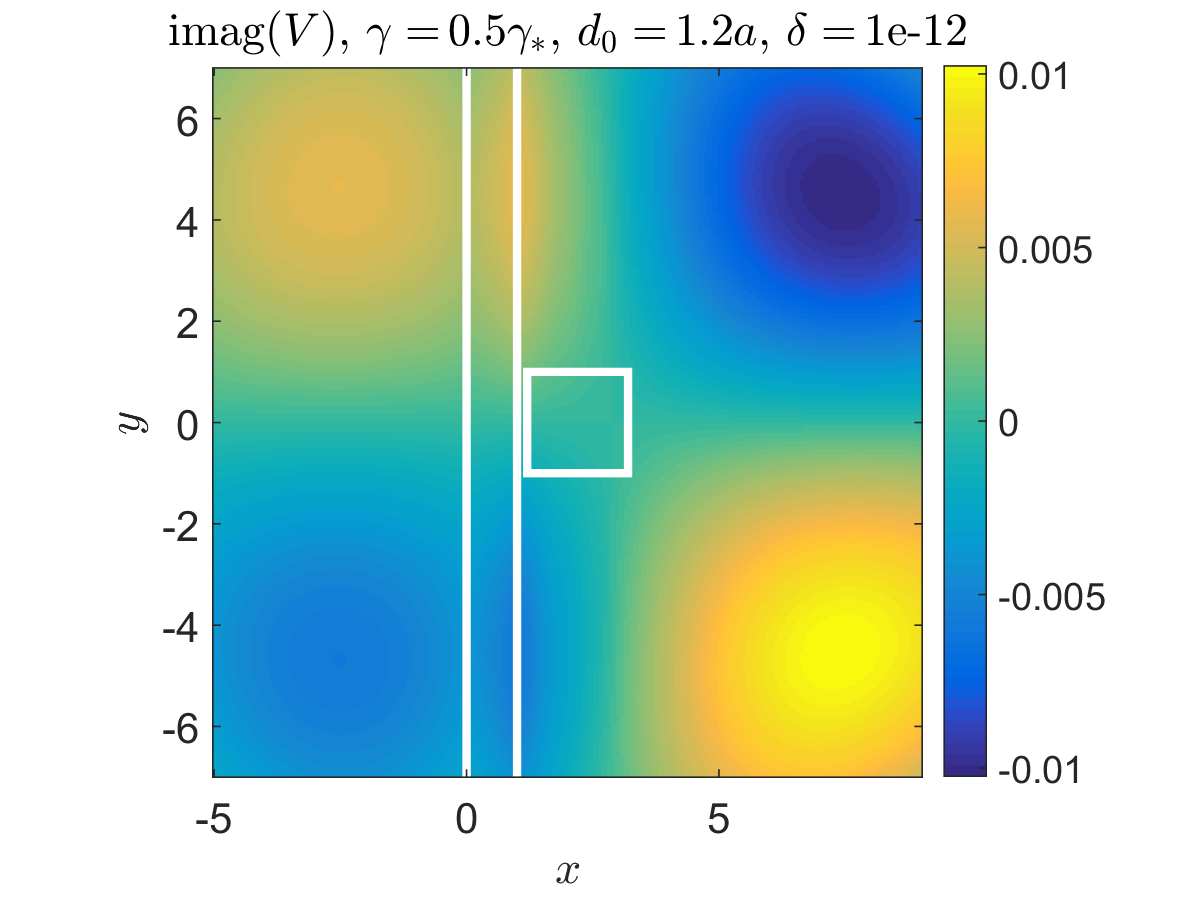}
            \\
            (c) & (d)
        \end{tabular}
        \caption{\emph{This is a plot of $V$, the solution to
                \eqref{eqn:finite_freq}, when $f$ is the function in
                \eqref{eqn:general} and
                $\gamma = 0.5\gamma_*$: (a) $\real(V)$ and (b) $\imag(V)$
                for $d_0 = 4a$; (c) $\real(V)$ and (d) $\imag(V)$ for
                $d_0 = 1.2a$.  To make the behavior of $\real V$ more
                clear, in (a) and (c)
                we clipped the maximum and minimum values in each plot 
                to $0.1$ (yellow) and $-0.1$ (blue) respectively.  The
            vertical, red line in (c) extends a distance of
        $2\lambda_{\gamma}$, where $\lambda_{\gamma}$ is defined in
\eqref{eqn:lambda_gamma}.}}
        \label{fig:general_small_gamma}
    \end{center}
\end{figure}

To illustrate how drastically different the behavior of $V$ is for
$\gamma > \gamma_*$ and $\gamma < \gamma_*$, in
Figure~\ref{fig:dipole_compare} we plotted $V$ corresponding to a dipole
source located at $d_0 = 1.2a$ for two different values of $\gamma$.  In
Figures~\ref{fig:dipole_compare}(a) and (b) we took $\gamma = 1.01\gamma_*$
while in Figures~\ref{fig:dipole_compare}(c) and (d) we took $\gamma =
0.99\gamma_*$.  The ALR is present when $\gamma < \gamma_*$ in
Figures~\ref{fig:dipole_compare}(c); on the other hand, in
Figure~\ref{fig:dipole_compare}(a) there are a few oscillations near
the $x$-axis, but they quickly die out as $|y|$ grows.
\begin{figure}[!hbp]
    \begin{center}
        \begin{tabular}{c c}
            \includegraphics[width=0.45\textwidth]{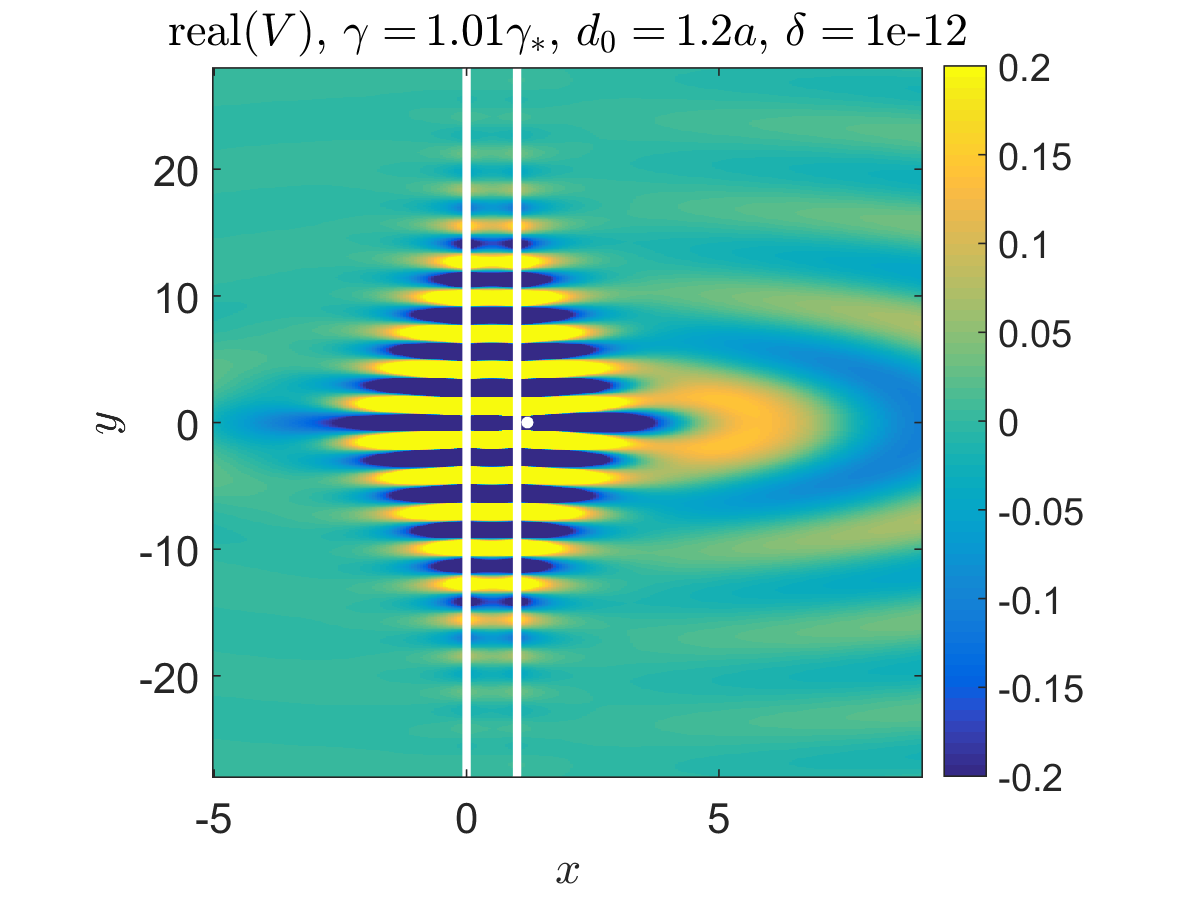}
            &
            \includegraphics[width=0.45\textwidth]{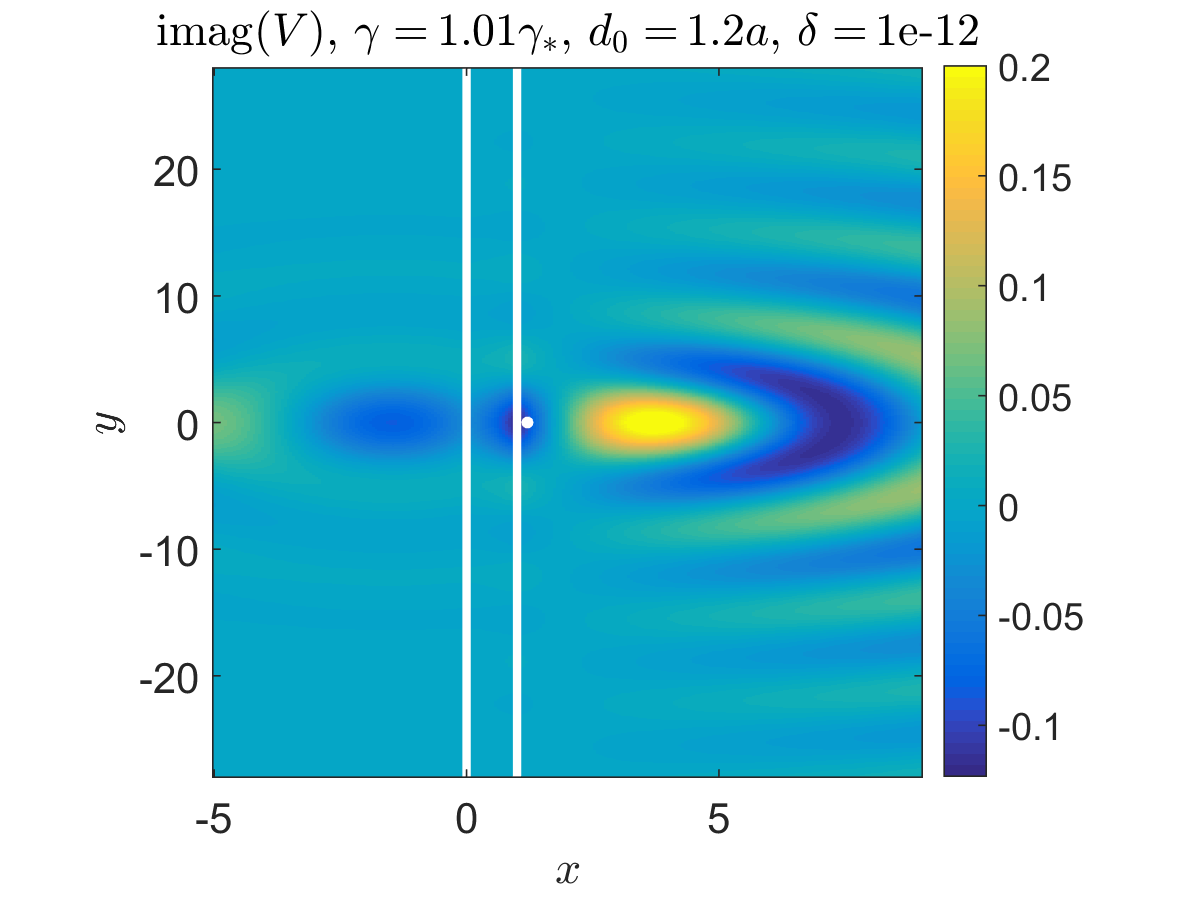}
            \\
            (a) & (b) \myspace{}
            \includegraphics[width=0.45\textwidth]{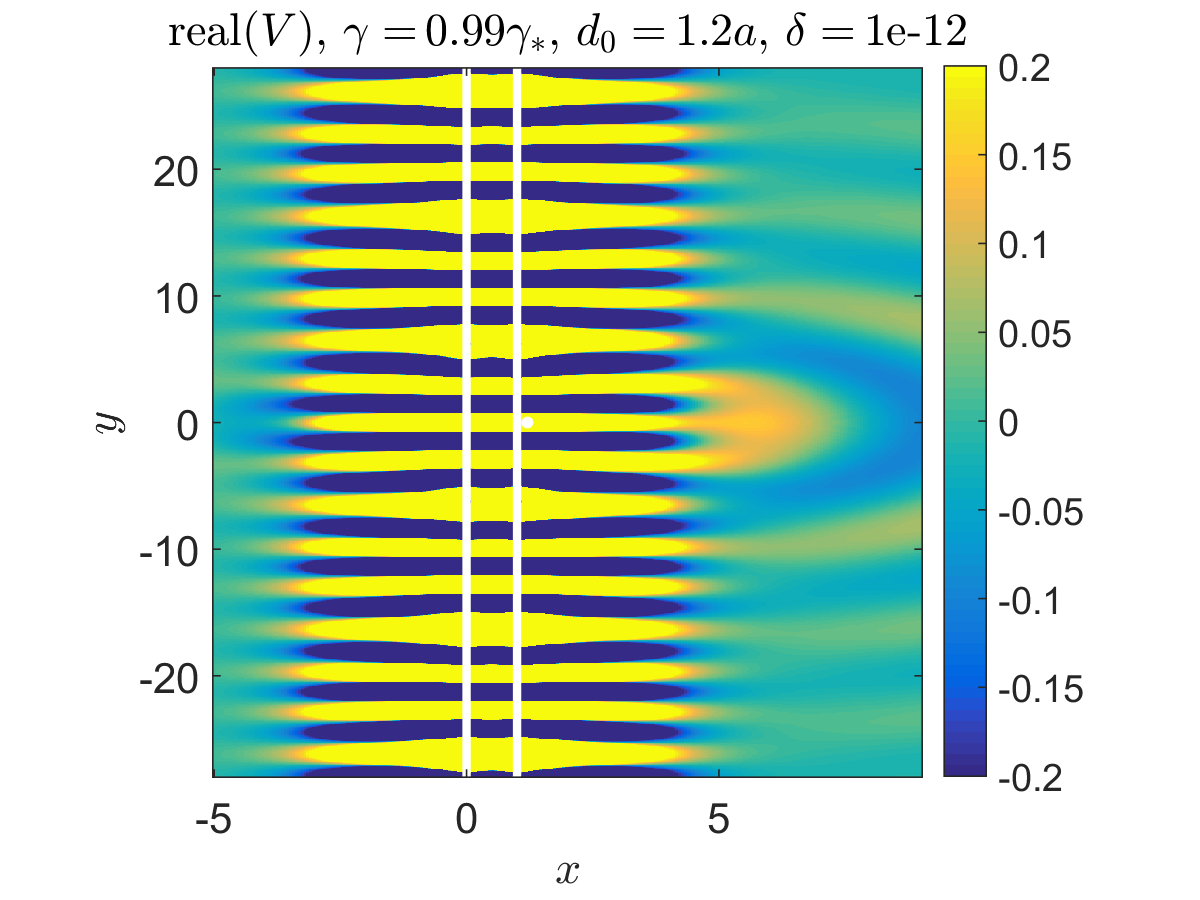}
            &
            \includegraphics[width=0.45\textwidth]{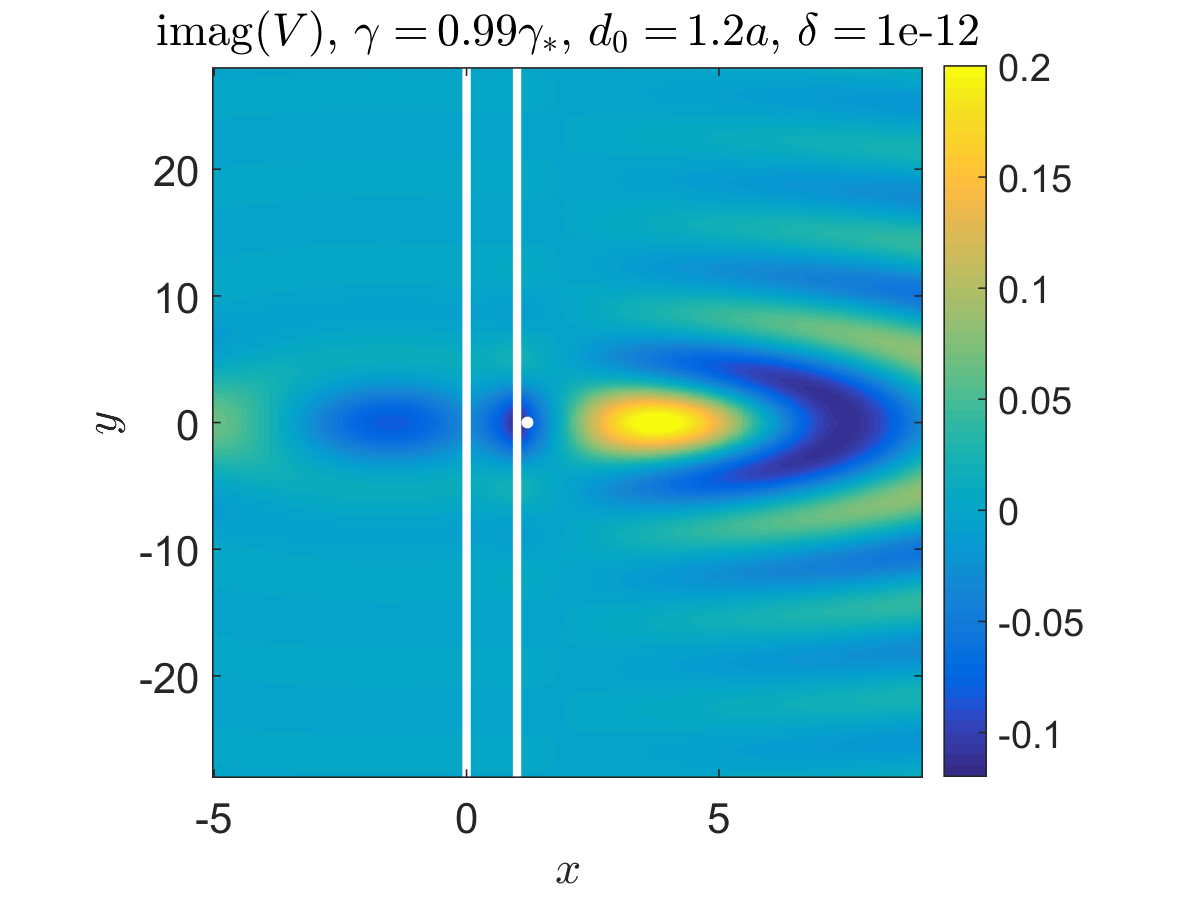}
            \\
            (c) & (d)
        \end{tabular}
        \caption{\emph{This is a plot of $V$, the solution to
                \eqref{eqn:finite_freq}, when $f$ is a dipole located at
                $d_0 = 1.2a$: (a) $\real(V)$ and (b) $\imag(V)$ for
                $\gamma = 1.01\gamma_*$; (c) $\real(V)$ and (d)
                $\imag(V)$ for $\gamma = 0.99\gamma_*$.  To make the
                behavior of $V$ more clear, in (a), (c), and (d) we clipped the maximum and minimum
values in each plot to $0.1$ (yellow) and $-0.1$ (blue) respectively.}}
            \label{fig:dipole_compare}
    \end{center}
\end{figure}

Unfortunately, we cannot provide an figure analogous to
Figure~\ref{fig:dipole_bounded_pd} for $E_{\delta}(a)$ when $f$ is a
dipole source --- MATLAB is unable to accurately compute the integral
\[
    \int_{1}^{\infty} L_{\delta}(p;\gamma) \di{p}
\]
because $|g_{\delta}(p;\gamma)|$ is very close to $0$ near the roots of
$g_0(p;\delta)$ for small values of $\delta$ (see
Conjecture~\ref{con:g_order_delta}).  However, to get a sense of what is
going on, we plotted
$L_{\delta}(p;\gamma)$ on a logarithmic scale for a dipole source $f$
with $\gamma =
0.99\gamma_*$ in Figures~\ref{fig:dipole_integrand}(a) and (b) and $\gamma =
1.01\gamma_*$ in Figures~\ref{fig:dipole_integrand}(c) and (d).  Each curve is
$L_{\delta}(p;\gamma)$ as a function of $p$ for various values of
$\delta$.  In Figures~\ref{fig:dipole_integrand}(a) and (b), where $\gamma <
\gamma_*$, we see that $L_{\delta}(p;\gamma)$ is quite large near
the poles of $g_0(p;\gamma)$, even if $\delta = 10^{-4}$.
Additionally, on comparing the $y$-axis scales in Figures~\ref{fig:dipole_integrand}(a) and
(b), we note that the poles seem somewhat less severe in
Figure~\ref{fig:dipole_integrand}(a) than in
Figure~\ref{fig:dipole_integrand}(b), which, in combination with results
from the quasistatic regime \cite{Meklachi:2016:SAL}, lends credence to our conjecture
(Conjecture~\ref{con:pd_blow_up})
that ALR may be present only if the source is located close enough to
the lens.  On the other
hand, in Figures~\ref{fig:dipole_integrand}(c) and (d), $\gamma > \gamma_*$ and
we see that $L_{\delta}(p;\gamma)$ remains bounded regardless of $d_0$\footnote{In
    Figure~\ref{fig:dipole_integrand}, all of the functions rapidly tend
to $0$ for larger values of $p$ (not shown in the figure).}. 
\begin{figure}[!hbp]
    \begin{center}
        \begin{tabular}{c c}
            \includegraphics[width=0.45\textwidth]{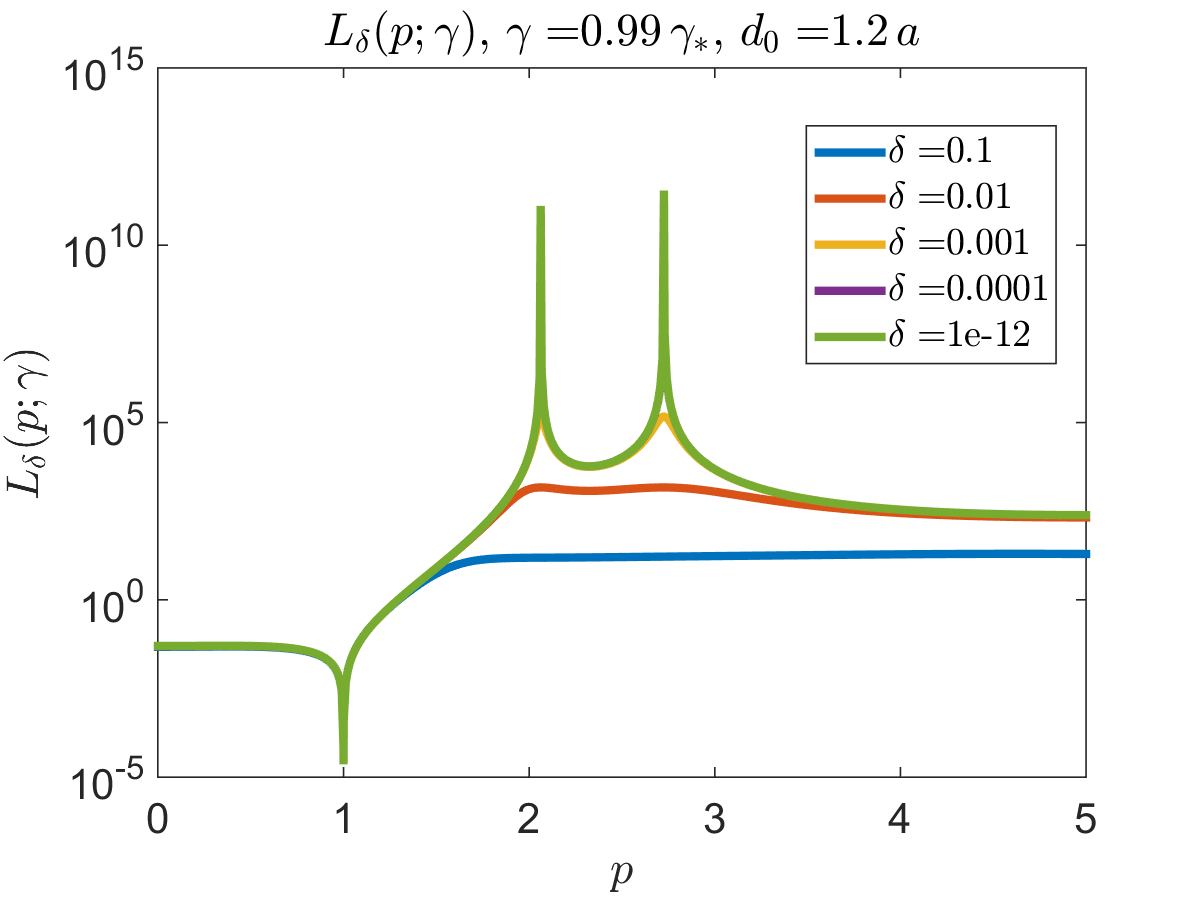}
            & 
            \includegraphics[width=0.45\textwidth]{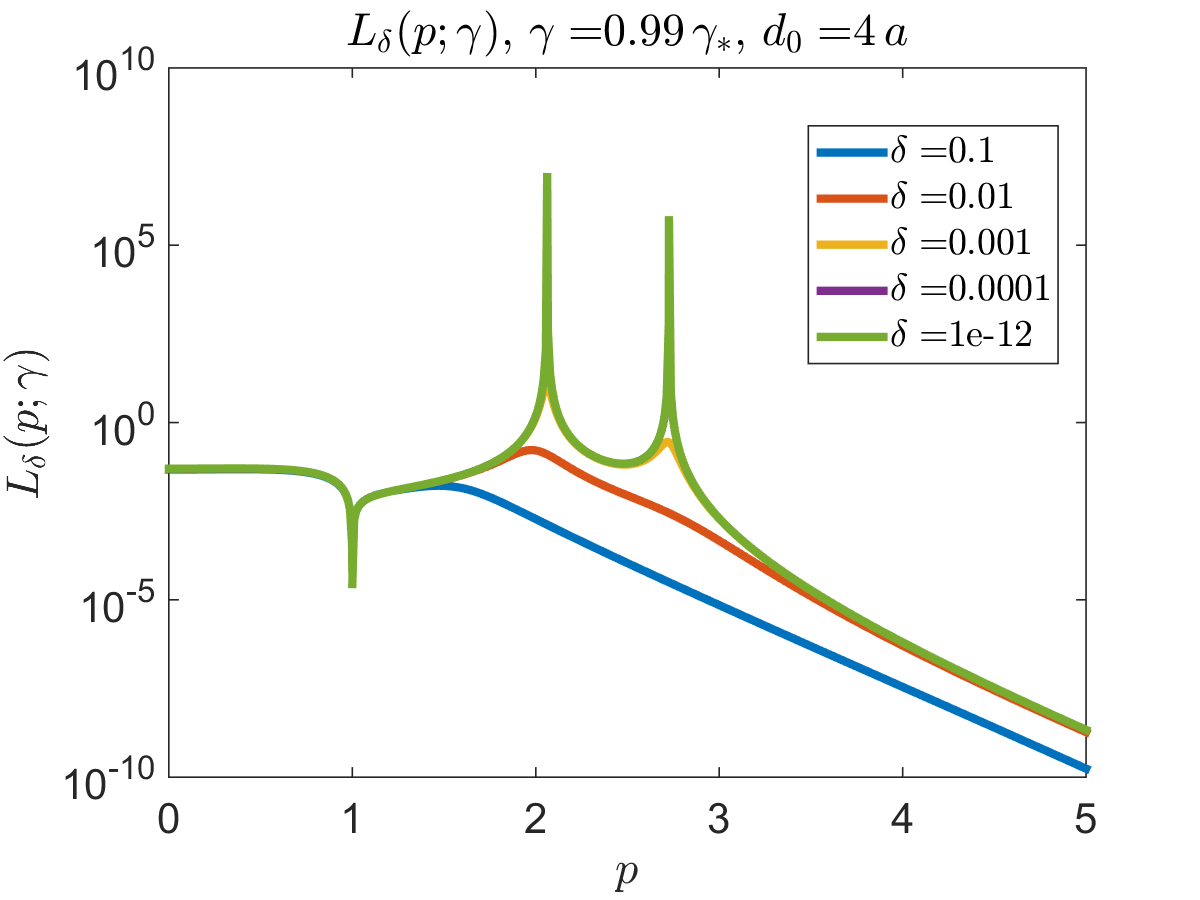}
            \\
            (a) & (b) \myspace{}
            \includegraphics[width=0.45\textwidth]{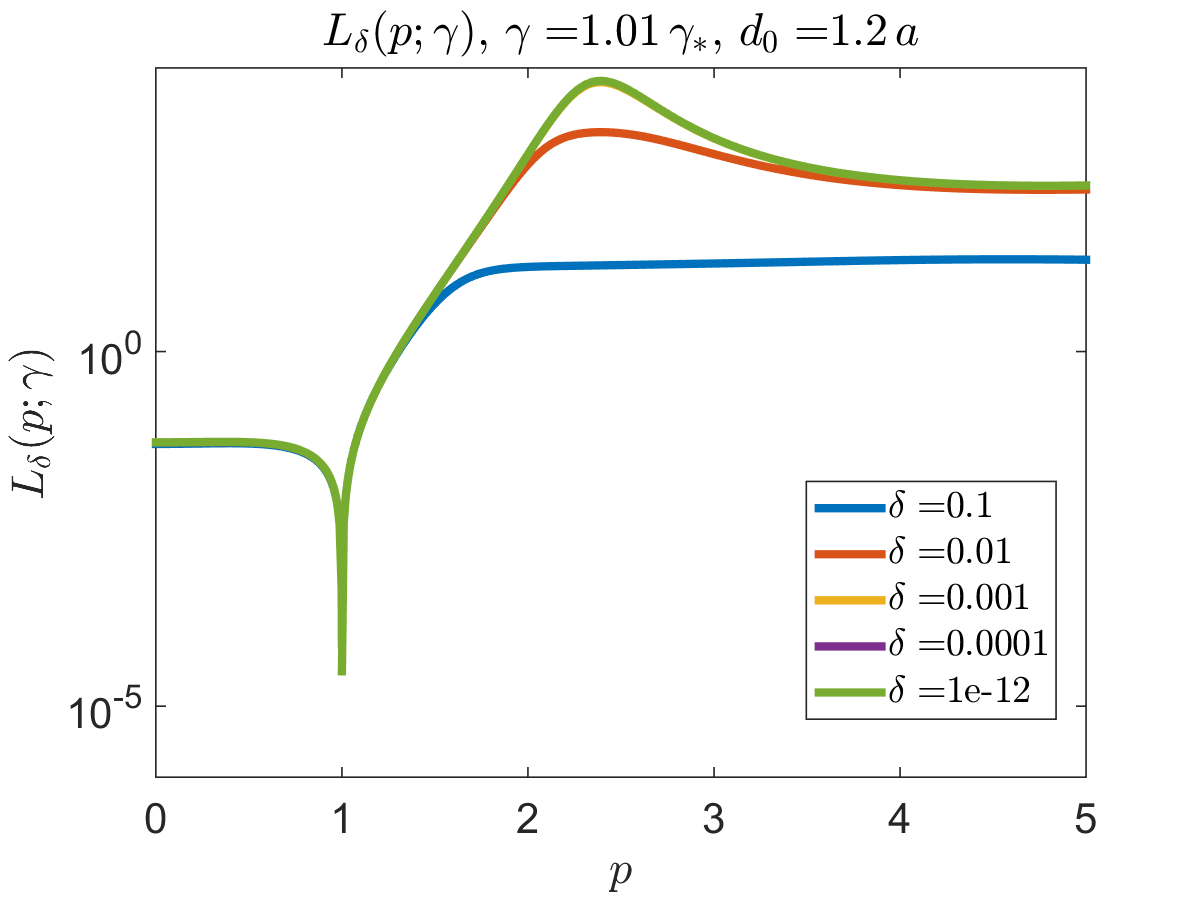}
            & 
            \includegraphics[width=0.45\textwidth]{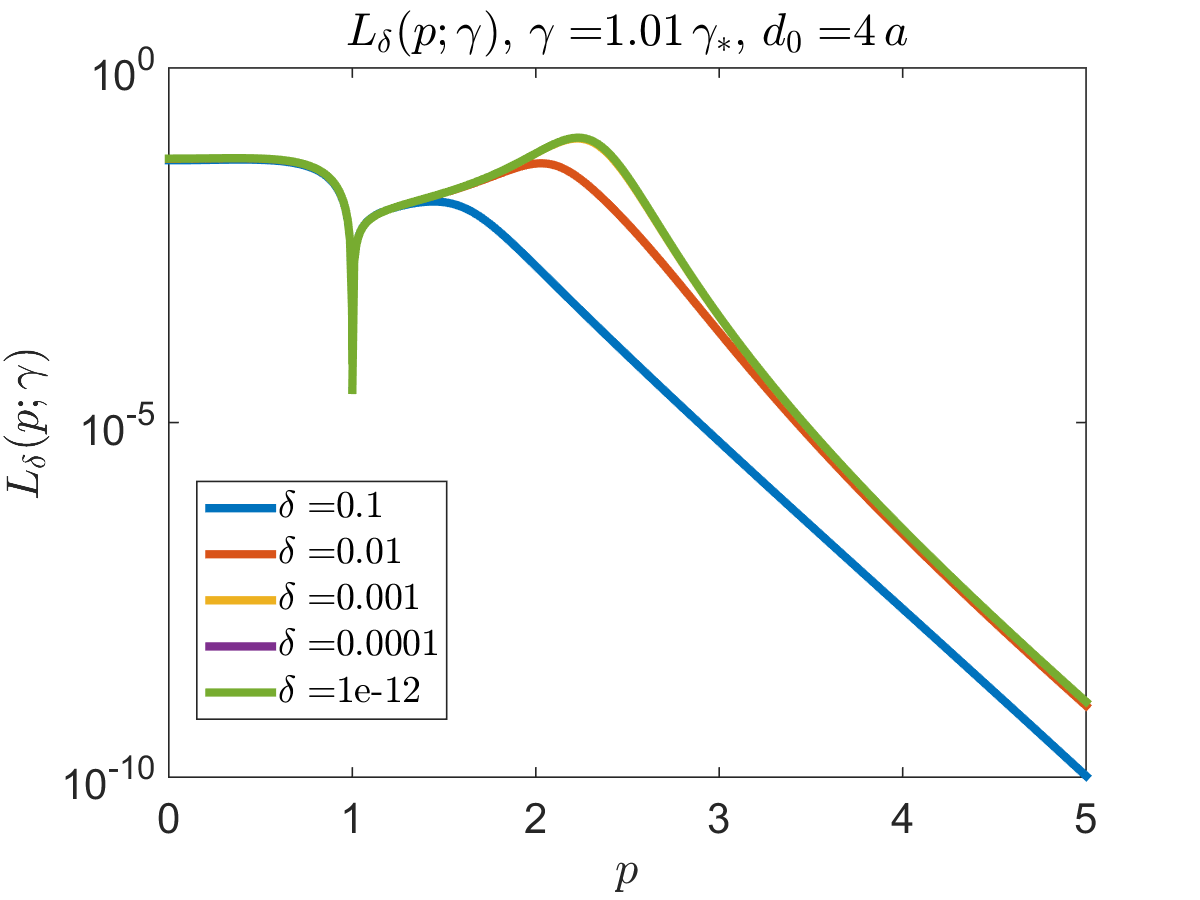}
            \\
            (c) & (d)
        \end{tabular}
        \caption{\emph{A plot of the integrand $L_{\delta}(p;\gamma)$
        from \eqref{eqn:pd_a} for several parameter values.  The
        separate curves in each plot represent different values of
        $\delta$, indicated in the legend: (a) $\gamma
= 0.99\gamma_*$ and $d_0 = 1.2a$; (b) $\gamma = 0.99\gamma_*$ and $d_0 =
4a$; (c) $\gamma = 1.01\gamma_*$ and $d_0 = 1.2a$; (d) $\gamma =
1.01\gamma_*$ and $d_0 = 4a$.}}
        \label{fig:dipole_integrand}
   \end{center}
\end{figure}


\subsection{Sources for which ALR does not occur}
\label{subsec:ALR_busting}

When $0 < \gamma < \gamma_*$, the conjectures from the previous section
suggest that the zeros of $g_0(p;\gamma)$ are responsible for forcing
$E_{\delta}(a)$ to blow up in the limit as $\delta \rightarrow
0^+$.  This begs the question of whether one can design a (realistic)
source in the finite frequency regime (with $0 < \gamma < \gamma_*$)
that effectively cancels the poles that show up in the limit $\delta
\rightarrow 0^+$.  In other words, we would like to design a source such
that $|I_p| = 0$ exactly at the zeros of $g_0(p;\gamma)$; heuristically,
in the limit as $\delta \rightarrow 0^+$, this would force the integrand
in \eqref{eqn:pd} to remain bounded at the zeros of $g_0(p;\gamma)$ and
annihilate the anomalous localized resonance that occurs in this limit.
Recall from \eqref{eqn:Iq} that
\begin{equation}\label{eqn:Ip}
    I_p = \int_{d_0}^{d_1} \fhat(s, k_0p) \ee^{-k_0\nu_m s} \di{s}.
\end{equation}
Lemma~\ref{lem:g_0_roots} implies that $g_0(p;\gamma)$ has two roots $1
< p^1_{\gamma} < p^2_{\gamma}$.  Using this and \eqref{eqn:Ip}, we see
that an ``ALR-busting'' source $f(x,y)$ can be constructed by choosing
$f$ such that $\fhat(s, k_0p^1_{\gamma}) = \fhat(s, k_0p^2_{\gamma}) =
0$ for all $s \in [d_0,d_1]$ (which implies $I_{p^1_{\gamma}} =
I_{p^2_{\gamma}} = 0$).  We do not want to just choose any $\fhat$
satisfying this property, however; we restrict ourselves to those
sources $f \in L^2(\mathcal{M})$ with compact support.  In summary, we
make the following conjecture.
\begin{conjecture}\label{con:ALR_busting}
    Suppose $f\in L^2(\mathcal{M})$ has compact support and
    \[
        \fhat(x,k_0p^1_{\gamma}) = \fhat(x,k_0p^2_{\gamma}) = 0,
    \]
    where $1 < p^1_{\gamma} < p^2_{\gamma}$ are the zeros of
    $g_0(p;\gamma)$ from Lemma~\ref{lem:g_0_roots} and $k_0p^j_{\gamma}$
    are zeros of order at least $1$ for $\fhat(x,k_0p)$.  Then there is
    a $\delta_{\gamma} > 0$ and a constant $C_{\gamma} > 0$ such that
    $E_{\delta}(a) \le C_{\gamma}$ for all $0 < \delta \le
    \delta_{\gamma}$.
\end{conjecture}
There are many sources that satisfy the hypotheses of this theorem.  We
will present $2$ examples here. First, consider
\begin{equation}\label{eqn:special_f_specific_Fourier}
    \fhat(x,q) =-\ii\chi_{(d_0,d_1)}(x)\sinc(\alpha_1 q)\sin(\alpha_2 q),
\end{equation}
where $\sinc(x) = \sin(x)/x$,
\[
    \chi_{(d_0,d_1)}(x) = 
    \begin{cases}
        1 &\text{for } d_0 < x < d_1, \\
        0 &\text{otherwise},
    \end{cases}
\]
\[
    \alpha_1 \equiv \frac{\pi}{k_0p^1_{\gamma}}, \eqtext{and} 
    \alpha_2 \equiv \frac{\pi}{k_0p^2_{\gamma}}.
\]
Then $\fhat(x,\cdot) \in L^2(\mathbb{R})$ and, hence, $f(x,\cdot) \in
L^2(\mathbb{R})$ by the Plancherel Theorem; moreover,
$\fhat(x,k_0p^1_{\gamma}) = \fhat(x,k_0p^2_{\gamma}) = 0$, where the
zeros are order $1$. Finally, by direct calculations we have
\begin{equation}\label{eqn:specific_special_f}
    f(x,y) = \chi_{(d_0,d_1)}(x) \cdot \frac{1}{4\alpha_1}
    \left[H(-y-\alpha_1-\alpha_2) - H(-y-\alpha_1+\alpha_2) +
    H(-y+\alpha_1+\alpha_2) - H(-y-\alpha_2+\alpha_1)\right],
\end{equation}
where $H(z)$ is the Heaviside step function; this $f \in
L^2(\mathcal{M})$ has compact support and thus satisfies the hypotheses
of Conjecture~\ref{con:ALR_busting}.
We may also take
\begin{equation}\label{eqn:special_Bessel_Fourier}
    \fhat(x,q) = \chi_{(d_0,d_1)}(x)J_0(\beta_0 q)J_1(\beta_1 q),
\end{equation}
where $J_0$ and $J_1$ are the Bessel functions of the first kind of
orders $0$ and $1$, respectively, and $\beta_0$ and $\beta_1$ are such
that $J_0(\beta_0 k_0 p^1_{\gamma}) = J_1(\beta_1 k_0 p^2_{\gamma}) = 0$
(we note that these zeros are also of order $1$).  Because the Bessel
functions of the first kind are $O(q^{-1/2})$ as $q\rightarrow \infty$
\cite{Watson:1996:WWW}, we have $\fhat(x,\cdot) \in L^2(\mathbb{R})$.  By the
convolution theorem for Fourier transforms,
\begin{equation}\label{eqn:special_Bessel}
    f(x,y) = \chi_{(d_0,d_1)}(x)(f_0 \ast f_1)(y),
\end{equation}
where $\ast$ denotes convolution and $f_0$ and $f_1$ are the inverse
Fourier transforms of $J_0(\beta_0 q)$ and $J_1(\beta_1 q)$,
respectively; in particular, we obtain
\begin{equation}\label{eqn:f0f1}
    f_0(y) =
    \frac{1}{\pi\sqrt{\beta_0^2-y^2}}\chi_{(-\beta_0,\beta_0)}(y)
    \eqtext{and}
    f_1(y) = \frac{-y}{\beta_1\pi\sqrt{\beta_1^2-y^2}} 
        \chi_{(-\beta_1,\beta_1)}(y).
\end{equation}
Although the convolution in \eqref{eqn:special_Bessel} is difficult to
compute analytically, since $f_0$ and $f_1$ both have compact support
the convolution of $f_0$ with $f_1$ will as well.  Thus $f$ as defined
in \eqref{eqn:special_Bessel} is in $L^2(\mathcal{M})$ and has compact
support.

In Figures~\ref{fig:ALR_bust}(a) and (b) we plot $\real(V)$ and
$\imag(V)$, respectively, for the source from
\eqref{eqn:special_f_specific_Fourier} (equivalently,
\eqref{eqn:specific_special_f}); in Figures~\ref{fig:ALR_bust}(c) and
(d), we plot $\real(V)$ and
$\imag(V)$, respectively, for the source from
\eqref{eqn:special_Bessel_Fourier} (equivalently,
\eqref{eqn:special_Bessel}).  We take the same parameters that we used
in Figures~\ref{fig:dipole_small_gamma}(c)--(d) and
Figures~\ref{fig:general_small_gamma}(c)--(d), namely $d_0 = 1.2a$ and
$\gamma = 0.5\gamma_*$.  In stark contrast with those figures, the
solution $V$ is well-behaved in Figure~\ref{fig:ALR_bust}.
\begin{figure}[!hbp]
    \begin{center}
        \begin{tabular}{c c}
            \includegraphics[width=0.45\textwidth]{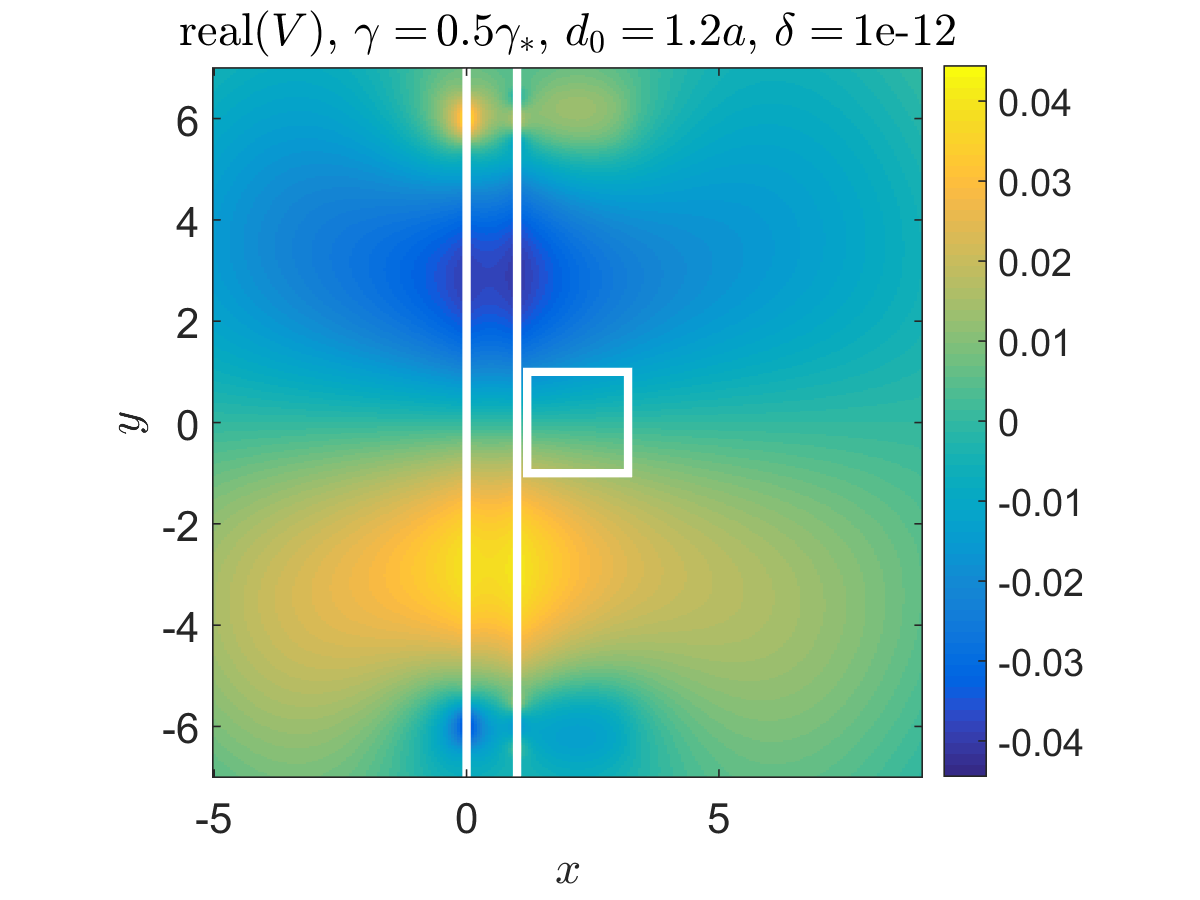}
            &
            \includegraphics[width=0.45\textwidth]{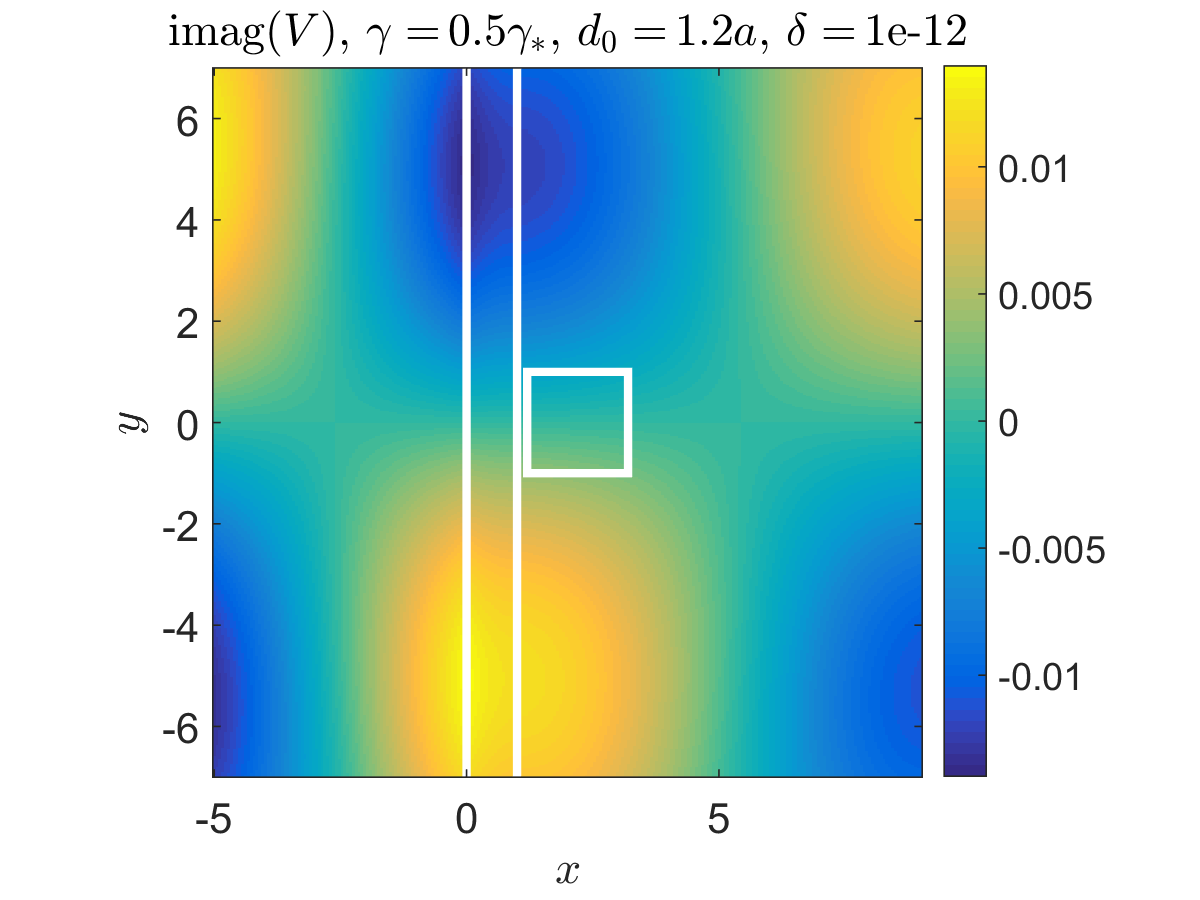}
            \\
            (a) & (b) \myspace{}
            \includegraphics[width=0.45\textwidth]{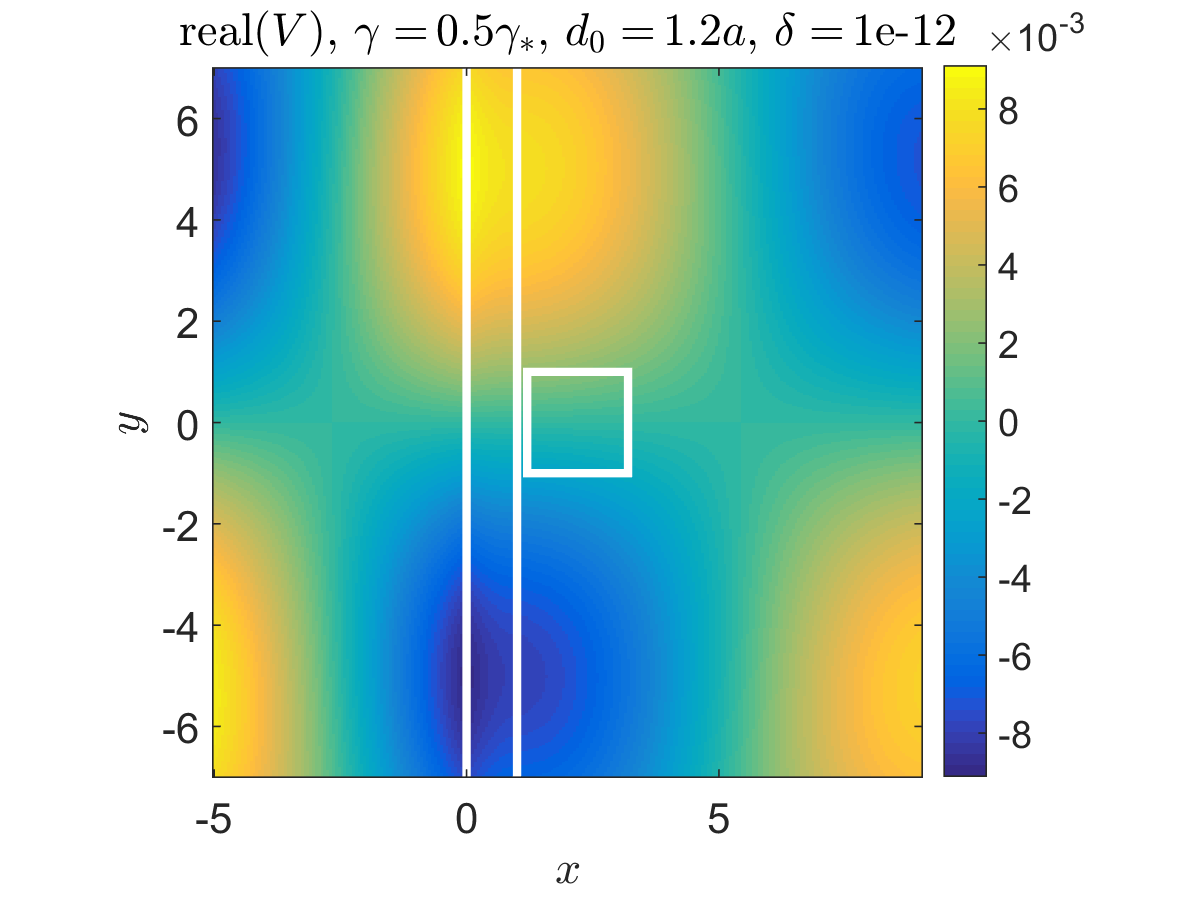}
            &
            \includegraphics[width=0.45\textwidth]{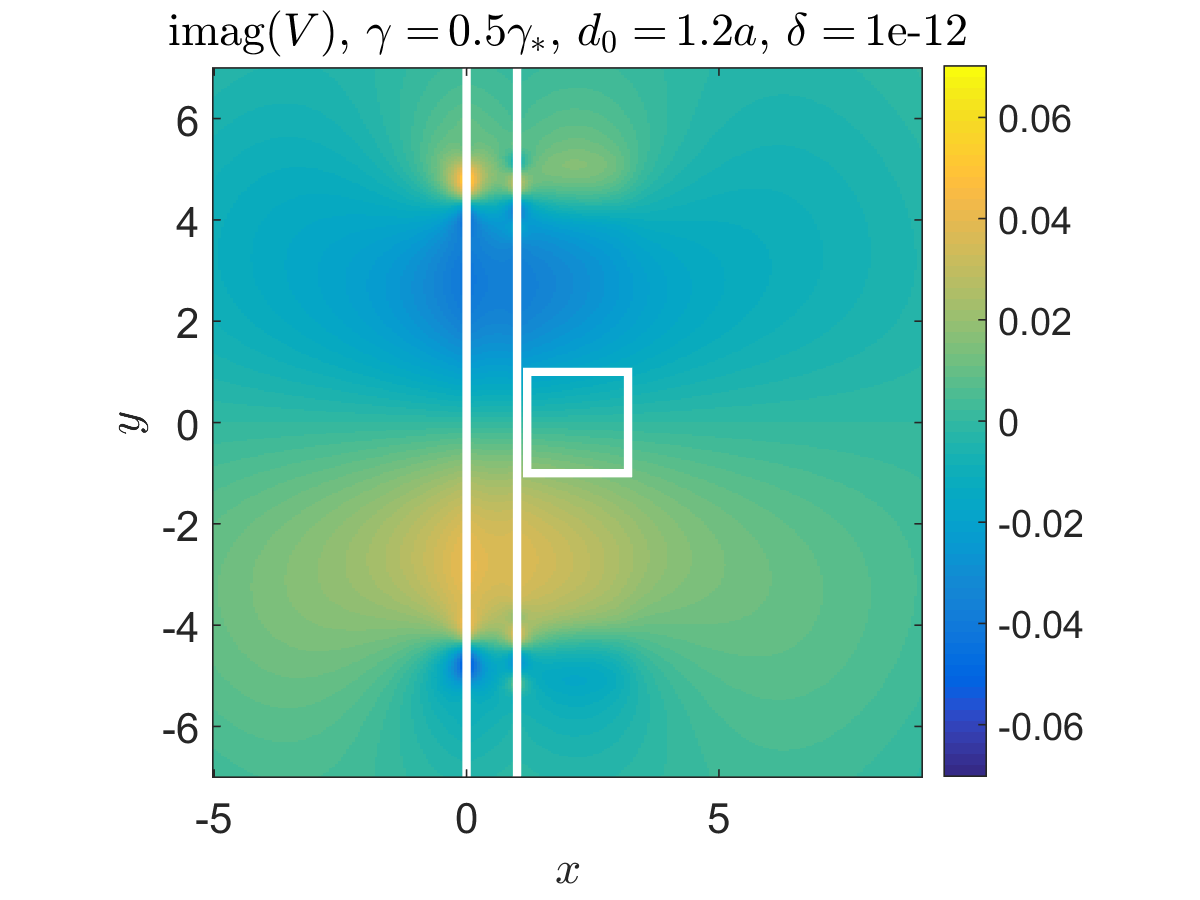}
            \\
            (c) & (d)
        \end{tabular}
        \caption{\emph{This is a plot of $V$, the solution to
                \eqref{eqn:finite_freq}, for two different sources with
                $d_0 = 1.2a$ and $\gamma = 0.5\gamma_*$: (a) $\real(V)$
                and (b) $\imag(V)$ for $f$ as in
                \eqref{eqn:special_f_specific_Fourier} (equivalently,
                \eqref{eqn:specific_special_f}); (c) $\real(V)$ and (d)
                $\imag(V)$ for $f$ as in
                \eqref{eqn:special_Bessel_Fourier} (equivalently
                \eqref{eqn:special_Bessel}).}}
        \label{fig:ALR_bust}
    \end{center}
\end{figure}


\subsubsection{Current sources for which ALR does not occur}

In the monochromatic electromagnetic setting, $f$ must satisfy
additional restrictions for it to represent a realistic
(divergence-free) current source --- see
\S~\ref{subsec:Maxwell_to_Helmholtz}.  In particular, the function $f$
should be in $L^2(\mathcal{M})$ with compact support and be of the form
\begin{equation}\label{eqn:f}
    f = \mu_0\left(\frac{\partial \widetilde{J}_y}{\partial x} -
        \frac{\partial \widetilde{J}_x}{\partial y}\right)
\end{equation}
for a current 
\begin{equation}\label{eqn:special_current_main}
    \widetilde{\mathbf{J}} = \widetilde{J}_x(x,y)\mathbf{e}_x +
    \widetilde{J}_y(x,y)\mathbf{e}_y
\end{equation}
satisfying the continuity equation
\begin{equation}\label{eqn:2D_cont_main}
    \frac{\partial \widetilde{J}_x}{\partial x} + \frac{\partial
    \widetilde{J}_y}{\partial y} = 0.
\end{equation}

We now construct a source $f$ satisfying the hypotheses of
Conjecture~\ref{con:ALR_busting} that is of the form
\eqref{eqn:f}--\eqref{eqn:2D_cont_main}.  For simplicity, we assume that
the current from \eqref{eqn:special_current_main} has the form
\begin{equation}\label{eqn:special_J_separated}
    \widetilde{J}_x(x,y) = r_1(x)t_1(y) \eqtext{and}
    \widetilde{J}_y(x,y) = r_2(x)t_2(y).
\end{equation}
with $r_1, r_2, t_1, t_2$ smooth enough. Then the continuity equation \eqref{eqn:2D_cont_main} becomes
\[
    r_1'(x)t_1(y) + r_2(x)t_2'(y) = 0.
\]
Taking the Fourier transform of this equation with respect to $y$ gives
\begin{equation}\label{eqn:cont_Fourier}
    r_1'(x)\widehat{t}_1(q) + r_2(x)\ii q\widehat{t}_2(q) = 0.
\end{equation}
We further simplify the problem by taking 
\begin{equation}\label{eqn:special_q_condition}
    \widehat{t}_1(q) = \ii q\widehat{t}_2(q).
\end{equation}
Then \eqref{eqn:cont_Fourier} becomes
\[
    \left[r_1'(x) + r_2(x)\right]\widehat{t}_1(q) = 0,
\]
which is satisfied for all $x$ and $q$ if 
\begin{equation}\label{eqn:special_x_condition}
    r_2(x) = -r_1'(x).
\end{equation}
Then \eqref{eqn:f}, \eqref{eqn:special_J_separated},
\eqref{eqn:special_q_condition}, and
\eqref{eqn:special_x_condition} imply that 
\begin{equation}\label{eqn:special_f_general}
    \fhat(x,q) 
    = \mu_0\left[r_2'(x)\widehat{t}_2(q) - r_1(x)\ii qt_1(q)\right]
    = \mu_0\left[-r_1''(x)\widehat{t}_2(q) + r_1(x)q^2\widehat{t}_2(q)\right]
\end{equation}
At this point, $f$ satisfies \eqref{eqn:f}--\eqref{eqn:2D_cont_main}.

By the Plancherel Theorem and \eqref{eqn:special_f_general},
$f\in L^2(\mathcal{M})$ if and only if $\widehat{t}_2(q) \in
L^2(\mathbb{R})$ and $q^2\widehat{t}_2(q) \in
L^2(\mathbb{R})$.  We must also be careful to also choose
$\widehat{t}_2(q)$ in such a way that $\widehat{t}_2(k_0p^1_{\gamma}) =
\widehat{t}_2(k_0p^2_{\gamma}) = 0$ and $f(x,y) =
\mu_0\left[-r_1''(x)t_2(y) - r_1(x)t_2''(y)\right]$
has compact support in $y$.

There are many examples of functions that accomplish these tasks.
Unfortunately, the functions in \eqref{eqn:specific_special_f} and
\eqref{eqn:special_Bessel} lead to current sources that are
discontinuous and, hence, not divergence free, so we need to be a bit
more careful.  To find a smooth current with compact support satisfying
our requirements, we take
\begin{equation}\label{eqn:special_t_2}
    \widehat{t}_2(q) = \sinc^3(\alpha_1 q) \sinc^2(\alpha_2 q),
\end{equation}
\[
    \alpha_1 \equiv \frac{\pi}{k_0p^1_{\gamma}}, \eqtext{and} 
    \alpha_2 \equiv \frac{\pi}{k_0p^2_{\gamma}}.
\]
Then $\widehat{t}_2(q) \in L^2(\mathbb{R})$, $q^2\widehat{t}_2(q) \in
L^2(\mathbb{R})$  and $\fhat(x,k_0p^1_{\gamma}) =
\fhat(x,k_0p^2_{\gamma}) = 0$\footnote{Here the zeros are of order $3$ and
$2$, respectively, so they are stronger than what we need according to
Conjecture~\ref{con:ALR_busting}.  We take
these higher-order zeros to ensure that $t_2(y)$ and $t_2''(y)$ are both
continuous.  There may be other choices of continuous functions $t_2(y)$
and $t_2''(y)$ such that $\widehat{t}_2(q) \in L^2(\mathbb{R})$ that has
zeros of order $1$ at $q = k_0p_1$ and $q = k_0p_2$.}.

One possible choice of $r_1(x)$ from \eqref{eqn:special_f_general} is
\begin{equation}\label{eqn:special_r_1}
    r_1(x) = 
    \begin{cases}
        C\left[\frac{2}{d_1-d_0}\left(x-d_0\right)-1\right]^3
        \left[\left|\frac{2}{d_1-d_0}\left(x-d_0\right)-1\right|
        -1\right]^3 &\text{for } d_0 \le x \le d_1, \\
        0 &\text{otherwise},
    \end{cases}
\end{equation}
where $C$ is a nonzero constant.  The function $r_1(x)$ is twice
continuously differentiable and has compact support, so $r_1''(x)$ is
continuous with compact support.  

Finally, $\fhat(x,q)$ may be computed via \eqref{eqn:special_f_general},
\eqref{eqn:special_t_2}, and \eqref{eqn:special_r_1}.  We note that the
inverse Fourier transform of $\fhat(x,q)$ can be computed analytically; for the benefit of the reader, we avoid writing out the
expression.  Importantly, $f$ is continuous with compact support.

The current source corresponding to this $f$ may be computed via
\eqref{eqn:special_current_main}, \eqref{eqn:special_J_separated},
\eqref{eqn:special_q_condition}, \eqref{eqn:special_x_condition},
\eqref{eqn:special_t_2}, and \eqref{eqn:special_r_1}.  We emphasize that 
both $\widetilde{J}_x(x,y)$ and $\widetilde{J}_y(x,y)$ are continuously
differentiable functions with compact support in $\mathcal{M}$.  

In Figure~\ref{fig:realistic_source}, we plot the source $f$ defined by
\eqref{eqn:special_f_general}--\eqref{eqn:special_r_1} with $C =
10^3\mu_0^{-1}$.  In Figures~\ref{fig:realistic_sine}(a) and (b), we
plot $\real(V)$ and $\imag(V)$, respectively, corresponding to this
source.  As expected, we see that there is no resonant region near the
slab even though the source is quite close to the slab ($d_0 = 1.2a$),
$\gamma = 0.5\gamma_*$, and $\delta = 10^{-12}$.
\begin{figure}[!hbp]
    \begin{center}
        \includegraphics[width=0.75\textwidth]{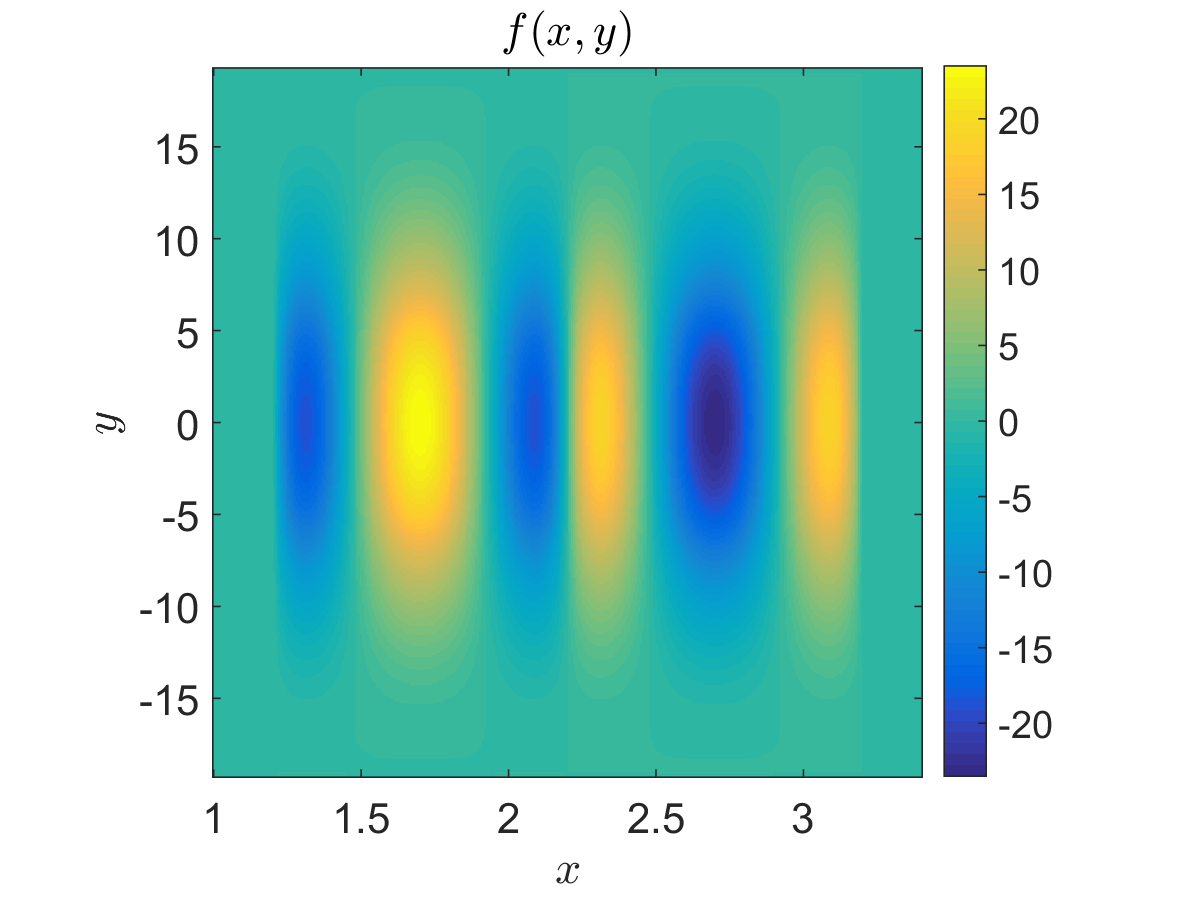}
        \caption{\emph{The source $f$ defined by
            \eqref{eqn:special_f_general}--\eqref{eqn:special_r_1} with 
            the parameters $d_0 = 1.2a$ and $d_1 = d_0 + 2$.}}
        \label{fig:realistic_source}
    \end{center}
\end{figure}
\begin{figure}[!hbp]
    \begin{center}
        \begin{tabular}{c c}
            \includegraphics[width=0.45\textwidth]{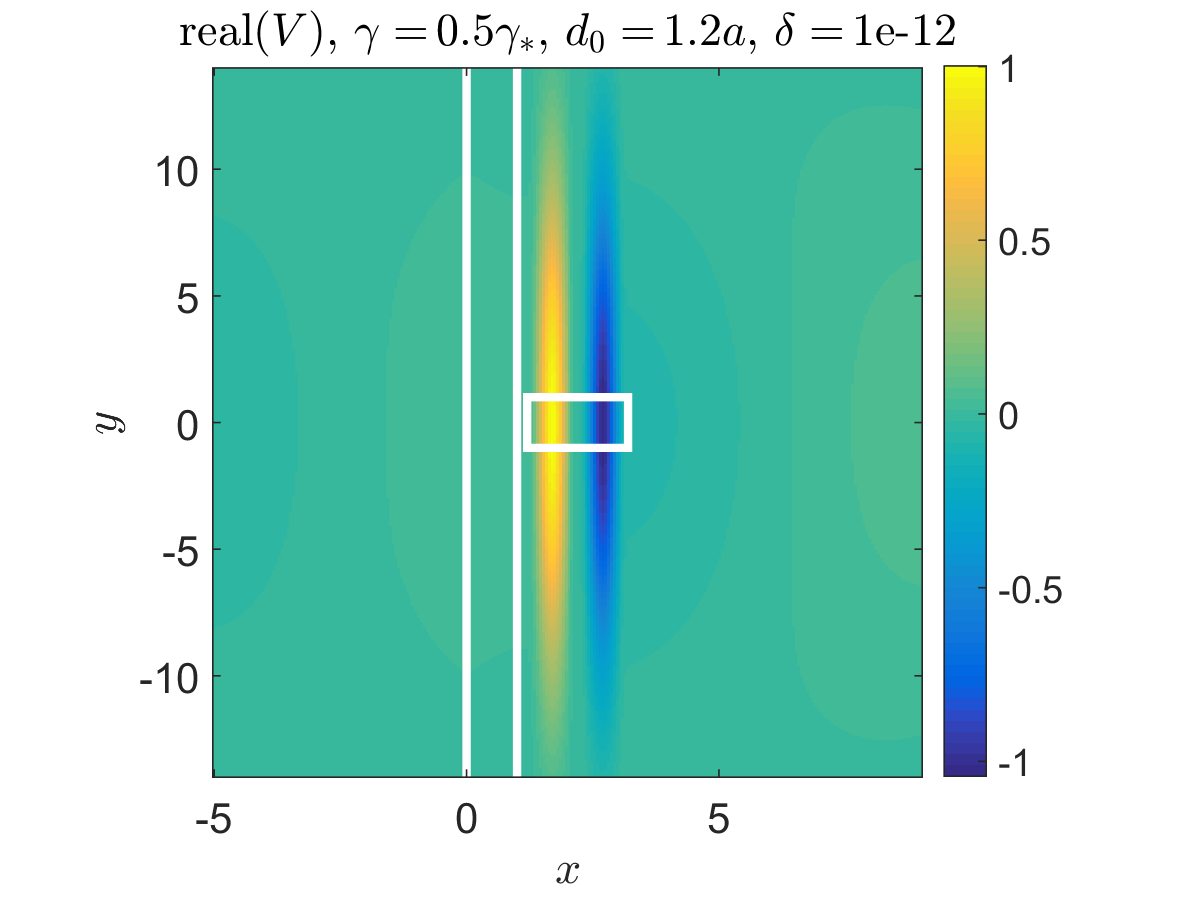}
            &
            \includegraphics[width=0.45\textwidth]{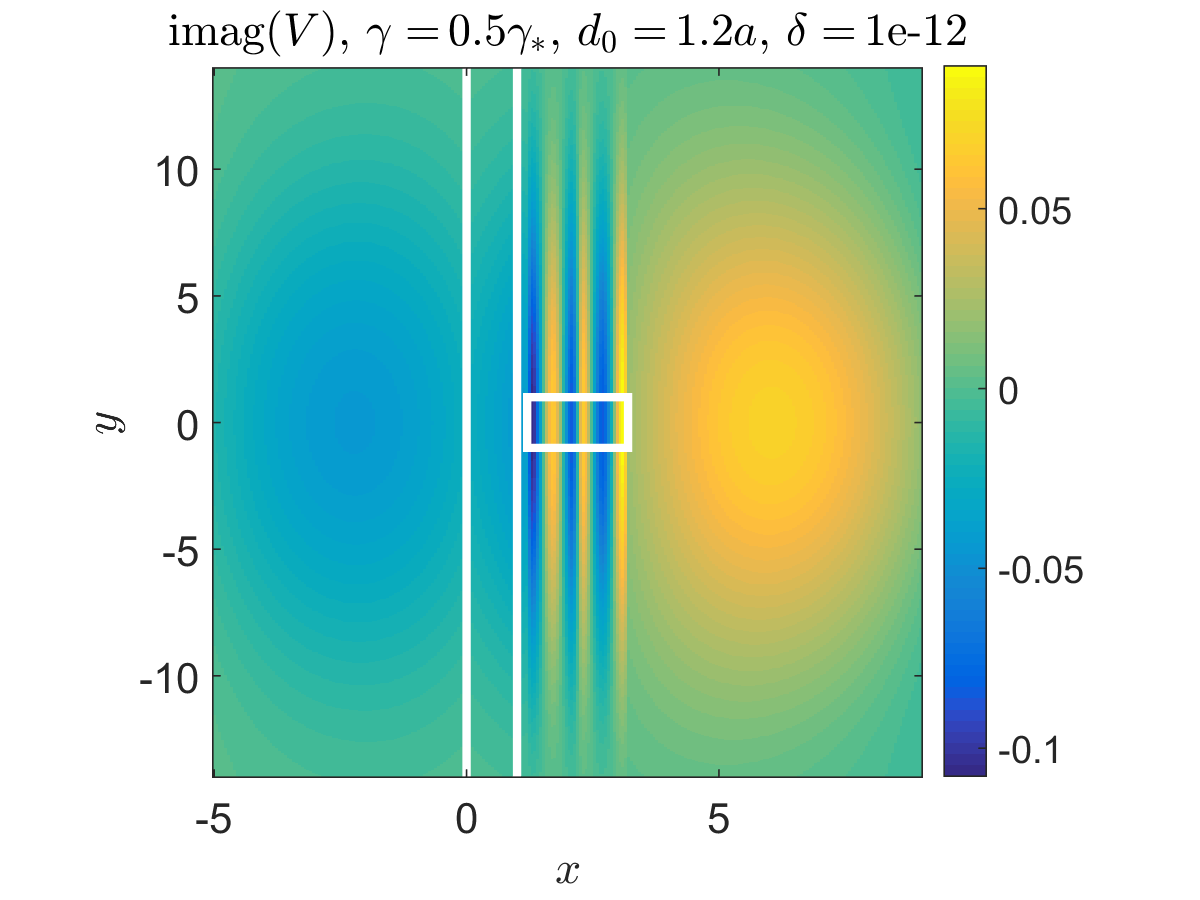}
            \\
            (a) & (b)
        \end{tabular}
        \caption{\emph{In this figure, we plot (a) $\real(V)$ and (b)
                $\imag(V)$, where $V$ is the solution to
                \eqref{eqn:finite_freq} corresponding to the source $f$
                defined through
                \eqref{eqn:special_f_general}--\eqref{eqn:special_r_1}
                with the parameters $d_0 = 1.2a$, $\gamma = 0.5\gamma_*$,
                and $\delta = 10^{-12}$.  To make the behavior of $V$
                more clear, we clipped the maximum and minimum values in
                each plot to $0.2$ (yellow) and $-0.2$ (blue)
        respectively.}}
        \label{fig:realistic_sine}
    \end{center}
\end{figure}


\appendix

\section{Proofs and derivations omitted in the text}

In this appendix, we provide detailed proofs we omitted in the
main body of the paper.


\subsection{Proof of Lemma~\ref{lem:g_0_roots}}
\label{subsec:proof_of_lem_g_0_roots}

Setting $g_0(p;\gamma) = 0$, defining a new variable $s \equiv p^2$, and
simplifying, we find that $g_0(p;\gamma) = 0$ is equivalent to having
\begin{equation}\label{eqn:g_0_roots_equivalent}
    s + \sqrt{s^2-1} = \ee^{\gamma\sqrt{s+1}}.
\end{equation}
We define
\begin{equation}\label{eqn:G_0}
    G_0(s;\gamma) \equiv s + \sqrt{s^2-1} - \ee^{\gamma\sqrt{s+1}}
    \eqtext{for} s\ge 0, \, \gamma > 0.
\end{equation}
Then $g_0(p;\gamma) = 0$ if and only if $G_0(s;\gamma) = 0$.  We will
complete the proof of the lemma in several steps.
\begin{enumerate}
    \item \emph{Claim: $G_0(s;\gamma) \ne 0$ for $s \le 1$.}\\
        \emph{Proof of claim:} For $s \le 1$, 
        \[
            G_0(s;\gamma) = s - \ee^{\gamma\sqrt{s+1}} +
            \ii\sqrt{1-s^2}.
        \]
        Since $\imag G_0(s;\gamma) = \sqrt{1-s^2} > 0$ for $s < 1$, the
        only point at which $G_0(s;\gamma)$ could possibly be $0$ is $s
        = 1$.  But, for $\gamma > 0$,
        \[
            G_0(1;\gamma) = 1-\ee^{\sqrt{2}\gamma} < 0.
        \]
        In particular, this proves that if $p_0 \ge 0$ is a root of
        $g_0(p;\gamma)$, then $p_0 > 1$.
    \item \emph{Claim: For $\gamma \ge \sqrt{\ee}/(\ee+1) \approx
        0.4434$, the function $s \mapsto G_0(s;\gamma)$ is concave for
        $s > 1$.} \\
        \emph{Proof of claim:} For $s > 1$, we have
        \[
            G_0(s;\gamma) = s + \sqrt{s^2-1} - \ee^{\gamma\sqrt{s+1}}
        \]
        and
        \[
            \frac{\partial^2 G_0}{\partial s^2} = 
            -\frac{1}{(s^2-1)^{3/2}} +
            \frac{\gamma\ee^{\gamma\sqrt{s+1}}}{4(s+1)^{3/2}}
            \left(1-\gamma\sqrt{s+1}\right).
        \]
        We note that $1-\gamma\sqrt{s+1} < 0$ for $\gamma >
        \frac{1}{\sqrt{2}} \approx 0.7071$ and $s > 1$.  This proves the
        claim for $\gamma > \frac{1}{\sqrt{2}}$.

        Now consider 
        \[
            \frac{\sqrt{\ee}}{\ee+1} \le \gamma \le \frac{1}{\sqrt{2}}.
        \]
        We note that $1-\gamma\sqrt{s+1} < 0$ for $s > \gamma^{-2}-1$,
        so $\partial^2G_0/\partial s^2 < 0$ for $s > \gamma^{-2}-1$.
        Then, for $1 < s \le \gamma^{-2} -1$, we have
        \begin{equation}\label{eqn:int_inequality}
            \frac{\partial^2G_0}{\partial s^2} < 0
            \quad
            \Leftrightarrow
            \quad
            \gamma\left(s-1\right)^{3/2}\ee^{\gamma\sqrt{s+1}}
            \left(1-\gamma\sqrt{s+1}\right) < 4.
        \end{equation}
        For $1 < s \le \gamma^{-2}-1$, the left-hand side of the above
        inequality satisfies
        \begin{equation}\label{eqn:int_inequality_2}
            \gamma\left(s-1\right)^{3/2}\ee^{\gamma\sqrt{s+1}}
            \left(1-\gamma\sqrt{s+1}\right)
            \le
            \ee\gamma\left(\gamma^{-2}-2\right)^{3/2}
            \left(1-\sqrt{2}\gamma\right).
        \end{equation}
        Using calculus and Maple, it can be shown that
        \[
            \max_{\gamma\in \left[\frac{\sqrt{\ee}}{\ee+1},
            \frac{1}{\sqrt{2}}\right]}
            \ee\gamma\left(\gamma^{-2}-2\right)^{3/2}
            \left(1-\sqrt{2}\gamma\right) 
            = \frac{\left(\ee^2+1\right)^{3/2} \left(\ee - \sqrt{2\ee} +
                1\right)}{\left(\ee+1\right)^2} \approx 2.4370.
        \]
        In combination with \eqref{eqn:int_inequality} and
        \eqref{eqn:int_inequality_2}, this implies
        \[
            \frac{\partial^2G_0}{\partial s^2} < 0
        \]
        for $s > 1$ as long as $\gamma \ge \sqrt{\ee}/(\ee+1)$ and
        verifies the claim.
    \item \emph{Claim: For $0 < \gamma < \sqrt{\ee}/(\ee+1) \approx
        0.4434$, the function $s \mapsto G_0(s;\gamma)$ has two real
        zeros $1 < s^1_{\gamma} < s^2_{\gamma}$.} \\
        \emph{Proof of claim:} We begin by defining the functions 
        \[
            G_1(s) \equiv s+\sqrt{s^2-1} 
            \eqtext{and}
            G_2(s;\gamma) \equiv \ee^{\gamma\sqrt{s+1}}
        \]
        for $s > 1$ and $\gamma \in (0,\sqrt{\ee}/(\ee+1)$.  Then 
        \[
            G_0(s;\gamma) = G_1(s)-G_2(s;\gamma).
        \]
        For $s > 1$ we have
        \begin{equation}\label{eqn:dG1ds}
            \frac{\dd G_1}{\dd s} = 1 + \frac{s}{\sqrt{s^2-1}} > 2
            \eqtext{and}
            \frac{\dd^2 G_1}{\dd s^2} = -\frac{1}{(s^2-1)^{3/2}} <
            0.
        \end{equation}
        Similarly, for $s > 1$ we have
        \[
            \frac{\partial G_2}{\partial s} =
            \frac{\gamma}{2}\cdot
            \frac{\ee^{\gamma\sqrt{s+1}}}{\sqrt{s+1}} > 0 
            \eqtext{and}
            \frac{\partial^2G_2}{\partial s^2} =
            \frac{\gamma\ee^{\gamma\sqrt{s+1}}}{4(s+1)^{3/2}}
            \left(\gamma\sqrt{s+1}-1\right).
        \]
        Then $\partial^2G_2/\partial s^2 < 0$ for $s < \gamma^{-2}-1$
        and $\partial^2G_2/\partial s^2 > 0$ for $s > \gamma^{-2}-1$.

        Next, we note that 
        \[
            G_0(1;\gamma) = G_1(1;\gamma) - G_2(1;\gamma) =
            1-\ee^{\sqrt{2}\gamma} < 0.
        \]
        In addition, we have
        \[
            \frac{\partial G_2}{\partial s}(1;\gamma) \le 
            \frac{\partial G_2}{\partial
            s}\left(1;\frac{\sqrt{\ee}}{\ee+1}\right) \approx 0.2935 <
            2.
        \]
        Because $\partial G_2/\partial s$ is decreasing for $1 < s <
        \gamma^{-2}-1$, and $\dd G_1/\dd s > 2$ for $s > 1$, the above
        inequality implies
        \[
            \frac{\partial G_0}{\partial s} = \frac{\dd
            G_1}{\dd s} - \frac{\partial G_2}{\partial s} > 0
            \eqtext{for} 1 < s < \gamma^{-2}-1.
        \]
        In fact, $G_0$ increases enough on this interval to become
        positive; we have
        \begin{align*}
            G_0(\gamma^{-2}-1;\gamma) &=
            \gamma^{-2}\left[1-(\ee+1)\gamma^2 +
            \sqrt{1-2\gamma^2}\right] \\
            &\ge 
            \gamma^{-2}\left[1-\frac{\ee}{\ee+1} +
            \sqrt{1-\frac{2\ee}{(\ee+1)^2}}\right] \\
            &\approx 1.0479\gamma^{-2} > 0.
        \end{align*}
        
        For $s > \gamma^{-2}-1$, $\partial G_2/\partial s$ increases
        without bound while $\dd G_1/\dd s$ will get arbitrarily close
        to $2$; the upshot of this is that $\partial G_0/\partial s$
        becomes arbitrarily negative for $s$ large enough.

        In summary, we have $G_0(1;\gamma) < 0$; the function
        $G_0(s;\gamma)$ increases at least until $s = \gamma^{-2}-1$
        where $G_0(\gamma^{-2}-1;\gamma) > 0$.  Next, $G_0(s;\gamma)$
        will continue increasing until $\partial G_2/\partial s$ becomes
        larger than $\dd G_1/\dd s$; then $G_0(s;\gamma)$ will decrease
        toward $-\infty$ as $s$ approaches $\infty$.  Thus
        $G_0(s;\gamma)$ has $2$ real zeros for $\gamma <
        \sqrt{\ee}/(\ee+1)$; by item (1) above, both of these zeros must
        be larger than $1$.  Finally, since
        $G_0(\gamma^{-2}-1;\gamma)$ is strictly greater than $0$, by
        continuity the roots cannot be equal.  This also proves
        that both of the zeros are of order $1$.  This proves the claim.
\end{enumerate}
We have shown that $G_0(s;\gamma)$ has two real roots $1 < s^1_{\gamma}
< s^2_{\gamma}$ provided $0 < \gamma < \sqrt{\ee}/(\ee+1)$.  For $\gamma
\ge \sqrt{\ee}/(\ee+1)$, the function $s\mapsto G_0(s;\gamma)$ is
concave for $s > 1$ by item (2) above.  Thus $G_0(s;\gamma)$ has a
unique maximum and will have two real roots of order $1$ if the maximum
is positive and no real roots if the maximum is negative.  

Because the maximum of $G_0(s;\gamma)$ is positive for $\gamma =
\sqrt{\ee}/(\ee+1)$ (see the proof of item (3) above) and
$G_0(s;\gamma)$ decreases at an exponential rate as a function of
$\gamma$, there exists a $\gamma_* > \sqrt{\ee}/(\ee+1)$ such that 
\[
    \max_{s>1} G_0(s;\gamma) 
    \begin{cases}
        > 0 &\text{for } \gamma < \gamma_*, \\
        < 0 &\text{for } \gamma > \gamma_*.
    \end{cases}
\]
In particular, we have
\begin{equation}\label{eqn:gamma_star}
    \gamma_* = \argmin_{\gamma > 0} \left[\max_{s > 1}
    G_0(s;\gamma)\right].
\end{equation}
MATLAB gives $\gamma_* \approx 0.9373$.  This completes the proof.


\subsection{$I_p$ for dipole sources}\label{subsec:dipole}

In this section, we derive an explicit formula for $I_q$, defined in
\eqref{eqn:Iq}, when the source $f$ is a dipole.  In particular, we
consider a source of the form
\begin{equation*}
    f(x,y) = \mathbf{d}\cdot\nabla[\delta(x-x_0)\delta(y-y_0)]
    = d_x\left[\frac{\partial}{\partial x}\delta(x-x_0)\right]
    \delta(y-y_0) + 
    d_y \delta(x-x_0)\frac{\partial}{\partial y}\delta(y-y_0);
\end{equation*}
here $\mathbf{d} = \left[d_x,d_y\right]^T$ is the dipole moment. Then
\eqref{eqn:Fourier_def} gives
\begin{align*}
    \widehat{f}(x,q) 
    &= d_x\frac{\partial}{\partial x}\delta(x-x_0)
        \int_{-\infty}^{\infty} \delta(y-y_0) \ee^{-\ii q y} \, dy +
        d_y\delta(x-x_0)\int_{-\infty}^{\infty} \frac{\partial}{\partial
        y}\delta(y-y_0) \ee^{-\ii q y} \, dy \\
    &= d_x\frac{\partial}{\partial x}\delta(x-x_0)\ee^{-\ii q y_0} -
        d_y\delta(x-x_0)\int_{-\infty}^{\infty} \delta(y-y_0)(-\ii q)
        \ee^{-\ii q y} \, dy\\
    &= d_x\frac{\partial}{\partial x}\delta(x-x_0)\ee^{-\ii q y_0} + \ii
        q d_y\delta(x-x_0)\ee^{-\ii q y_0}.
\end{align*}
Next, \eqref{eqn:Iq} implies
\begin{align*}
    I_q 
    &= \int_{a}^{\infty} \widehat{f}(s,q) \ee^{-k_0\nu_m s} \, ds \\
    &= d_x\ee^{-\ii q y_0} \int_a^{\infty} \frac{\partial}{\partial s}
    \delta(s-x_0)\ee^{-k_0\nu_m s} \, ds + \ii q d_y\ee^{-\ii q y_0}
    \int_{a}^{\infty} \delta(s-x_0)\ee^{-k_0\nu_m s} \, ds \\
    &= -d_x\ee^{-\ii q y_0} \int_a^{\infty} \delta(s-x_0)
    (-k_0\nu_m)\ee^{-k_0\nu_m s} \, ds + \ii q d_y \ee^{-\ii q y_0}
    \ee^{-k_0\nu_m x_0} \\
    &= d_x\ee^{-\ii q y_0} k_0\nu_m\ee^{-k_0\nu_m x_0} + \ii q d_y 
    \ee^{-\ii q y_0} \ee^{-k_0\nu_m x_0} \\
    &= (d_x k_0\nu_m + \ii d_y q)\ee^{-k_0\nu_m x_0}\ee^{-\ii q y_0}.
\end{align*}
If we take $x_0 = d_0$ and $y_0 = 0$ (the typical case), this becomes
\[
    I_q = (d_x k_0\nu_m + \ii d_y q)\ee^{-\nu_m d_0},
\]
which, after the changes of variables $p = q/k_0$ and $k_0 = \gamma/a$, 
becomes
\begin{equation}\label{eqn:I_p_dipole}
    I_p = 
    \begin{cases}
        \ii \dfrac{\gamma}{a}\left(d_x\sqrt{1-p^2} + d_yp\right) 
        \ee^{-\ii \gamma\frac{d_0}{a}\sqrt{1-p^2}} 
        &\text{if } 0 \le p \le 1, \myspace{}
        \dfrac{\gamma}{a}\left(d_x\sqrt{p^2-1} + \ii d_y p\right)
        \ee^{-\gamma\frac{d_0}{a}\sqrt{p^2-1}} 
        &\text{if } 1 \le p < \infty.
    \end{cases}
\end{equation}
From here the triangle inequality implies that $I_p$ satisfies the bound
\eqref{eqn:I_p_stricter_bounds} with
\[
    B(p;\gamma) \equiv
    \begin{cases}
        \dfrac{\gamma}{a}\left(|d_x|\sqrt{1-p^2} + |d_y|p\right)
        &\text{for } 0 \le p \le 1, \myspace{}
        \dfrac{\gamma}{a}\left(|d_x|\sqrt{p^2-1} + |d_y|p\right)
        &\text{for } 1 \le p < \infty.
    \end{cases}
\]

Finally, the same computations as those leading up to
\eqref{eqn:I_p_dipole} imply that quadrupole, octopole, and higher order
distributional sources satisfy equations similar to
\eqref{eqn:I_p_dipole} with higher order powers of $p$ and $\gamma$.
Therefore, as discussed in Remark~\ref{rem:dipole_bounded}, such sources
satisfy Theorems~\ref{thm:bounded} and \ref{thm:shielding}.


\subsection{Proof of Theorem~\ref{thm:bounded}}
\label{subsec:proof_of_thm_bounded}

We will prove this theorem in several steps.  Essentially, our goal is
to bound the integrand $L_{\delta}(p;\gamma)$ from above by a function
that is integrable and independent of $\delta$.  We begin by finding a
lower bound on $|g_{\delta}(p;\gamma)|$; in particular, we note that
\eqref{eqn:g} and the reverse triangle inequality imply that
\begin{equation}\label{eqn:g_reverse_triangle}
    |g_{\delta}(p;\gamma)| \ge 
    \left|\left|\nu_s-(1+\ii\delta)\nu_m\right|^2 -
    \left|\nu_s+(1+\ii\delta)\nu_m\right|^2\ee^{-2\gamma\nu_s'}\right|.
\end{equation}
In the next lemma, we provide a lower bound on the first term in
\eqref{eqn:g_reverse_triangle}.
\begin{lemma}\label{lem:g_first_term_lower_bound}
    Suppose $\gamma > \gamma_*$ and $\delta > 0$; then 
    \begin{equation}\label{eqn:a_lower_bound}
        \left|\nu_s-(1+\ii\delta)\nu_m\right|^2 \ge
        \begin{cases}
            1 &\text{for } 0 \le p \le 1,\\
            \left(\sqrt{p^2+1}-\sqrt{p^2-1}\right)^2
                &\text{for } 1 \le p < \infty.
        \end{cases}
    \end{equation}
\end{lemma}
\begin{proof}
    We have
    \begin{equation}\label{eqn:first_term_lower_int}
        \left|\nu_s-(1+\ii\delta)\nu_m\right|^2
        \ge \left|\real\left[\nu_s-(1+\ii\delta)\nu_m\right]
            \right|^2.
    \end{equation}
    If $0 \le p \le 1$, then 
    \begin{equation}\label{eqn:first_term_lower_int_2}
        \left|\real\left[\nu_s-(1+\ii\delta)\nu_m\right]\right|
        = \sqrt{\dfrac{(p^2+1) + \sqrt{(p^2+1)^2 +
            \delta^2}}{2}} + \delta\sqrt{1-p^2}
        \ge 1.
    \end{equation}
    If $1 \le p < \infty$, then
    \begin{equation}\label{eqn:first_term_lower_int_3}
        \left|\real\left[\nu_s-(1+\ii\delta)\nu_m\right]\right|
        = \sqrt{\dfrac{(p^2+1) + \sqrt{(p^2+1)^2 +
            \delta^2}}{2}} - \sqrt{p^2-1}
            \ge \sqrt{p^2+1}-\sqrt{p^2-1}.
    \end{equation}
    Squaring both sides of \eqref{eqn:first_term_lower_int_2} and
    \eqref{eqn:first_term_lower_int_3} and utilizing
    \eqref{eqn:first_term_lower_int} gives us the desired result.
\end{proof}
In the next lemma, we provide upper bounds on the second term in
\eqref{eqn:g_reverse_triangle}.  
\begin{lemma}\label{lem:g_second_term_upper_bound}
    Suppose $\gamma > \gamma_*$;
    then there is a constant $0 < C < 1$ such that 
    \begin{equation}\label{eqn:b_upper_bound}
        \left|\nu_s+(1+\ii\delta)\nu_m\right|^2\ee^{-2\gamma\nu_s'} \le
        \begin{cases}
            C &\text{for } 0 \le p \le 1, \ 0 < \delta \le 0.4, \\
            \left(1+3\sqrt{\delta}\right)
            \left(\sqrt{p^2+1}+\sqrt{p^2-1}\right)^2
            \ee^{-2\gamma\sqrt{p^2+1}} &\text{for } 1 \le p < \infty, \ 
            0 < \delta \le 1.
        \end{cases}
    \end{equation}
\end{lemma}
\begin{proof}
    The triangle inequality, the assumptions that $\gamma >
    \gamma_*$ and $\delta > 0$, and the bound
    \begin{equation}\label{eqn:shell_equation_hat_prime_bound}
        \nu_s' = \sqrt{\dfrac{(p^2+1) + \sqrt{(p^2+1)^2 +
        \delta^2}}{2}} \ge \sqrt{p^2+1}
    \end{equation}
    imply that
    \[
        \left|\nu_s+(1+\ii\delta)\nu_m\right|\ee^{-\gamma\nu_s'}
        \le \left(\left|\nu_s\right| +
        \sqrt{1+\delta^2}|\nu_m|\right)\ee^{-\gamma_*\sqrt{p^2+1}}.
    \]
    For $0 \le p \le 1$ and $0 < \delta \le 0.4$, the definitions of
    $\nu_s$ and $\nu_m$ from \eqref{eqn:shell_equation_hat} and \eqref{eqn:matrix_equation_hat},
    respectively, imply that the above
    bound becomes
    \begin{align*}
        \left|\nu_s+(1+\ii\delta)\nu_m\right|\ee^{-\gamma\nu_s'}
        &\le \left\{\left[(p^2+1)^2 + \delta^2\right]^{1/4} +
        \sqrt{1+\delta^2}\sqrt{1-p^2}\right\}\ee^{-\gamma_*\sqrt{p^2+1}}
        \\
        &\le \left[\left(4+\delta^2\right)^{1/4} +
        \sqrt{1+\delta^2}\right]\ee^{-\gamma_*} \\
        &\le \left[\left(4+0.4^2\right)^{1/4} +
            \sqrt{1+0.4^2}\right]\ee^{-\gamma_*} \\
        &\approx 0.9813.
    \end{align*}
    Therefore, if we take $C = 0.99$, we have
    \[
        \left|\nu_s+(1+\ii\delta)\nu_m\right|^2\ee^{-2\gamma\nu_s'} \le
        C.
    \]
    
    Now we consider the case $1 \le p < \infty$.  The triangle
    inequality, the bound in \eqref{eqn:shell_equation_hat_prime_bound},
    and the definitions of $\nu_s$ and $\nu_m$ for $p\ge 1$ (see
    \eqref{eqn:shell_equation_hat} and \eqref{eqn:matrix_equation_hat},
    respectively), imply that
    \begin{align*}
        \left|\nu_s + (1+\ii\delta)\nu_m\right|\ee^{-\gamma\nu_s'} 
        &\le \left(\left|\nu_s\right| +
            \sqrt{1+\delta^2}\left|\nu_m\right|\right)
            \ee^{-\gamma\sqrt{p^2+1}} \\
        &= \left\{\left[(p^2+1)^2+\delta^2\right]^{1/4} +
            \sqrt{1+\delta^2}\sqrt{p^2-1}\right\} 
            \ee^{-\gamma\sqrt{p^2+1}}.
    \end{align*}
    Applying the bound $(x+y)^r \le x^r + y^r$ for $r = 1/4$ and $r =
    1/2$ to the right-hand side of the above inequality implies
    \begin{equation}\label{eqn:b_upper_bound_big_int}
        \left|\nu_s + (1+\ii\delta)\nu_m\right|\ee^{-\gamma\nu_s'} 
        \le \left[\sqrt{p^2+1}+\sqrt{\delta} +
        (1+\delta)\sqrt{p^2-1}\right]\ee^{-\gamma\sqrt{p^2+1}}.
    \end{equation}
    Squaring the term in brackets on right-hand side of the above
    expression gives
    \begin{equation}\label{eqn:bound_with_Q}
        \left[\sqrt{p^2+1}+\sqrt{\delta} +
        (1+\delta)\sqrt{p^2-1}\right]^2 = 
        \left(\sqrt{p^2+1}+\sqrt{p^2+1}\right)^2 + Q_{\delta}(p),
    \end{equation}
    where
    \begin{align*}
        Q_{\delta}(p)
        &\equiv 2\left(\sqrt{p^2+1}+\sqrt{p^2-1}\right)
        \left(\sqrt{\delta} + \delta\sqrt{p^2-1}\right) +
        \left(\sqrt{\delta}+\delta\sqrt{p^2-1}\right)^2 \\
        &= \sqrt{\delta}\left[2\left(\sqrt{p^2+1}+\sqrt{p^2-1}\right)
        \left(1 + \sqrt{\delta}\sqrt{p^2-1}\right) +
        \sqrt{\delta}\left(1+\sqrt{\delta}\sqrt{p^2-1}\right)^2\right].
        \numberthis\label{eqn:Q_bound_int}
    \end{align*}
    For $p\ge 1$ and $0 < \delta \le 1$, we have the bound
    \[
        1+\sqrt{\delta}\sqrt{p^2-1} \le \sqrt{p^2+1}+\sqrt{p^2-1}.
    \]
    Using this bound in \eqref{eqn:Q_bound_int} gives
    \begin{equation}\label{eqn:Q_bound}
        Q_{\delta}(p) \le \sqrt{\delta}\left(2+\sqrt{\delta}\right) 
        \left(\sqrt{p^2+1}+\sqrt{p^2-1}\right)^2
        \le 3\sqrt{\delta}\left(\sqrt{p^2+1}+\sqrt{p^2-1}\right)^2.
    \end{equation}
    Combining \eqref{eqn:b_upper_bound_big_int},
    \eqref{eqn:bound_with_Q}, and \eqref{eqn:Q_bound} gives us the
    second bound
    in \eqref{eqn:b_upper_bound}.
\end{proof}
We are now ready to give lower bounds on $|g_{\delta}(p;\gamma)|$.
\begin{lemma}\label{lem:abs_g_lower_bound}
    Suppose $\gamma > \gamma_*$; then there exists a $\delta_{\gamma}$
    satisfying $0 < \delta_{\gamma} \le 0.4$ and a positive constant $C$
    such that, for $0 < \delta \le \delta_{\gamma}$,
    \begin{equation}\label{eqn:g_lower_bound}
        |g_{\delta}(p;\gamma)| \ge 
        \begin{cases}
            C &\text{for } 0 \le p \le 1,\\
            \dfrac{|g_0(p;\gamma)|}{2} &\text{for } 1 \le p < \infty.
        \end{cases}
    \end{equation}
\end{lemma}
\begin{proof}
    If $0 \le p \le 1$ and $0 < \delta \le 0.4$, then
    \eqref{eqn:g_reverse_triangle},
    Lemma~\ref{lem:g_first_term_lower_bound}, and
    Lemma~\ref{lem:g_second_term_upper_bound} imply that there is a
    constant $0 < C < 1$ such that
    \[
        |g_{\delta}(p;\gamma)| \ge 
        \left|\nu_s-(1+\ii\delta)\nu_m\right|^2 -
        \left|\nu_s+(1+\ii\delta)\nu_m\right|^2
        \ee^{-2\gamma\nu_s'}
        \ge 1 - C > 0.
    \]
    This gives us first part of \eqref{eqn:g_lower_bound}.

    We now assume $1 \le p < \infty$.  First, note that $g_0(p;\gamma) >
    0$ for all $p \ge 1$ because it has no zeros by
    Lemma~\ref{lem:g_0_roots} and $g_0(1;\gamma) > 0$.  Then
    \eqref{eqn:g_0},
    \eqref{eqn:g_reverse_triangle},
    Lemma~\ref{lem:g_first_term_lower_bound}, and
    Lemma~\ref{lem:g_second_term_upper_bound} imply that
    \begin{align*}
        |g_{\delta}(p;\gamma)| 
        &\ge \left(\sqrt{p^2+1}-\sqrt{p^2-1}\right)^2 - 
            \left(1+3\sqrt{\delta}\right)
            \left(\sqrt{p^2+1}+\sqrt{p^2-1}\right)^2
            \ee^{-2\gamma\sqrt{p^2+1}}\\
        &= |g_0(p;\gamma)| - 
            3\sqrt{\delta}\left(\sqrt{p^2+1}+\sqrt{p^2-1}\right)^2
            \ee^{-2\gamma\sqrt{p^2+1}}
    \end{align*}
    for all $0 < \delta \le 1$.  Therefore,
    \begin{align}
        |g_{\delta}| - \frac{|g_0|}{2}
        &\ge \frac{|g_0|}{2} - 
            3\sqrt{\delta}\left(\sqrt{p^2+1}+\sqrt{p^2-1}\right)^2
            \ee^{-2\gamma\sqrt{p^2+1}} \label{eqn:small_p} \\
        &= \frac{1}{2}\left(\sqrt{p^2+1}-\sqrt{p^2-1}\right)^2 -
            \left(\frac{1}{2}+3\sqrt{\delta}\right)
            \left(\sqrt{p^2+1}+\sqrt{p^2-1}\right)^2
            \ee^{-2\gamma\sqrt{p^2+1}} \label{eqn:large_p}
    \end{align}
    for all $0 < \delta \le 1$.  From \eqref{eqn:large_p}, for all $0 <
    \delta \le 0.4$ we have
    \begin{equation}\label{eqn:large_p_int}
        |g_{\delta}|-\frac{|g_0|}{2} \ge 
        \frac{1}{2}\left(\sqrt{p^2+1}-\sqrt{p^2-1}\right)^2 -
            \left(\frac{1}{2}+3\sqrt{0.4}\right)
            \left(\sqrt{p^2+1}+\sqrt{p^2-1}\right)^2
            \ee^{-2\gamma\sqrt{p^2+1}}.
    \end{equation}
    Thanks to the exponential decay in the second term on the right-hand
    side of \eqref{eqn:large_p_int}, there exists a
    $\widetilde{p}_{\gamma} \ge 1$ such that the expression on the
    right-hand side of \eqref{eqn:large_p_int} is strictly positive for
    all $p \ge \widetilde{p}_{\gamma}$.  Therefore,
    \begin{equation}\label{eqn:large_p_final_1}
        |g_{\delta}(p;\gamma)| \ge 
        \frac{|g_0(p;\gamma)|}{2} \eqtext{for\ all}
        p \ge \widetilde{p}_{\gamma} \eqtext{and\ all} 0 < \delta \le
        0.4.
    \end{equation}
    
    Finally, we consider $1 \le p \le \widetilde{p}_{\gamma}$.  Because
    both terms on the
    right-hand side of
    \eqref{eqn:small_p} are continuous, we may define
    \[
        m_{\gamma} \equiv \min_{1\le p \le \widetilde{p}_{\gamma}}
            \frac{|g_0(p;\gamma)|}{2}
        \eqtext{and}
        M_{\gamma} \equiv \max_{1\le p \le \widetilde{p}_{\gamma}}
            \left[3\left(\sqrt{p^2+1}+\sqrt{p^2-1}\right)^2
            \ee^{-2\gamma\sqrt{p^2+1}}\right];
    \]
    because $\gamma > \gamma_*$, $m_{\gamma} > 0$ by
    Lemma~\ref{lem:g_0_roots}.  Hence \eqref{eqn:small_p} becomes
    \[
        |g_{\delta}(p;\gamma)|-\frac{|g_0(p;\gamma)|}{2} \ge  
        m_{\gamma} - \sqrt{\delta}M_{\gamma},
    \]
    which is nonnegative if we take 
    \[
        \delta \le \left(\frac{m_{\gamma}}{M_{\gamma}}\right)^2.
    \]
    Therefore
    \begin{equation}\label{eqn:large_p_final_2}
        |g_{\delta}(p;\gamma)| \ge 
        \frac{|g_0(p;\gamma)|}{2} \eqtext{for\ all}
        1 \le p \le \widetilde{p}_{\gamma} \eqtext{and\ all} 
        0 < \delta \le \left(\frac{m_{\gamma}}{M_{\gamma}}\right)^2.
    \end{equation}
    We define $\delta_{\gamma} \equiv \min\{0.4,
    m_{\gamma}^2/M_{\gamma}^2\}$; then \eqref{eqn:large_p_final_1} and
    \eqref{eqn:large_p_final_2} imply that
    \[
        |g_{\delta}(p;\gamma)| \ge \frac{|g_0(p;\gamma)|}{2} 
        \eqtext{for\ all} 1 \le p < \infty \eqtext{and} 
        0 < \delta \le \delta_{\gamma}.
    \]
    This completes the proof.
\end{proof}

Recalling that our ultimate goal is to prove Theorem~\ref{thm:bounded},
in the next lemma we derive upper bounds on $M_{\delta}(p;\gamma)$,
defined in \eqref{eqn:M}.
\begin{lemma}\label{lem:M_upper_bound}
    Suppose $\gamma > \gamma_*$ and $0 < \delta \le \delta_{\gamma}$,
    where $\delta_{\gamma}$ is defined in
    Lemma~\ref{lem:abs_g_lower_bound}.  Then there exists a constant
    $C_{\gamma} > 0$ such that
    \begin{equation}\label{eqn:M_bound}
        \left|M_{\delta}(p;\gamma)\right| \le C_{\gamma}
        \left\{\left[(p^2+1)^2 + 1\right]^{1/4} +
        \sqrt{2}\sqrt{\left|p^2-1\right|}\right\}^2
        \left[\sqrt{(p^2+1)^2 + 1} + p^2 \right]
    \end{equation}
    for all $0 \le p < \infty$ and all $0 < \delta \le \delta_{\gamma}$.
\end{lemma}
\begin{proof}
    By the triangle inequality, we may derive upper bounds on each term
    of $|M_{\delta}(p;\gamma)|$ individually.
    \begin{enumerate}
        \item For the term outside of the braces in \eqref{eqn:M}, we 
            have
            \begin{align*}
                \left|\nu_s-(1+\ii\delta)\nu_m\right|
                &\le |\nu_s| + |1+\ii\delta||\nu_m| \\
                &= \left[(p^2+1)^2 + \delta^2\right]^{1/4} +
                \sqrt{1+\delta^2}\sqrt{|p^2-1|} \\
                &\le \left[(p^2+1)^2 + \delta_{\gamma}^2\right]^{1/4} +
                \sqrt{1+\delta_{\gamma}^2}\sqrt{\left|p^2-1\right|}.
            \end{align*}
            Therefore, for all $p \ge 0$ and all $0 < \delta \le
            \delta_{\gamma} < 1$, we have
            \begin{align*}
                \left|\nu_s-(1+\ii\delta)\nu_m\right|^2 
                &\le
                \left\{\left[(p^2+1)^2 + \delta_{\gamma}^2\right]^{1/4}
                + \sqrt{1+\delta_{\gamma}^2}
                \sqrt{\left|p^2-1\right|}\right\}^2 \\
                &\le 
                \left\{\left[(p^2+1)^2 + 1\right]^{1/4}
                + \sqrt{2}\sqrt{\left|p^2-1\right|}\right\}^2.
                \numberthis\label{eqn:M_zeroth_term}
            \end{align*}
        \item For the first term in brackets in \eqref{eqn:M}, we have,
            thanks to \eqref{eqn:shell_equation_hat_prime_bound},
            \begin{align*}
                \left|\frac{1-\ee^{-2\gamma\nu_s'}}{2\nu_s'}\right|
                &= \frac{1-\ee^{-2\gamma\nu_s'}}{2\nu_s'} \\
                &\le \frac{1-\ee^{-2\gamma\sqrt{p^2+1}}}{2\sqrt{p^2+1}}.
            \end{align*}
            Because the above function is continuous for $p \in
            [0,\infty)$ and tends to $0$ as $p\rightarrow \infty$, it
            attains its maximum value on $[0,\infty)$.  Thus there is a
            constant $C_{\gamma} > 0$ such that 
            \begin{equation}\label{eqn:M_first_brackets}
                \left|\frac{1-\ee^{-2\gamma\nu_s'}}{2\nu_s'}\right| 
                \le C_{\gamma}.
            \end{equation}
        \item Using our result from item (2) and \eqref{eqn:alpha_R}, we
            find that the second term in brackets in \eqref{eqn:M}
            satisfies
            \begin{equation*}
                |R|^2\ee^{-2\gamma\nu_s'}
                \left|\frac{1-\ee^{-2\gamma\nu_s'}}{2\nu_s'}\right| 
                \le C_{\gamma}
                \frac{\left|\nu_s+(1+\ii\delta)\nu_m\right|^2
                \ee^{-2\gamma\nu_s'}}
                {\left|\nu_s-(1+\ii\delta)\nu_m\right|^2},
            \end{equation*}
            where $C_{\gamma}$ is the constant from
            \eqref{eqn:M_first_brackets}.  Applying the bounds from
            \eqref{eqn:a_lower_bound} and \eqref{eqn:b_upper_bound} as
            well as the bounds $\gamma > \gamma_*$ and $\delta \le
            \delta_{\gamma} < 1$ to the above expression gives
            \begin{equation}\label{eqn:R_int}
                |R|^2\ee^{-2\gamma\nu_s'}
                \left|\frac{1-\ee^{-2\gamma\nu_s'}}{2\nu_s'}\right| 
                \le  
                \begin{cases}
                    C_{\gamma} &\text{for } 0 \le p \le 1,\\
                    C_{\gamma}\dfrac{4
                    \left(\sqrt{p^2+1}+\sqrt{p^2-1}\right)^2
                    \ee^{-2\gamma_*\sqrt{p^2+1}}} 
                    {\left(\sqrt{p^2+1}-\sqrt{p^2-1}\right)^2}
                    &\text{for } 1 \le p < \infty.
                \end{cases}
            \end{equation}
            The function on the right-hand side of the above inequality
            is continuous as a function of $p$ for $p \in [1,\infty)$
            and decays to $0$ as $p\rightarrow \infty$.  Thus it
            attains its maximum value on $[1,\infty)$ (this maximum
            value is \emph{independent} of $\gamma$); this and
            \eqref{eqn:R_int} imply that there is a 
            constant $C_{\gamma} > 0$ such that
            \begin{equation}\label{eqn:M_second_brackets}
                |R|^2\ee^{-2\gamma\nu_s'}
                \left|\frac{1-\ee^{-2\gamma\nu_s'}}{2\nu_s'}\right| 
                \le C_{\gamma} 
            \end{equation}
            for all $p \ge 0$ and all $0 < \delta \le \delta_{\gamma}$.
        \item For the last term in \eqref{eqn:M}, we have
            \begin{equation}\label{eqn:M_last_int}
                \left|\ee^{-2\gamma\nu_s'}
                \imag\left[\overline{R}\ee^{2\ii\gamma\nu_s''}
                \left(\frac{1-\ee^{-2\ii\gamma\nu_s''}}{2\nu_s''}
                \right)\right]\right| \le 
                \ee^{-2\gamma\nu_s'}|R|
                \left|\frac{1-\ee^{-2\ii\gamma\nu_s''}}{2\nu_s''}
                \right|.
            \end{equation}
            Arguments similar to those in item (3) can be used to show
            that there is a positive constant $C$ such that
            \begin{equation}\label{eqn:M_last_int_1}
                |R|\ee^{-\gamma\nu_s'} \le C.
            \end{equation}
            Because 
            \[
                \nu_s'' = \frac{\delta}{2\nu_s'},
            \] 
            the function
            \[
                \ee^{-\gamma\nu_s'}
                \left|\frac{1-\ee^{-2\gamma\ii\nu_s''}}{2\nu_s''}\right|
            \]
            is continuous for $0 \le \delta \le \delta_{\gamma}$ and $0
            \le p < \infty$ (after modification at $\delta = 0$).
            Moreover, this function goes to $0$ as $p\rightarrow
            \infty$, so it attains its maximum value.  This implies that
            there is a constant $C_{\gamma} > 0$ such that
            \begin{equation}\label{eqn:M_last_int_2}
                \ee^{-\gamma\nu_s'}
                \left|\frac{1-\ee^{-2\gamma\ii\nu_s''}}{2\nu_s''}\right|
                \le C_{\gamma}.
            \end{equation}
            Inserting \eqref{eqn:M_last_int_1} and
            \eqref{eqn:M_last_int_2} into \eqref{eqn:M_last_int} implies
            that there is a constant $C_{\gamma} > 0$ such that
            \begin{equation}\label{eqn:M_last}
                \left|\ee^{-2\gamma\nu_s'}
                \imag\left[\overline{R}\ee^{2\ii\gamma\nu_s''}
                \left(\frac{1-\ee^{-2\ii\gamma\nu_s''}}{2\nu_s''}
                \right)\right]\right| \le 
                \ee^{-2\gamma\nu_s'}|R|
                \left|\frac{1-\ee^{-2\ii\gamma\nu_s''}}{2\nu_s''}
                \right| \le C_{\gamma}
            \end{equation}
            for all $0 \le p < \infty$ and all $0 < \delta \le
            \delta_{\gamma}$.
    \end{enumerate}
    Using the bound 
    \[
        \left|\pm\left|\nu_s\right|^2 + p^2\right| \le
        \left|\nu_s\right|^2 + p^2 \le \sqrt{(p^2+1)^2 + 1} +
        p^2,
    \]
    which holds for all $p \ge 1$ and all $\delta \le \delta_{\gamma} <
    1$, as well as the bounds \eqref{eqn:M_zeroth_term},
    \eqref{eqn:M_first_brackets}, \eqref{eqn:M_second_brackets}, and
    \eqref{eqn:M_last} in \eqref{eqn:M} gives 
    \[
        \left|M_{\delta}(p;\gamma)\right| \le C_{\gamma}
        \left\{\left[(p^2+1)^2 + 1\right]^{1/4} +
                \sqrt{2}\sqrt{\left|p^2-1\right|}\right\}^2
        \left[\sqrt{(p^2+1)^2 + 1} + p^2 \right],
    \]
    for some positive constant $C_{\gamma}$.
\end{proof}

We are finally ready to complete the proof of Theorem~\ref{thm:bounded}.
First, we split the integral in \eqref{eqn:pd_a} to obtain
\[
    \int_0^{\infty} L_{\delta}(p;\gamma)\di{p} = 
    \int_0^1 L_{\delta}(p;\gamma)\di{p} + 
    \int_1^{\infty} L_{\delta}(p;\gamma)\di{p}.
\]
We will focus on each integral separately.  For $0 \le p \le 1$, we use
the bounds from Lemma~\ref{lem:abs_g_lower_bound}, \eqref{eqn:M_bound},
and \eqref{eqn:I_p_bounds} in \eqref{eqn:L} to obtain the bound
\begin{equation}\label{eqn:L_small_p_bound}
    L_{\delta}(p;\gamma) \le 
        C_{\gamma}\left\{\left[(p^2+1)^2 + 1\right]^{1/4} +
        \sqrt{2}\sqrt{\left|p^2-1\right|}\right\}^2
        \left[\sqrt{(p^2+1)^2 + 1} + p^2 \right],
\end{equation}
which holds for all $0 \le p \le 1$ and $0 < \delta \le
\delta_{\gamma}$.  The function on the right-hand
side of the above inequality is continuous for $p \in [0,1]$, so it
attains its maximum value. Thus there
is a positive constant $C_{\gamma}$ such that
\begin{equation}\label{eqn:L_int_small_p_bound}
    \int_0^1 L_{\delta}(p;\gamma) \le C_{\gamma} 
    \eqtext{for\ all} 0 \le p \le 1
    \eqtext{and\ all} 0 < \delta \le \delta_{\gamma}.
\end{equation}

Similarly, for $p \ge 1$, the bounds from
Lemmas~\ref{lem:abs_g_lower_bound}, \eqref{eqn:M_bound}, and
\eqref{eqn:I_p_bounds} imply that there is a positive constant
$C_{\gamma}$ such that 
\begin{equation}\label{eqn:L_large_p_bound}
    L_{\delta}(p;\gamma) \le 
        C_{\gamma}\dfrac{\ee^{-2\gamma\left(\frac{d_0}{a}-1\right)
        \sqrt{p^2-1}}}
        {|g_0(p;\gamma)|^2}\left\{\left[(p^2+1)^2 + 1\right]^{1/4} +
        \sqrt{2}\sqrt{\left|p^2-1\right|}\right\}^2
        \left[\sqrt{(p^2+1)^2 + 1} + p^2 \right];
\end{equation}
this bound holds for all $1 \le p < \infty$ and all $0 < \delta \le
\delta_{\gamma}$.  Note from \eqref{eqn:g_0} that $|g_0| \rightarrow 0$
as $p\rightarrow \infty$ at an algebraic (i.e., nonexponential) rate.
Because $d_0/a > 1$, the exponential term in the numerator in
\eqref{eqn:L_large_p_bound} defeats the nonexponential terms in the
braces and in $|g_0|^2$; in other words, there is a positive constant
$C_{\gamma} > 0$ such that 
\begin{equation}\label{eqn:L_int_large_p_bound}
    \int_1^{\infty} L_{\delta}(p;\gamma) \le C_{\gamma} 
    \eqtext{for\ all} 1 \le p < \infty
    \eqtext{and\ all} 0 < \delta \le \delta_{\gamma}.
\end{equation}
Using \eqref{eqn:L_int_small_p_bound} and
\eqref{eqn:L_int_large_p_bound} in \eqref{eqn:pd_a} gives
\[
    E_{\delta}(a) \le C_{\gamma} \eqtext{for \ all} 0 < \delta \le
    \delta_{\gamma}.
\]
This completes the proof of the theorem.


\subsection{Proof of Theorem~\ref{thm:shielding}}
\label{subsec:proof_of_thm_shielding}

We begin by taking care of an important technicality.  Recall from
\S~\ref{subsec:proof_of_thm_bounded} that $\widetilde{p}_{\gamma}$ was
chosen so that the expression on the right-hand side of
\eqref{eqn:large_p_int} is strictly positive for $p \ge
\widetilde{p}_{\gamma}$.  Because of this choice, $\delta_{\gamma}$
depends on $\gamma$ in a nontrivial way --- see
\eqref{eqn:large_p_final_2}.  However, there is a $\widetilde{\gamma} >
0$ such that, if $\gamma \ge \widetilde{\gamma}$, the right-hand side of
\eqref{eqn:large_p_int} is positive for all $p \ge 1$ (thanks to the
exponential decay in $\gamma$ of the last term in
\eqref{eqn:large_p_int}).  Thus, for $\gamma \ge \widetilde{\gamma}$,
\eqref{eqn:large_p_int} immediately implies that there is a $\delta_0$,
\emph{independent} of $\gamma$, such that $0 < \delta_0 \le 0.4$ and
\[
    |g_{\delta}(p;\gamma)| \ge \frac{g_0(p;\gamma)}{2}
\]
for all $p \ge 1$, all $0 < \delta \le \delta_0$, and all $\gamma \ge
\widetilde{\gamma}$.  To visualize this, in Figure~\ref{fig:large_p_int}
we plot the expression on the right-hand side of \eqref{eqn:large_p_int}
as a function of $p$ for different values of $\gamma$.
Figure~\ref{fig:large_p_int}(a) is a plot over the interval $1 \le p \le
7$ while Figure~\ref{fig:large_p_int}(b) is a plot over the interval $7
\le p \le 10$; in addition the blue, dashed curve is for $\gamma =
\gamma_*$, the red, dotted curve is for $\gamma = \frac{3}{2}\gamma_*$,
and the yellow, solid curve is for $\gamma = 2\gamma_*$ (the $\gamma =
\frac{3}{2}\gamma_*$ and $\gamma = 2\gamma_*$ curves overlap in
Figure~\ref{fig:large_p_int}(b).  In Figure~\ref{fig:large_p_int}(a), we
note that the expression is negative for some values of $p$ if $\gamma =
\gamma_*$; however, if $\gamma = \frac{3}{2}\gamma_*$ or $\gamma =
2\gamma_*$, the curve is always positive.  Therefore, for the remainder
of this proof, we will assume $0 < \delta \le \delta_0$ and $\gamma \ge
2\gamma_*$.
\begin{figure}[!hbp]
    \begin{center}
        \begin{tabular}{c c}
            \includegraphics[width=0.45\textwidth]{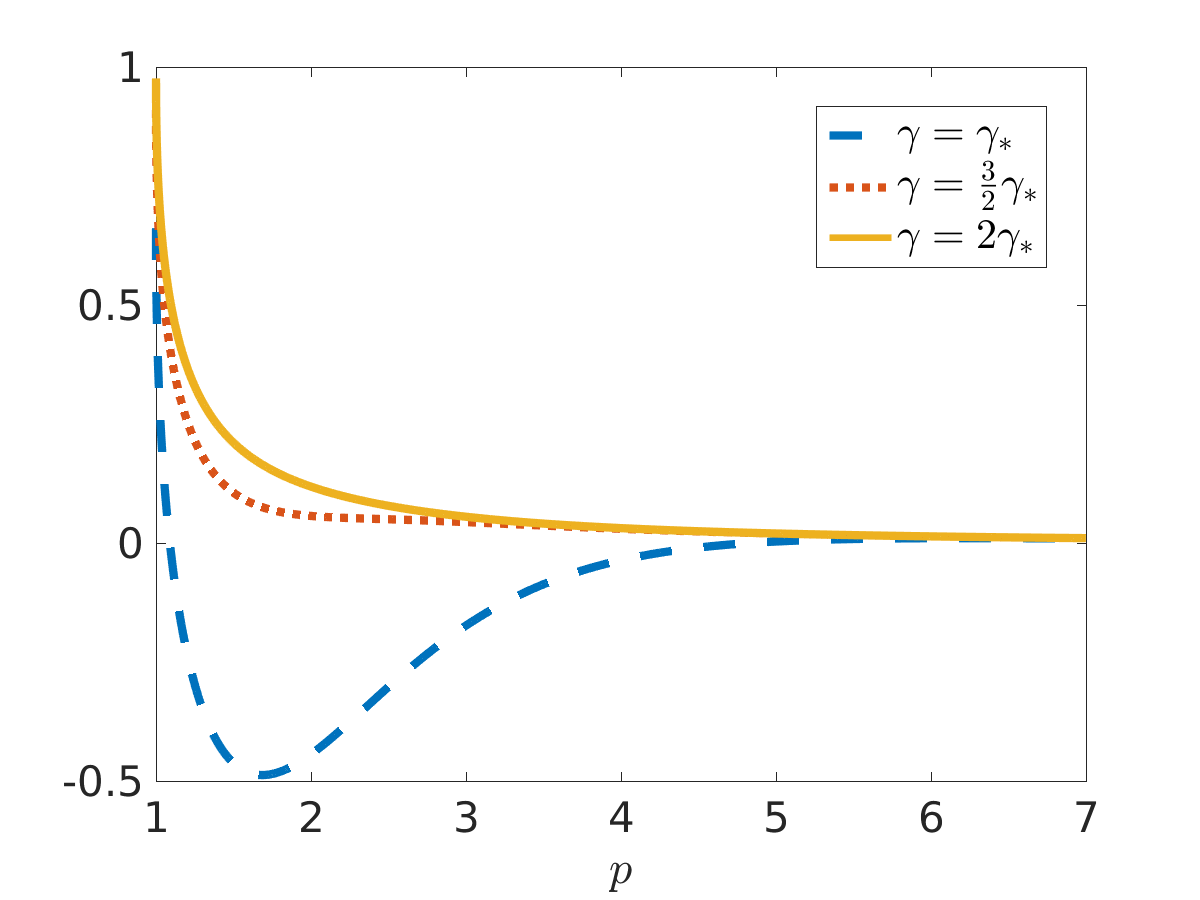} &
            \includegraphics[width=0.45\textwidth]{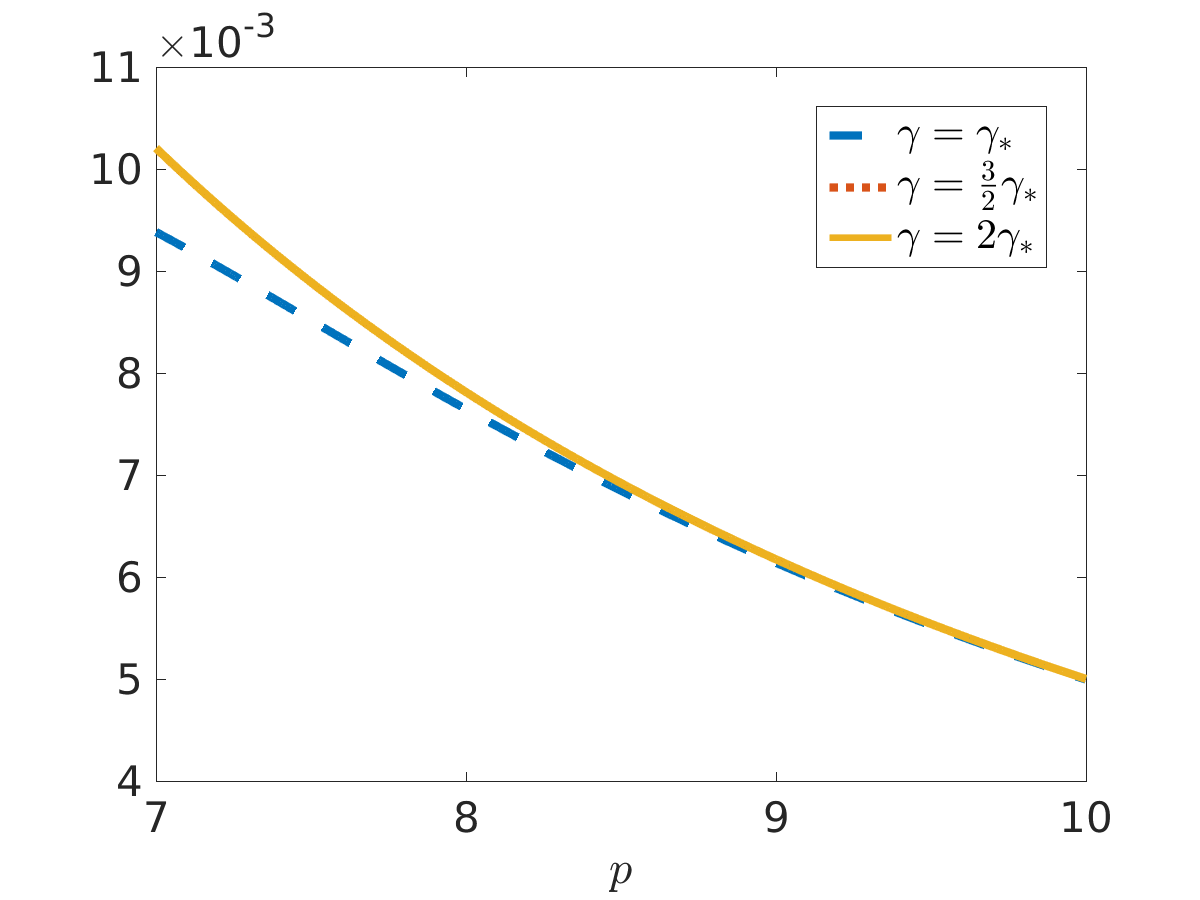}
            \\
            (a) & (b)
        \end{tabular}
        \caption{\emph{This figure contains a plot of the expression on
                the right-hand side of \eqref{eqn:large_p_int}.  In both
                plots, the blue, dashed curve corresponds to $\gamma =
                \gamma_*$, the red, dotted curve corresponds to $\gamma
                = \frac{3}{2}\gamma_*$, and the yellow, solid curve
                corresponds to $\gamma = 2\gamma_*$.  In particular, we
                have plotted the expression in \eqref{eqn:large_p_int}
                over the intervals (a) $1 \le p \le 7$ and (b) $7 \le p
                \le 10$.}}
        \label{fig:large_p_int}
    \end{center}
\end{figure}

The Fourier inversion theorem, the triangle inequality, and
\eqref{eqn:Vc_hat} imply that
\[
    \left|V_c(x,y)\right| = \left|\frac{1}{2\pi}\int_{-\infty}^{\infty}
    A_q \ee^{k_0\nu_c x} \ee^{\ii q y} \di{q}\right| \le 
    \frac{1}{2\pi} \int_{-\infty}^{\infty}
    \left|A_q\right|\ee^{k_0\nu_c' x} \di{q} = 
    \frac{1}{\pi}\int_{0}^{\infty}
    \left|A_q\right|\ee^{k_0\nu_c' x} \di{q}.
\]

Using \eqref{eqn:psi}--\eqref{eqn:Aq} and \eqref{eqn:g} and making the
change of variables $p = q/k_0$ in the above integral gives
\begin{equation}\label{eqn:Vc_bound}
    |V_c(x,y)| \le \frac{2|1+\ii\delta|}{\pi}\int_0^{\infty}
    \frac{|I_p||\nu_s| \ee^{\gamma\nu_m'}
    \ee^{-\gamma\nu_s'}\ee^{k_0\nu_c'x}} {|g_{\delta}(p;\gamma)|} 
    \di{p}.
\end{equation}
\emph{Case 1:} $0 \le p \le 1$.  In this case, the integral in
\eqref{eqn:Vc_bound} (restricted to $0 \le p \le 1$) is
\[
    \int_0^1 \frac{|I_p||\nu_s|\ee^{-\gamma\nu_s'}}
    {|g_{\delta}(p;\gamma)|} \di{p}.
\]
For $\gamma \ge 2\gamma_*$ and $0 < \delta \le \delta_0 \le 1$,
\eqref{eqn:shell_equation_hat_prime_bound} and Lemmas~\ref{lem:abs_g_lower_bound} and
\ref{lem:I_p_upper_bound} imply that there is a constant $C > 0$ such
that the above integral is less than or equal to
\begin{align*}
    C\int_0^1 |\nu_s|\ee^{-\gamma\nu_s'}\di{p} 
    &\le
    C\int_0^1
    \left[(p^2+1)^2+1\right]^{1/4}\ee^{-\gamma\sqrt{p^2+1}}\di{p} \\
    &= C\int_0^1  
    \left[(p^2+1)^2+1\right]^{1/4}
    \ee^{-\eta\gamma\sqrt{p^2+1}}
    \ee^{-(1-\eta)\gamma\sqrt{p^2+1}}\di{p} \\
    &\le C\ee^{-\eta\gamma} 
    \int_0^1  
    \left[(p^2+1)^2+1\right]^{1/4}
    \ee^{-(1-\eta)2\gamma*\sqrt{p^2+1}}\di{p}.
    \numberthis\label{eqn:shielding_small_int}
\end{align*}
Because the integrand in \eqref{eqn:shielding_small_int} is continuous,
it is bounded above by a constant (independent of $p$, $\delta$, and
$\gamma$).  Thus \eqref{eqn:shielding_small_int} implies that there is
$C_{\eta} > 0$ such that
\begin{equation}\label{eqn:shielding_small_p}
    \int_0^1 \frac{|I_p||\nu_s|\ee^{-\gamma\nu_s'}}{|g|} \di{p} 
    \le C_{\eta}\ee^{-\eta\gamma}.
\end{equation}

\emph{Case 2:} $1 \le p < \infty$.  In this case, the integral in
\eqref{eqn:Vc_bound} (restricted to $1 \le p < \infty$) is
\[
    \int_1^{\infty}
    \frac{|I_p||\nu_s| \ee^{\gamma\sqrt{p^2-1}}
    \ee^{-\gamma\nu_s'}\ee^{k_0\sqrt{p^2-1}x}} {|g|} \di{p}.
\]
For $\gamma \ge 2\gamma_*$ and $0 < \delta \le \delta_0 \le 1$,
\eqref{eqn:shell_equation_hat_prime_bound} and
Lemmas~\ref{lem:abs_g_lower_bound} and \ref{lem:I_p_upper_bound} imply
that there is a constant $C > 0$ such that the above integral is less
than or equal to
\[
    C\int_1^{\infty}\frac{\left[(p^2+1)^2+1\right]^{1/4}
        \ee^{-\gamma\left(\frac{d_0}{a}-1\right)
        \sqrt{p^2-1}}\ee^{-\gamma\sqrt{p^2+1}}\ee^{k_0\sqrt{p^2-1}x}}
        {|g_0(p;\gamma)|}\di{p}.
\]
From \eqref{eqn:g_0}, the denominator of the above integrand is an
increasing function of $\gamma$.  Together with the fact that all of the
exponential terms in the numerator attain their maximum values at $p = 1$,
this implies, for $\gamma \ge 2\gamma_*$, that the
above integral is bounded above by
\begin{align*}
    C\int_1^{\infty}\frac{\left[(p^2+1)^2+\delta^2\right]^{1/4}
        \ee^{-\gamma\sqrt{p^2+1}}}
        {|g_0(p;2\gamma_*)|}\di{p} 
    &\le C
        \int_1^{\infty}\frac{\left[(p^2+1)^2+1\right]^{1/4}
        \ee^{-\gamma\sqrt{p^2+1}}}
        {|g_0(p;2\gamma_*)|}\di{p} \\
    &= C\int_1^{\infty}\frac{\left[(p^2+1)^2+1\right]^{1/4}
        \ee^{-\eta\gamma\sqrt{p^2+1}}
        \ee^{-(1-\eta)\gamma\sqrt{p^2+1}}}
        {|g_0(p;2\gamma_*)|}\di{p} \\
    &\le C\ee^{-\sqrt{2}\eta\gamma}
        \int_1^{\infty}\frac{\left[(p^2+1)^2+1\right]^{1/4}
        \ee^{-(1-\eta)2\gamma_*\sqrt{p^2+1}}}
        {|g_0(p;2\gamma_*)|}\di{p}.
\end{align*}
Because $|g_0(p;2\gamma_*)|$ has no roots (by
Lemma~\ref{lem:g_0_roots}) and tends to $0$ as $p\rightarrow\infty$ only
algebraically, the above integral converges.  This implies that
\begin{equation}\label{eqn:shielding_large_p}
    \int_1^{\infty}
    \frac{|I_p||\nu_s| \ee^{\gamma\nu_m'}
    \ee^{-\gamma\nu_s'}\ee^{k_0\nu_c'x}} {|g|} \di{p} 
    \le C_{\eta}\ee^{-\sqrt{2}\eta\gamma}
\end{equation}
for some $C_{\eta} > 0$.  Using \eqref{eqn:shielding_small_p} and
\eqref{eqn:shielding_large_p} in \eqref{eqn:Vc_bound} gives 
\begin{equation}\label{eqn:V_c_bound_final}
    |V_c(x,y)| \le C_{\eta}\ee^{-\eta\gamma} = C_{\eta}\ee^{-\eta k_0a};
\end{equation}
this bound holds for all $0 < \eta < 1$, $\gamma \ge 2\gamma_*$, all $0
< \delta \le \delta_0$, and all $(x,y) \in \mathcal{C}$.  Thus $V_c(x,y)$
goes to $0$ exponentially as $k_0\rightarrow\infty$.  This completes the
proof of the theorem.


\subsection{Derivation of Helmholtz equation from Maxwell equations}
\label{subsec:Maxwell_to_Helmholtz}

When no charge source is present, the Maxwell equations are
\begin{equation}
    \left\{
    \begin{aligned}
        & \nabla\cdot\mathbf{D} = 0, \quad && \nabla\times\mathbf{E} =
            -\frac{\partial\mathbf{B}}{\partial t}, \\
        & \nabla\cdot\mathbf{B} = 0, \quad && \nabla\times\mathbf{H} =
        \frac{\partial\mathbf{D}}{\partial t} + \mathbf{J}.
    \end{aligned}
    \right.
\end{equation}
In linear media the relevant fields satisfy the constitutive relations
\begin{equation*}
    \mathbf{D} = \ve\mathbf{E} \eqtext{and} \mathbf{B} = \mu\mathbf{H}.
\end{equation*}
By taking the divergence of the Amp\'ere Law with
the Maxwell correction (the fourth equation) and utilizing the Gauss Law
(the first equation),
we find that the current $\mathbf{J}$ must satisfy the continuity
equation, namely
\begin{equation}\label{eqn:continuity}
    \nabla\cdot\mathbf{J} = 0.
\end{equation}
We then take the curl of the Amp\'ere Law with the Maxwell correction
and apply a vector identity to obtain
\[
    \nabla\left(\nabla\cdot\mathbf{H}\right) - \Delta\mathbf{H} =
    \frac{\partial \left(\nabla\times\mathbf{D}\right)}{\partial t} +
    \nabla \times \mathbf{J}.
\]
In an isotropic and homogeneous medium (where $\ve$ and $\mu$ are
constant scalars), the above equation becomes
\[
    \frac{1}{\mu}\nabla\left(\nabla\cdot\mathbf{B}\right) -
    \frac{1}{\mu}\Delta\mathbf{B} = 
    \ve\frac{\partial \left(\nabla\times\mathbf{E}\right)}{\partial t} +
    \nabla \times \mathbf{J}.
\]
Utilizing the Faraday Law in combination with the fact that $\mathbf{B}$
is divergence free, we obtain
\begin{equation}\label{eqn:general_B}
    -\frac{1}{\mu}\Delta\mathbf{B} =
    -\ve\frac{\partial^2\mathbf{B}}{\partial t^2} +
    \nabla\times\mathbf{J}.
\end{equation}
Finally, we assume that all fields have harmonic time-dependence of the
form $\ee^{\ii\omega t}$; in particular, we assume that $\mathbf{B} =
\widetilde{\mathbf{B}}(\mathbf{x})\ee^{\ii\omega t}$ and $\mathbf{J} =
\widetilde{\mathbf{J}}(\mathbf{x})\ee^{\ii\omega t}$.  We also define
\[
    k_0 = \frac{\omega}{c} = \omega\sqrt{\varepsilon_0\mu_0}.
\]
Then, thanks to \eqref{eqn:general_B}, we find that
$\widetilde{\mathbf{B}}$ satisfies
\begin{equation}\label{eqn:B_hat}
    \Delta\widetilde{\mathbf{B}} + k_0^2\ve_r\mu_r\widetilde{\mathbf{B}} =
    -\mu_0\mu_r\nabla\times\widetilde{\mathbf{J}}, 
\end{equation}
where $\ve_r = \ve/\ve_0$ and $\mu_r = \mu/\mu_0$.  Thus each component
of $\widetilde{\mathbf{B}}$ satisfies a 3D Helmholtz equation.


\subsubsection{2D Helmholtz equation}\label{subsubsec:2D_Helmholtz}

We now assume that the current source $\widetilde{\mathbf{J}}$ is a line
current of the form
\begin{equation}\label{eqn:special_J}
    \widetilde{\mathbf{J}} = \widetilde{J}_x(x,y) \mathbf{e}_x +
    \widetilde{J}_y(x,y)\mathbf{e}_y,
\end{equation}
where, because $\widetilde{\mathbf{J}}$ must satisfy
\eqref{eqn:continuity}, 
\begin{equation}\label{eqn:special_continuity}
    \frac{\partial \widetilde{J}_x}{\partial x} + \frac{\partial
    \widetilde{J}_y}{\partial y} = 0.
\end{equation}
By symmetry, none of the fields will depend on $z$.  In addition, we
assume that we are dealing with nonmagnetic materials for which $\mu_r =
1$.  Thus \eqref{eqn:B_hat} and \eqref{eqn:special_J}
imply that the $z$-component of $\widetilde{\mathbf{B}}$ satisfies
\begin{equation}\label{eqn:almost_B_z}
    \frac{\partial^2 \widetilde{B}_z}{\partial x^2} + \frac{\partial^2
    \widetilde{B}_z}{\partial y^2} + k_0^2\ve_r \widetilde{B}_z =
    -\mu_0\left(\frac{\partial \widetilde{J}_y}{\partial x} -
    \frac{\partial \widetilde{J}_x}{\partial y}\right).
\end{equation}
Finally, the Maxwell equations can be used to show that $\widetilde{B}_z$
and $\ve_r^{-1}\frac{\partial \widetilde{B}_z}{\partial x}$ must be
continuous across the boundaries of the slab at $x = 0$ and $x = a$
\cite{Griffiths:1999:ITE}.  Then \eqref{eqn:almost_B_z} can be written in
divergence form as
\begin{equation}\label{eqn:final_B_z}
    \nabla\cdot\left(\frac{1}{\ve_r}\nabla \widetilde{B}_z\right) +
    k_0^2\widetilde{B}_z = -f,
\end{equation}
where 
\begin{equation*}
    f \equiv \mu_0\left(\frac{\partial \widetilde{J}_y}{\partial x} -
    \frac{\partial \widetilde{J}_x}{\partial y}\right).
\end{equation*}

\section*{Conflicts of Interest}

The authors declare that there are no conflicts of interest regarding
the publication of this article.

\bibliographystyle{plain}
\bibliography{Helmholtz_paper}

\end{document}